\documentclass[
    aps,
    prx,
    reprint,
    superscriptaddress,
    longbibliography,
    nofootinbib,
    floatfix,
]{revtex4-2}

\usepackage[T1]{fontenc}
\usepackage[utf8]{inputenc}
\usepackage{microtype}

\usepackage{graphicx}

\usepackage{booktabs}
\usepackage{makecell}

\usepackage{adjustbox}
\usepackage{capt-of}

\usepackage{amsmath,amssymb,mathtools}
\usepackage{amsthm}
\usepackage{braket} 
\usepackage{bm}
\usepackage{siunitx}

\usepackage[colorlinks=true,citecolor=blue,linkcolor=blue,urlcolor=blue]{hyperref}
\usepackage[nameinlink,capitalize]{cleveref}

\crefname{figure}{Fig.}{Figs.}
\Crefname{figure}{Fig.}{Figs.}

\crefname{equation}{Eq.}{Eqs.}
\Crefname{equation}{Eq.}{Eqs.}

\crefname{section}{Sec.}{Secs.}
\Crefname{section}{Sec.}{Secs.}

\crefname{appendix}{App.}{Apps.}
\Crefname{appendix}{App.}{Apps.}

\crefname{table}{Table}{Tables}
\Crefname{table}{Table}{Tables}

\crefname{definition}{Definition}{Definitions}
\Crefname{definition}{Definition}{Definitions}

\crefname{theorem}{Theorem}{Theorems}
\Crefname{theorem}{Theorem}{Theorems}

\crefname{proposition}{Proposition}{Propositions}
\Crefname{proposition}{Proposition}{Propositions}

\crefname{lemma}{Lemma}{Lemmas}
\Crefname{lemma}{Lemma}{Lemmas}

\crefname{corollary}{Corollary}{Corollaries}
\Crefname{corollary}{Corollary}{Corollaries}

\newcommand{\Tr}{\mathrm{Tr}}
\newcommand{\inp}{\mathfrak{i}}
\newcommand{\out}{\mathfrak{o}}

\newcommand{\ketbra}[2]{\ket{#1}\!\!\bra{#2}}  
\newcommand{\dket}[1]{| #1 \rangle \! \rangle}  

\newcommand{\dketbra}[2]{| #1 \rangle \! \rangle \! \langle \! \langle #2 |}

\newtheorem{theorem}{Theorem}
\newtheorem{lemma}[theorem]{Lemma}
\newtheorem{proposition}[theorem]{Proposition}
\newtheorem{corollary}[theorem]{Corollary}
\newtheorem{definition}[theorem]{Definition}

\newcommand{\valignbox}{\adjustbox{valign=c,padding=0 0em}}

\begin{document}

\title{Spatiotemporal Pauli processes: Quantum combs for modeling correlated noise in quantum error correction}

\author{John F. Kam}
\email{john.kam@monash.edu}
\affiliation{School of Physics and Astronomy, Monash University, Clayton, VIC 3168, Australia}
\affiliation{Quantum Innovation Centre (Q.InC), Agency for Science, Technology and Research (A*STAR), 4 Fusionopolis Way, Kinesis, \#05, Singapore 138635, Republic of Singapore}
\affiliation{Quantum Systems, Data61, CSIRO, Clayton, VIC 3168, Australia}

\author{Angus Southwell}
\affiliation{School of Physics and Astronomy, Monash University, Clayton, VIC 3168, Australia}

\author{Spiro Gicev}
\affiliation{School of Physics, The University of Melbourne, Parkville, VIC 3052, Australia}

\author{Muhammad Usman}
\affiliation{Quantum Systems, Data61, CSIRO, Clayton, VIC 3168, Australia}
\affiliation{School of Physics, The University of Melbourne, Parkville, VIC 3052, Australia}

\author{Kavan Modi}
\affiliation{Science, Mathematics and Technology Cluster, Singapore University of Technology and Design, 8 Somapah Road, 487372 Singapore}
\affiliation{Centre for Quantum Technologies, National University of Singapore, 3 Science Drive 2, 117543}


\begin{abstract}
    Correlated noise is a critical failure mode in quantum error correction (QEC), yet a gap remains between the stochastic Pauli models used for scalable QEC analyses and the microscopic, non-Markovian descriptions of noise in physical devices. We bridge this gap by introducing \emph{Spatiotemporal Pauli Processes} (SPPs): the natural multi-time generalization of Pauli channels---joint probability distributions over Pauli faults across space and time---obtained exactly from any noise process, however non-Markovian, under standard Pauli-frame randomization. SPPs thereby provide a single language in which correlated noise can be recast, compared, and extended, compatible with both microscopic open systems modeling and the stabilizer workflow of QEC design and analysis. Our central result is constructive: the multi-time Pauli twirl acts by local contractions on a process tensor network, yielding an explicit classical tensor network whose virtual bonds encode memory, bounded by the environment's Liouville-space dimension. Transfer operator diagnostics link memory spectra to correlation decay, and hidden Markov representations enable efficient sampling of correlated Pauli fault trajectories---mapping microscopically derived noise directly into circuit-level QEC simulation. We demonstrate this with surface code memory and stability benchmarks up to distance \(19\). A temporal ``storm'' model with tunable correlation time shows that memory alone, at strictly fixed marginal error rates, systematically erodes the exponential error suppression expected from code distance. A genuinely spatiotemporal quantum cellular automaton bath maps exactly, under system twirling, to a nonlinear probabilistic cellular automaton; tuning its coherent interactions drives the noise into a pseudo-critical regime with critical slowing down and macroscopic error avalanches that reverse surface code distance scaling. These results establish that below-threshold average error rates alone do not guarantee scalable error suppression, and position SPPs as a practical bridge from microscopic correlated dynamics to circuit-level QEC.
\end{abstract}

\maketitle

\tableofcontents

\section{Introduction}\label{sec:introduction}
Correlations are increasingly recognized as decisive failure modes in quantum error correction (QEC)~\cite{harperLearningCorrelatedNoise2023,fkamDetrimentalNonMarkovianErrors2025}. As devices reach below-threshold \emph{average} error rates, logical failure becomes dominated by rare correlated events; the physics of real devices permits temporal memory and spatial structure that can concentrate faults into bursts~\cite{mcewenRemovingLeakageinducedCorrelated2021,mcewenResolvingCatastrophicError2022,googlequantumaiandcollaboratorsQuantumErrorCorrection2025,bratrudMeasurementCorrelatedCharge2025,harringtonSynchronousDetectionCosmic2025}, induce heavy tails in syndrome statistics~\cite{claderImpactCorrelationsHeavy2021}, and ultimately undermine the independence assumptions behind standard threshold theorems~\cite{terhalFaulttolerantQuantumComputation2005,aharonovFaultTolerantQuantumComputation2006,alickiInternalConsistencyFaulttolerant2006}. Yet the modeling tools on either side of the physics--QEC interface remain disconnected. On one hand, stochastic Pauli models are the effective language of QEC design and analysis---underpinning the scalable simulation~\cite{gidneyStimFastStabilizer2021}, decoding~\cite{higgottPyMatchingPythonPackage2022,higgottSparseBlossomCorrecting2025}, and threshold analyses~\cite{dennisTopologicalQuantumMemory2002,fowlerSurfaceCodesPractical2012} that have driven much recent progress---but they typically impose near-independence, with correlations postulated phenomenologically through specialized, model-specific extensions. On the other hand, microscopic open quantum systems theories capture coherent dynamics and non-Markovian memory~\cite{breuerTheoryOpenQuantum2007,breuerColloquiumNonMarkovianDynamics2016,milzQuantumStochasticProcesses2021}, but lack a direct interface with circuit-level benchmarking at scale.

A natural starting point for describing general open system dynamics is the quantum \emph{process tensor}~\cite{pollockNonMarkovianQuantumProcesses2018,costaQuantumCausalModelling2016,milzQuantumStochasticProcesses2021}. Conceptually, it is the multi-time generalization of a quantum channel, and is represented as a \emph{quantum comb}~\cite{chiribellaQuantumCircuitArchitecture2008,chiribellaTheoreticalFrameworkQuantum2009,oreshkovQuantumCorrelationsNo2012,whiteUnifyingNonMarkovianCharacterization2025,tarantoHigherOrderQuantumOperations2025}. It provides an operationally complete description of non-Markovian dynamics by encoding the joint statistics of the system's response under arbitrary interventions at multiple times. It thereby captures the full temporal correlation structure induced by an environment with quantum memory, and versions of this object have been characterized experimentally on quantum hardware~\cite{whiteDemonstrationNonMarkovianProcess2020,whiteNonMarkovianQuantumProcess2022,whiteManytimePhysicsPractice2024,zhangLearningForecastingOpen2025}.

For QEC applications~\cite{tanggaraStrategicCodeUnified2024,kobayashiTensornetworkDecodersProcess2024}, however, this generality exposes a clear tension between expressivity and tractability. Process tensors are intrinsically quantum, high-rank objects whose complexity scales exponentially with both time and system dimension. In contrast, scalable circuit-level simulation and decoding are typically built around stochastic Pauli fault models. While such models are natural in stabilizer settings~\cite{gottesmanStabilizerCodesQuantum1997}, they are most often used in independent or weakly correlated forms, and therefore may fail to capture the genuinely spatiotemporal correlations present in real hardware. What is needed is a principled intermediate description---rich enough to encapsulate genuine spatiotemporal correlations, yet structured enough to support scalable circuit-level simulation, inference, and benchmarking.

\begin{figure*}[t]
    \centering
    \includegraphics[scale=1.0]{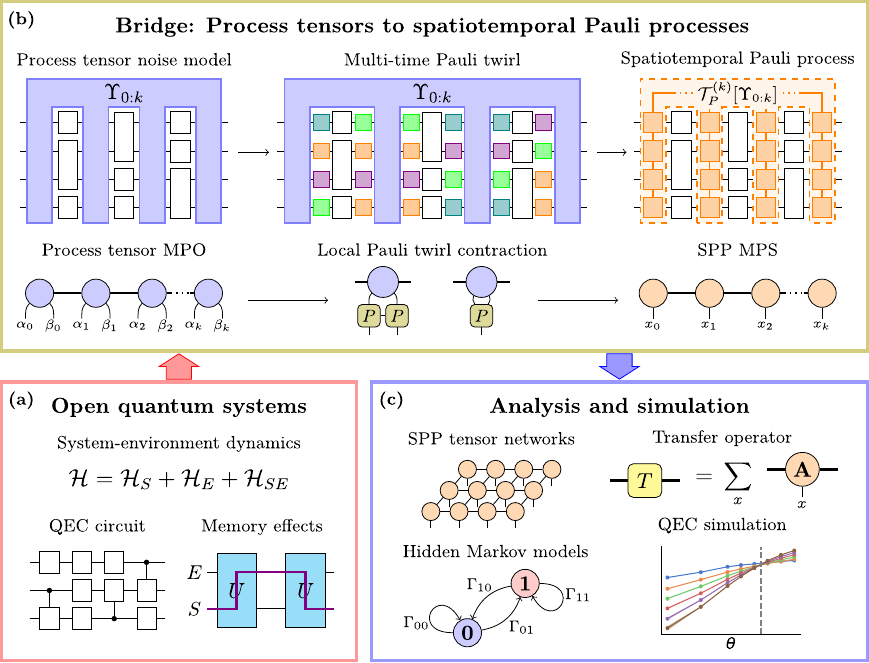}
    \caption{\textbf{Framework overview: Process tensors to spatiotemporal Pauli processes.} \textbf{(a)} Starting from microscopic open quantum system dynamics, we represent a process tensor description of noise and its tensor network form. \textbf{(b)} Applying the multi-time Pauli twirl---as induced operationally by Pauli-frame randomization---maps \(\Upsilon_{0:k}\) onto a \emph{spatiotemporal Pauli process} (SPP), i.e., a joint distribution over Pauli trajectories with a classical tensor network representation. \textbf{(c)} The resulting SPP admits scalable analysis and simulation tools, including transfer operator correlation diagnostics, hidden Markov model realizations, and Monte Carlo sampling for QEC benchmarking.}
    \label{fig:paper-overview}
\end{figure*}

To bridge this divide, we introduce \emph{Spatiotemporal Pauli Processes} (SPPs): the natural multi-time generalization of Pauli channels, describing stochastic processes over Pauli operators across space and time. Formally, an SPP is obtainable from any noise process---however non-Markovian---by applying the \emph{multi-time Pauli twirl} to its process tensor (\cref{fig:paper-overview}(b), upper), yielding a \emph{process-separable} Pauli comb, equivalently a joint probability distribution over Pauli trajectories. Crucially, this projection is operationally motivated: it formalizes the effective noise model induced by Pauli-frame randomization~\cite{silvaScalableProtocolIdentification2008,magesanModelingQuantumNoise2013} and related randomized compiling techniques~\cite{wallmanNoiseTailoringScalable2016,wareExperimentalPauliframeRandomization2021,bealeRandomizedCompilingFaulttolerant2023,jainImprovedQuantumError2023}. SPPs thereby constitute an operationally grounded noise framework compatible with both microscopic open systems modeling and the stabilizer workflow of QEC design and analysis---a single language in which existing correlated noise models can be recast, compared, and extended.

We go beyond the operational definition by showing that SPPs admit constructive tensor network representations directly inherited from underlying system-environment dynamics (\cref{fig:paper-overview}(b)). Because the multi-time Pauli twirl acts strictly via fixed \emph{local} contractions on the system legs, the induced SPP inherits a classical tensor network form whose entries encode joint spatiotemporal Pauli trajectory probabilities. Moreover, the inner bond dimension of this network is bounded by the Liouville space~\cite{woodTensorNetworksGraphical2015} dimension of the underlying environment, tying the effective memory of the Pauli process to the size of the physical bath. To make these correlations interpretable, we develop a transfer operator formalism for SPPs that links spectral properties to correlation times, and demonstrate that, for suitable classes of SPPs, these tensor networks can be recast exactly as classical hidden Markov models. Together, these structures make the bridge practical rather than merely formal: a Hamiltonian-level device model can be mapped, exactly, into the standard stabilizer Monte Carlo pipelines of circuit-level QEC simulation, supplying generative correlated noise models and priors for correlation-aware decoding and benchmarking.

We deploy this framework to systematically benchmark primitives of surface code computation---memory and stability---under correlated noise, with circuit-level simulations of code distances up to \(19\). First, a one-dimensional temporal ``storm'' model sweeps the correlation time while fixing single-round marginal error rates, isolating memory as the sole variable. This controlled setting shows that temporal correlations alone progressively degrade logical memory and stability, blunting the effective exponential suppression of errors expected from distance scaling. Second, we introduce a genuinely spatiotemporal \((2+1)D\) SPP derived from a two-dimensional quantum cellular automaton (QCA)~\cite{schumacherReversibleQuantumCellular2004,farrellyReviewQuantumCellular2020} bath. Under system Pauli twirling, the QCA dynamics maps to an effective SPP described \emph{exactly} by a probabilistic cellular automaton (PCA) with a nonlinear transition kernel---the nonlinearity a distinctly quantum footprint of bath coherence. Tuning a single parameter controlling coherent bath interactions drives the Pauli process into a pseudo-critical regime with critical slowing down and macroscopic error avalanches, whose spatiotemporal fault patterns cause a breakdown and eventual reversal of standard surface code distance scaling. That correlated environments can defeat fault tolerance in principle has long been established analytically, through studies of codes coupled to common bosonic baths and threshold theorems for correlated Hamiltonian noise~\cite{klesseQuantumErrorCorrection2005,novaisDecoherenceCorrelatedNoise2006,novaisSurfaceCodeThreshold2013,hutterBreakdownSurfacecodeError2014,terhalFaulttolerantQuantumComputation2005,aharonovFaultTolerantQuantumComputation2006,ngFaulttolerantQuantumComputation2009,preskillSufficientConditionNoise2013}. Our results supply a constructive counterpart: the same class of microscopic mechanisms, rendered as SPPs, can be simulated, tuned, and benchmarked directly at the circuit level.

Significantly, rather than positing correlated bursts phenomenologically, we obtain them---with their full spatiotemporal profile---from a single microscopic parameter. This inverts the standard workflow at the physics--QEC interface: correlated failure modes need not be inserted by hand at the circuit level, but can be derived from device-level dynamics, tuned continuously, and stress-tested against full circuit-level QEC pipelines, providing a microscopically grounded testbed for the correlated error mechanisms observed but not yet fully understood in hardware. More broadly, these results establish a concrete, computationally tractable link between non-equilibrium noise dynamics and QEC performance, and demonstrate that below-threshold average error rates alone do not guarantee scalable error suppression.

Several related works have also investigated non-Markovian dynamics under Pauli randomization or structure. Operational Markovianization has been studied in randomized benchmarking, where randomized compiling and dynamical decoupling can suppress aspects of non-Markovian behavior~\cite{figueroa-romeroOperationalMarkovianizationRandomized2024}. Raza \emph{et al.}~\cite{razaOnlineLearningQuantum2024} developed online learning and shadow tomography results for Pauli channels and quantum combs; notably, their multi-time construction includes a closely related Pauli twirl of the comb's Choi representation. Liu \emph{et al.}~\cite{liuNonMarkovianNoiseSuppression2024} derived an analogous reduction using the associated ``Choi-channel'' representation, with related effects demonstrated experimentally in noise-mitigation settings~\cite{liuRealizingUniversalNonMarkovian2025}. Furthermore, a related, but conceptually distinct line of work has investigated the \emph{single-snapshot} non-Markovianity of Pauli and Pauli-twirled channels via CP divisibility and embeddability criteria~\cite{kattemolleNonMarkovianityInducedPaulitwirling2026,seifSingleSnapshotNonMarkovianity2026}. Our framework, in contrast, treats non-Markovianity as an explicitly multi-time property of general quantum combs; none of these works provide the constructive tensor network implementation, physical memory bounds, or circuit-level QEC deployment developed here.

The remainder of this paper is structured as follows. \Cref{sec:background} reviews background on quantum channels, Pauli twirling, and process tensors. \Cref{sec:pt-tn} constructs tensor network representations of multi-time processes from microscopic dynamics. \Cref{sec:spp} introduces SPPs via the multi-time Pauli twirl and derives their tensor network structure, with worked examples. \Cref{sec:correlation-structure} develops the transfer operator formalism and the connection to hidden Markov models. \Cref{sec:temporalstormmodel} benchmarks surface code memory and stability under the temporally correlated storm model. \Cref{sec:qcamodel} derives an effective PCA from a coherent quantum bath and analyzes the resulting pseudo-critical breakdown of distance scaling. Finally, \cref{sec:discussion} concludes and discusses future applications.

\section{Background and formalism}\label{sec:background}
In this section, we establish the requisite machinery for our spatiotemporal Pauli process framework. We begin by reviewing quantum channels and their representations, then focus on Pauli channels and Pauli twirling, motivating their central role in quantum error correction while discussing their limitations. Finally, we generalize from two-time maps to multi-time dynamics via the process tensor formalism, which provides a rigorous description of non-Markovian memory and sets the stage for our later analysis of spatiotemporal noise correlations.

\subsection{Quantum channels and representations}
The physical evolution of a quantum state \(\rho \in \mathcal{B}(\mathcal{H})\), where \(\mathcal{B}(\mathcal{H})\) is the space of bounded linear operators on a Hilbert space \(\mathcal{H}\), is entirely described by a \emph{quantum channel} \cite{nielsenQuantumComputationQuantum2010}. A quantum channel is a linear map \(\mathcal{E}: \mathcal{B}(\mathcal{H}_{\mathrm{in}}) \to \mathcal{B}(\mathcal{H}_{\mathrm{out}})\) that is completely positive and trace preserving (CPTP). Formally, a quantum channel \(\mathcal{E}\) must satisfy
\begin{enumerate}
\item \textbf{Linearity}
\begin{equation}
    \mathcal{E}[\alpha\rho_1+\beta\rho_2]=\alpha\mathcal{E}[\rho_1]+\beta\mathcal{E}[\rho_2].
\end{equation}
\item \textbf{Complete positivity (CP)}
\begin{equation}
    (\mathcal{I}_A\otimes \mathcal{E}_B)[\rho_{AB}]\geq 0, \quad \forall \rho_{AB} \ge 0.
\end{equation}
\item \textbf{Trace preservation (TP)}
\begin{equation}
    \Tr(\mathcal{E}[\rho]) = \Tr(\rho).
\end{equation}
\end{enumerate}

Physically, a channel represents a \emph{two-time} quantum map evolving an initial state \(\rho(t_0)\) at time \(t_0\) to a final state \(\rho(t_1)\) at time \(t_1\).

A quantum channel admits several equivalent representations, each useful in different contexts. The most well-known representation is the Kraus (operator-sum) representation, which expresses the channel's action on an input state \(\rho\) as
\begin{equation}
    \mathcal{E}[\rho] = \sum_i K_i \rho K_i^\dagger,
\end{equation}
where the Kraus operators \(\{K_i\}\) satisfy \(\sum_i K_i^\dagger K_i = \mathbb{I}\) to ensure trace preservation. 

Another important representation is provided by the Choi-Jamiołkowski isomorphism (CJI), which establishes a bijection between CP maps \(\mathcal{E}: \mathcal{B}(\mathcal{H}_{\mathrm{in}}) \to \mathcal{B}(\mathcal{H}_{\mathrm{out}})\) and positive semidefinite operators \(\Upsilon \in \mathcal{B}(\mathcal{H}_{\mathrm{in}} \otimes \mathcal{H}_{\mathrm{out}})\). The Choi operator \(\Upsilon\) is defined as
\begin{equation}
    \Upsilon = (\mathcal{I} \otimes \mathcal{E})[\ketbra{\Phi^+}{\Phi^+}],
\end{equation}
where \(\ket{\Phi^+} = \sum_i \ket{i}\otimes \ket{i}\) is the (unnormalized) Bell state on \(\mathcal{H}_{\mathrm{in}} \otimes \mathcal{H}_{\mathrm{in}}\). The action of the Choi operator \(\Upsilon\) on an input state \(\rho\) can then be written as
\begin{equation}\label{eq:choi-rep}
    \mathcal{E}[\rho] = \Tr_{\mathrm{in}}[(\rho^T \otimes \mathbb{I}_{\mathrm{out}})\Upsilon].
\end{equation}

The CP condition on \(\mathcal{E}\) translates into the positive semidefiniteness of \(\Upsilon\), while the TP condition translates to the partial trace relation \(\Tr_{\mathrm{out}}[\Upsilon] = \mathbb{I}_{\mathrm{in}}\). This representation directly encodes the complete dynamical information of a channel into a static bipartite state. This duality is generalized to later multi-time processes described by process tensors in \cref{sec:process_tensors}.

The Stinespring dilation theorem provides an additional, physically motivated perspective on quantum channels. It states that any CPTP map acting on a system space \(\mathcal{B}(\mathcal{H}_S)\) can be realized as a unitary operation \(U\) on a larger, dilated space \(\mathcal{B}(\mathcal{H}_S \otimes \mathcal{H}_E)\), followed by a partial trace over an auxiliary environment \(E\). This is expressed as
\begin{equation}
    \mathcal{E}[\rho_S] = \Tr_E[U(\rho_S \otimes \sigma_E)U^\dagger],
\end{equation}
where \(\sigma_E\) is an initial state of the environment, which must be initially separable from \(\rho_S\).

Throughout this work, we also employ further representations, including the \(\chi\)-matrix\footnote{Also known as a process matrix, a term we avoid using due to a terminological clash with the higher-order quantum object with potentially indefinite causal order~\cite{oreshkovQuantumCorrelationsNo2012}.} representation and the (Liouville-)superoperator representation. The \(\chi\)-matrix is related to the Choi matrix via a change of basis to the Pauli basis, while the superoperator representation is convenient for numerical calculations. We do not explicitly review these representations nor the detailed mappings between them. For a comprehensive overview of quantum channel representations and the transformations between them in a graphical calculus, we refer the reader to Ref.~\cite{woodTensorNetworksGraphical2015}.

\subsection{Pauli channels and Pauli twirling}
An important class of quantum channels, especially in quantum error correction, is the class of Pauli channels~\cite{nielsenQuantumComputationQuantum2010}. A Pauli channel \(\mathcal{E}_P\) acting on \(n\) qubits is defined as a convex mixture of Pauli operators,
\begin{equation}\label{eq:pauli-channel}
    \mathcal{E}_P[\rho] = \sum_{P \in \mathbb{P}^{(n)}} p_P P \rho P,
\end{equation}
where \(\mathbb{P}^{(n)} = \{I, X, Y, Z\}^{\otimes n}\) is the \(n\)-qubit Pauli basis, and coefficients \(p_P \geq 0\) satisfy \(\sum_{P} p_P = 1\). 

Any general channel \(\mathcal{E}\) can be projected into a Pauli channel via a procedure known as \emph{Pauli twirling}. The Pauli twirl of \(\mathcal{E}\) is defined as the group average
\begin{equation}\label{eq:pauli-twirl}
    \mathcal{T}_P(\mathcal{E})[\rho] = \frac{1}{|\mathbb{P}^{(n)}|} \sum_{P \in \mathbb{P}^{(n)}} P \mathcal{E}(P \rho P) P,
\end{equation}
where \(|\mathbb{P}^{(n)}| = 4^n\). The image of the twirling operation is strictly a Pauli channel of the form in \cref{eq:pauli-channel}. Equivalently, given the \(\chi\)-matrix representation of \(\mathcal{E}\), the twirled channel \(\mathcal{T}_P(\mathcal{E})\) can be computed by zeroing all off-diagonal elements, leaving only the diagonal elements corresponding to probabilities \(p_P\) of the Pauli operators. In the Choi representation, the twirling operation acts as a CPTP \emph{projector} (an idempotent and hermiticity-preserving map) on the Choi operator \(\Upsilon\):
\begin{equation}
    \mathcal{T}_P(\Upsilon) = \frac{1}{4^n} \sum_{P \in \mathbb{P}^{(n)}} (P \otimes P) \Upsilon (P \otimes P).
\end{equation}

This operation removes coherences between different Pauli operators (off-diagonal elements in the \(\chi\)-matrix) while preserving the CPTP condition. Pauli twirling can therefore be viewed as a \emph{supermap}~\cite{chiribellaQuantumCircuitArchitecture2008}---a higher-order quantum operation---as it maps valid quantum channels to valid quantum channels. We will later generalize this projection in the process tensor framework. Before turning to multi-time processes, however, we first discuss the relevance and limitations of Pauli noise models in quantum error correction. 

\subsection{Pauli noise as an effective model for QEC}
Pauli noise models occupy a privileged role in quantum error correction. The discretization of errors via projective syndrome measurements collapses continuous error processes into discrete outcomes associated with Pauli errors---allowing the correction of general noise~\cite{knillTheoryQuantumErrorcorrecting1997,gottesmanStabilizerCodesQuantum1997}. Furthermore, via the Gottesman-Knill theorem~\cite{gottesmanHeisenbergRepresentationQuantum1998}, Clifford circuits subject to Pauli noise are amenable to efficient classical simulation, enabling large-scale Monte Carlo studies of QEC performance in regimes relevant for fault-tolerant computation~\cite{aaronsonImprovedSimulationStabilizer2004,gidneyStimFastStabilizer2021}.

Physical noise processes, however, are described by general CPTP maps that need not be Pauli channels~\cite{nielsenQuantumComputationQuantum2010}. Coherent unitary errors---arising from miscalibrated gates, drift, or unwanted system-system or system-environment couplings---induce structure in noise that is explicitly discarded by the Pauli twirling approximation. Consequently, stochastic Pauli models cannot wholly capture the full underlying dynamics and can, in some cases, underestimate logical error rates, as seen in worst-case analyses with highly coherent gate-based noise~\cite{greenbaumModelingCoherentErrors2017,huangPerformanceQuantumError2019}. These discrepancies are further exacerbated in non-Markovian settings~\cite{kobayashiTensornetworkDecodersProcess2024}, where coherent environmental memory effects defy characterization even by sequences of general quantum channels.

Despite these limitations, there are compelling reasons to treat Pauli noise as a central effective noise model for QEC. First and foremost, Pauli twirling can be engineered in practice. Randomized compiling (RC) actively tailors general noise into stochastic Pauli noise on the physical level by interleaving and compiling random Pauli gates into the circuit~\cite{wallmanNoiseTailoringScalable2016,bealeRandomizedCompilingFaulttolerant2023,jainImprovedQuantumError2023}. Second, many of the most dramatic examples where Pauli models misdiagnose logical failure rates rely on highly structured and fully coherent error patterns that are unlikely to persist in large-scale devices~\cite{greenbaumModelingCoherentErrors2017}. QEC performance under less pathological non-Pauli models, such as amplitude or phase damping, is well approximated by their Pauli twirled counterparts~\cite{darmawanTensornetworkSimulationsSurface2017,myersSimulatingGeneralNoise2025}. Furthermore, previous studies suggest that residual coherent effects on logical qubits are naturally suppressed with increasing code distance for standard stabilizer codes~\cite{bealeQuantumErrorCorrection2018,iversonCoherenceLogicalQuantum2020}.

Finally, the efficiency of Pauli noise representations enables a range of practical noise characterization protocols. In particular, the reconstruction of effective phenomenological Pauli noise from syndrome measurement data alone~\cite{wagnerOptimalNoiseEstimation2021,wagnerPauliChannelsCan2022,gicevQuantumComputerError2024,remmExperimentallyInformedDecoding2025} is highly valuable for in situ monitoring and calibration of real-world QEC implementations~\cite{kunjummenSituCalibrationUnitary2025}.

Crucially, even with a Pauli description of noise, one can still accommodate the rich spatiotemporal correlations induced by non-Markovian dynamics. Most existing analyses, however, assume uncorrelated, Markovian error models, whereas correlations in stochastic models can have a pronounced impact on the efficacy of QEC protocols \cite{fowlerQuantifyingEffectsLocal2014,fkamDetrimentalNonMarkovianErrors2025}. 

This motivates the central arc of the paper: We begin from a fully general, multi-time description of noise, and then construct an operational bridge to the Pauli fault language used by stabilizer-based QEC. In the next subsection, we introduce the process tensor formalism, which provides a complete description of multi-time quantum processes---including non-Markovian memory effects---and serves as the starting point of our construction. We then develop tensor network representations of process tensors in \cref{sec:pt-tn}, which are essential for deriving the efficient spatiotemporal Pauli descriptions in \cref{sec:spp,sec:correlation-structure}.

\subsection{Process tensors and non-Markovianity}\label{sec:process_tensors}
While quantum channels describe two-time dynamics, they are insufficient for capturing general multi-time quantum processes, especially in the presence of non-Markovian memory effects. The quantum \emph{process tensor} provides a general spatiotemporal framework that encodes all such correlations. It can be viewed as a \emph{higher-order quantum map} (or quantum comb), that takes as input a sequence of \(k\) \emph{instruments} \(\mathbf{A}_{0:k-1}=(\mathcal{A}_0,\dots,\mathcal{A}_{k-1})\), and outputs the resulting system state \(\rho_\mathrm{out}\). Here each \(\mathcal{A}_j\) denotes a CP map acting on the system at time \(t_j\), representing an arbitrary control operation, gate, or projective measurement performed by an experimenter. Additionally, what exactly constitutes an instrument can be defined contextually. For instance, relevant for QEC, a syndrome extraction cycle can be treated as a single instrument~\cite{tanggaraStrategicCodeUnified2024,kobayashiTensornetworkDecodersProcess2024}. The set of sequences \(\{\mathbf{A}_{0:k-1}\}\) thus broadly encapsulates all possible controllable dynamics imposed on the system between the initial preparation and the final time.

For the purposes of modeling the multi-time dynamics starting from an initial system state \(\rho_\mathrm{in}\)\footnote{For cases where the system and environment are initially correlated, the process tensor acts purely on a sequence of instruments \(\mathbf{A}_{0:k-1}\), where the initial map \(\mathcal{A}_0\) serves as a preparation map. Correspondingly, a \emph{supermap}~\cite{chiribellaQuantumCircuitArchitecture2008} is exactly a 1-slot process tensor in this definition.}, we define a \emph{\(k\)-slot} process tensor \(\mathcal{T}_{0:k}\) as a multilinear map acting on \(\rho_\mathrm{in}\) and \(\mathbf{A}_{0:k-1}\) by
\begin{equation}
\mathcal{T}_{0:k}[\rho_\mathrm{in},\mathbf{A}_{0:k-1}] = \rho_\mathrm{out}.
\end{equation}
Like quantum channels, a process tensor admits several equivalent representations, the most prevalent being the Choi representation. Just as the CJI establishes a channel-state duality, there exists a process-state duality that maps any quantum process \(\mathcal{T}_{0:k}\) to a positive semidefinite operator \(\Upsilon_{0:k} \in \mathcal{B}(\bigotimes_{j=0}^k \mathcal{H}_{\mathrm{in}_j} \otimes \mathcal{H}_{\mathrm{out}_j})\). Its action on an initial state \(\rho_{\mathrm{in}}\) and sequence \(\mathbf{A}_{0:k-1}\) is given by
\begin{multline}\label{eq:pt-choi-rep}
    \rho_\mathrm{out}(\rho_\mathrm{in}, \mathbf{A}_{0:k-1}) \\
    = 
    \Tr_{\overline{\mathrm{out}}} \Big[ \big( \rho_{\mathrm{in}}^T
    \otimes \hat{\mathcal{A}}_0^T \otimes \cdots
    \otimes \hat{\mathcal{A}}_{k-1}^T \otimes \mathbb{I}_{\mathrm{out}} \big)
    \Upsilon_{0:k} \Big],
\end{multline}
where \(\Tr_{\overline{\mathrm{out}}}\) denotes the partial trace over all subspaces except the final output space \(\mathcal{H}_{\mathrm{out}_k}\), and \(\hat{\mathcal{A}}_j \in \mathcal{B}(\mathcal{H}_{\mathrm{out}_j} \otimes \mathcal{H}_{\mathrm{in}_{j+1}})\) is the Choi operator of the CP map \(\mathcal{A}_j\). For \(k=0\), the above expression reduces to the channel case in \cref{eq:choi-rep}.

Equivalent Kraus, superoperator, and \(\chi\)-matrix representations can be derived from the Choi form in direct analogy with quantum channels. The Choi representation is convenient because it translates the physical properties of a quantum process into concrete properties of the operator \(\Upsilon_{0:k}\). Crucially, a process tensor \(\mathcal{T}_{0:k}\) represents a valid, physical process if and only if its Choi representation \(\Upsilon_{0:k}\) satisfies the following conditions:

\begin{enumerate}
    \item \textbf{Multi-linearity}

    The Choi operator \(\Upsilon_{0:k}\) represents a multilinear map \(\mathcal{T}_{0:k}\) with respect to the input state \(\rho_{\mathrm{in}}\) and each instrument in the instrument sequence \(\mathbf{A}_{0:k-1}\),
    \begin{equation}
        \begin{aligned}
        &\mathcal{T}_{0:k}[\alpha\rho_1 + \beta\rho_2, \mathbf{A}_{0:k-1}] \\
        &\quad =
        \alpha\mathcal{T}_{0:k}[\rho_1, \mathbf{A}_{0:k-1}]
        + \beta\mathcal{T}_{0:k}[\rho_2, \mathbf{A}_{0:k-1}], \\[0.5em]
        &\mathcal{T}_{0:k}[\rho_{\mathrm{in}}, \alpha\mathbf{A}_1 + \beta\mathbf{A}_2] \\
        &\quad =
        \alpha\mathcal{T}_{0:k}[\rho_{\mathrm{in}}, \mathbf{A}_1]
        + \beta\mathcal{T}_{0:k}[\rho_{\mathrm{in}}, \mathbf{A}_2].
        \end{aligned}
    \end{equation}

    \item \textbf{Positive semidefiniteness}
    
    This is the direct generalization of the CP condition for quantum channels. The Choi operator must be positive semidefinite,
    \begin{equation}
        \Upsilon_{0:k} \geq 0,
    \end{equation}
    ensuring that, for any valid initial state and sequence of instruments, the output state is positive semidefinite.
    \item \textbf{Causal (or affine) constraints}
    
    Also known as the containment property, these constraints are a more intricate generalization of the TP condition of channels to multi-time processes. The Choi operator must satisfy the recursive constraints,
    \begin{equation}\label{eq:affine-constraint}
        \begin{split}
            &\Tr_{\mathrm{out}_j}[\Upsilon_{0:j}] = \Upsilon_{0:j-1} \otimes \mathbb{I}_{\mathrm{in}_j}, \quad \forall j=1,\dots,k, \\
            &\Tr_{\mathrm{out}_0}[\Upsilon_{0:0}] = \mathbb{I}_{\mathrm{in}_0}.
        \end{split}
    \end{equation}
    Operationally, this condition ensures that instruments performed at future times cannot influence outcome statistics at earlier times, thereby enforcing causal ordering. Equivalently, these constraints can be expressed as a linear system of Pauli expectation values. Denoting
    \begin{multline}
        \braket{P_{\mathrm{in}_0}P_{\mathrm{out}_0}\cdots P_{\mathrm{in}_k}P_{\mathrm{out}_k}} \\
        = \Tr\!\left[(P_{\mathrm{in}_0}\otimes P_{\mathrm{out}_0} \otimes\cdots\otimes P_{\mathrm{in}_k}\otimes P_{\mathrm{out}_k})\Upsilon_{0:k}\right],
    \end{multline}
with Pauli strings \(P_j \in \mathbb{P}^{(n)}\) of cardinality \(|\mathbb{P}^{(n)}| = 4^n\), and non-identity Pauli strings \(Q_j \in \mathbb{P}^{(n)} \setminus \{\mathbb{I}\}\), \cref{eq:affine-constraint} can be rewritten as the following set of independent linear equations,
    \begin{equation}
        \begin{split}
            &\braket{P_{\mathrm{in}_0}P_{\mathrm{out}_0}\cdots P_{\mathrm{out}_{k-1}}Q_{\mathrm{in}_k}\mathbb{I}_{\mathrm{out}_k}} = 0, \\
            &\braket{P_{\mathrm{in}_0}P_{\mathrm{out}_0}\cdots P_{\mathrm{out}_{k-2}}Q_{\mathrm{in}_{k-1}}\mathbb{I}_{\mathrm{out}_{k-1}}\mathbb{I}_{\mathrm{in}_k}\mathbb{I}_{\mathrm{out}_k}} = 0, \\
            &\quad\vdots \\
            &\braket{Q_{\mathrm{in}_0}\mathbb{I}_{\mathrm{out}_0}\cdots \mathbb{I}_{\mathrm{in}_k}\mathbb{I}_{\mathrm{out}_k}} = 0.
        \end{split}
    \end{equation}
    The normalization condition 
    \begin{equation}
        \braket{\mathbb{I}_{\mathrm{in}_0}\mathbb{I}_{\mathrm{out}_0}\cdots \mathbb{I}_{\mathrm{in}_k}\mathbb{I}_{\mathrm{out}_k}} = \Tr{[\Upsilon_{0:k}]} = 2^{n(k+1)},
    \end{equation}
    provides an additional linear constraint, however it is sometimes convenient to work with the unnormalized Choi state. 
\end{enumerate}


Beyond these general validity conditions, structural features of a process tensor reflect the physical nature of the underlying dynamics. Notably, entanglement across temporal partitions of a Choi state \(\Upsilon_{0:k}\)---referred to as \emph{temporal entanglement}---signals that a process is \emph{genuinely quantum non-Markovian}. Concretely, a process exhibits genuinely quantum temporal correlations if and only if its Choi state cannot be written as a convex mixture of products, 
\begin{equation}\label{eq:process-separable}
\Upsilon_{0:k} = \sum_{\vec{\alpha} = (\alpha_0, \alpha_1, \dots, \alpha_k)} p_{\vec{\alpha}} \Upsilon_0^{\alpha_0} \otimes \Upsilon_1^{\alpha_1} \otimes \cdots \otimes \Upsilon_k^{\alpha_k},
\end{equation}
with \(p_\alpha \geq 0\) and each \(\Upsilon_j^\alpha\) is the Choi representation of a CPTP channel. Processes that can be written in this form are \emph{process-separable}, possessing temporal correlations that are classically simulable. Conversely, a process is fully \emph{Markovian}---exhibiting neither quantum nor classical temporal correlations---if and only if its Choi operator factorizes across all time steps,
\begin{equation}\label{eq:markovian-process}
    \Upsilon_{0:k}^{\textrm{Markov}} = \Upsilon_0 \otimes \Upsilon_1 \otimes \cdots \otimes \Upsilon_k,
\end{equation}
where each \(\Upsilon_j\) is CPTP. The dynamics of a Markovian process is therefore equivalently generated by a sequence of independent quantum channels.

The physical origin of these structures can be traced back to the system-environment picture. By a natural extension of the Stinespring dilation theorem, any \(k\)-slot process tensor can be realized as the joint multi-time dynamics of a system \(S\) and an environment \(E\), followed by a partial trace over the environment. This is represented by the equation,
\begin{multline}
    \mathcal{T}_{0:k}[\rho_S, \mathbf{A}_{0:k-1}] \\ 
    = \Tr_E\left[\mathcal{U}_k \circ \mathcal{A}_{k-1} \circ \mathcal{U}_{k-1} \circ \dots \circ \mathcal{A}_0 \circ \mathcal{U}_0(\sigma_E\otimes\rho_S) \right],
\end{multline}
where \(\mathcal{U}(\cdot) = U(\cdot)U^\dagger\) implements a joint system-environment unitary, \(\mathbf{A}_{0:k-1}\) is a sequence of CP maps acting only on \(S\), and we have assumed an initially separable system and environment state \(\sigma_E \otimes \rho_S\).

This dilation is more than a mathematical device; it depicts the fundamental physics governing open quantum dynamics. Non-Markovianity arises when the environment retains information about the system’s past and feeds it back into its future evolution. What is precisely considered an environment can be broadly interpreted as any inaccessible (or ignored) degrees of freedom that influences the system's evolution. In the QEC context, this could include spurious two-level systems, higher energy levels of transmon qubits, cosmic ray-induced phonons, or even other neighboring qubits when considering reduced subsystem dynamics.

Chiefly, the process tensor formalism abstracts away the explicit modeling of the environment. Instead, it centralizes the environment's \emph{influence} on the system with respect to all possible interventions by an experimenter.

In the next section, we exploit this structure to construct tensor network representations of process tensors, which serve as the foundation of our spatiotemporal Pauli process framework.

\section{Process tensors as tensor networks}\label{sec:pt-tn}
Quantum maps are tensors; when correlations between their indices are even modestly constrained, they often admit compact tensor network representations. A paradigmatic example is the matrix product state (MPS) description of quantum many-body states with area-law entanglement. The same principle extends to quantum processes: process tensors are high-rank objects whose temporal correlations are often structured, making tensor network representations both natural and efficient~\cite{keelingProcessTensorApproaches2025}. Accordingly, process tensors have been represented as matrix product operators (MPOs)~\cite{pollockNonMarkovianQuantumProcesses2018,cygorekUnderstandingUtilizingInner2025}, locally purified density operators (LPDOs)~\cite{whiteUnifyingNonMarkovianCharacterization2025}, and tree tensor networks (TTNs)~\cite{dowlingCapturingLongRangeMemory2024}.

In what follows, we construct tensor network representations of process tensors directly from their microscopic dynamics. We begin with a one-dimensional temporal MPO obtained from a Stinespring dilation, before generalizing to higher-dimensional spatiotemporal tensor networks for composite systems. Throughout, we employ abstract index notation supplemented by a graphical tensor network calculus. A table translating primitives of quantum information in Dirac notation, abstract indices, and graphical diagrams is provided in \cref{app:notations}.

\subsection{Stinespring-to-MPO construction}
We begin by constructing an MPO representation of a process tensor from its Stinespring dilation. A \(k\)-slot process tensor is fully specified by a sequence of joint system--environment unitaries \(\{U_j\}_{j=0}^k\) together with an initial environment state \(\sigma_E\). In general, one may start from a sequence of joint CPTP maps \(\{\mathcal{E}_j\}_{j=0}^k\), which, by Stinespring's theorem can always be dilated to unitaries.

For a unitary \(U\) acting on \(\mathcal{H}_E \otimes \mathcal{H}_S\) of dimensions \(d_E\) and \(d_S\), its Liouville (superoperator) form is \(\mathcal{U} = \overline{U} \otimes U\), where \(\overline{U} = U^*\) denotes elementwise complex conjugation. In index notation,
\begin{equation}
    \mathcal{U}^{e'_\out s'_\out e_\out s_\out}{}_{e'_\inp s'_\inp e_\inp s_\inp} = \overline{U}^{e'_\out s'_\out}{}_{e'_\inp s'_\inp} U^{e_\out s_\out}{}_{e_\inp s_\inp},
\end{equation}
where \(e_\inp, s_\inp\) (\(e_\out, s_\out\)) label environment and system input (output) indices, and primed indices mark dual `bra' spaces, with unprimed indices denoting vector ket spaces. To make the connection between the Stinespring dilation and tensor networks transparent, we introduce a specialized tensor representation derived from the Liouville form via a \emph{reorder--fuse--transpose} (RFT) operation:
\begin{enumerate}
    \item \textbf{Reorder} indices to group each ket index with its corresponding bra index.
    \item \textbf{Fuse} each bra--ket (primed--unprimed) pair into a single compound index.
    \item \textbf{Transpose} globally so that time flows from left to right, consistent with standard circuit diagrams.
\end{enumerate}
Explicitly, for a superoperator \(\mathcal{S}\),
\begin{multline}\label{eq:rft}
    \mathcal{S}^{b'_0 b'_1 \dots b_0 b_1}{}_{a'_0 a'_1 \dots a_0 a_1} 
    \xmapsto{\;\text{reorder}\;}
    \mathcal{S}^{b'_0 b_0 b'_1 b_1 \dots}{}_{a'_0 a_0 a'_1 a_1 \dots} \\
    \xmapsto{\;\text{fuse}\;}
    \mathcal{S}^{\beta_0 \beta_1 \dots}{}_{\alpha_0 \alpha_1 \dots} 
    \xmapsto{\;\text{transpose}\;}
    \mathbf{S}^{\alpha_0 \alpha_1 \dots}{}_{\beta_0 \beta_1 \dots},
\end{multline}
where \(\alpha_j = (a'_j a_j)\) and \(\beta_j = (b'_j b_j)\) are fused indices of dimension \(\dim(\alpha_j) = \dim(a_j)^2\). Latin letters denote basic indices and Greek letters fused indices. The RFT mapping simply re-expresses a superoperator as a tensor network-friendly object whose left--right orientation matches the temporal ordering of the process. This form is convenient for both graphical calculus and numerical implementations, where each fused index corresponds to the vectorized operator space of a particular physical subsystem.

Applying the RFT operation to a unitary superoperator \(\mathcal{U}_j\) yields a 4-legged tensor \(\mathbf{U}_{(j)}\),
\begin{equation}
    \mathcal{U}_j^{e'_{j+1} e_{j+1} b'_j b_j}{}_{e'_j e_j a'_j a_j} \xmapsto{\;\mathrm{RFT}\;} \mathbf{U}_{(j)}^{\mu_j \alpha_j}{}_{\mu_{j+1} \beta_j},
\end{equation}
with fused indices \(\mu_j = (e'_j e_j)\) for the environment and \(\alpha_j=(a'_j a_j)\), \(\beta_j=(b'_j b_j)\) for the system input and output, respectively. Their dimensions are \(\dim(\mu_j) = d_E^2\) and \(\dim(\alpha_j) = \dim(\beta_j) = d_S^2\). The mapping is graphically depicted as
\begin{equation}
    \valignbox{\includegraphics[scale=0.9]{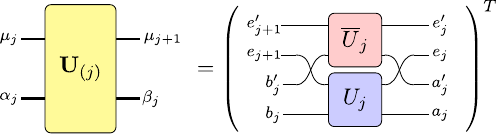}},
\end{equation}
where thick wires represent fused indices.

The initial environment state \(\dket{\sigma_E}\) is mapped similarly,
\begin{equation}
    \sigma^{e'_0e_0} 
    \xmapsto {\;RFT\;} 
    \pmb{\sigma}_{\mu_0} 
    \;\;\longleftrightarrow\;\;
    \valignbox{\includegraphics[scale=0.9]{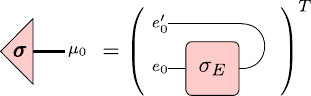}},
\end{equation}
where we have adopted the column-stacking convention for vectorization~\cite{woodTensorNetworksGraphical2015}. To represent the final partial trace over the environment, we include a trace tensor \(\mathbf{t}\), corresponding to the unnormalized Bell state \(\ket{\Phi^+} = \dket{\mathbb{I}}\),
\begin{equation}
    \mathbf{t}^{\mu_{k+1}} = \delta^{e'_{k+1}e_{k+1}}
    \,\;\longleftrightarrow\;\;
    \valignbox{\includegraphics[scale=0.9]{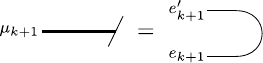}}\ .
\end{equation}

The full \(k\)-slot process tensor is then constructed as an MPO by contracting the environment indices \(\mu_j\) across all time. This yields a tensor \(\mathbf{T}\) with open system indices corresponding to the inputs (\(\alpha_j\)) and outputs (\(\beta_j\)) at each time \(j\). We formalize this in the following definition.

\begin{definition}[MPO representation of a process tensor]\label[definition]{def:pt-mpo}
    A \(k\)-slot process tensor \(\Upsilon_{0:k}\) admits a matrix product operator (MPO) representation,
    \begin{equation}\label{eq:pt-mpo}
        \mathbf{T}^{\vec{\alpha}}{}_{\vec{\beta}}
        =
        \sum_{\mu_0,\dots,\mu_{k+1}}
        \pmb{\sigma}_{\mu_0}
        \bigg(
        \prod_{j=0}^{k}
        \mathbf{U}_{(j)}^{\mu_j\alpha_j}{}_{\mu_{j+1}\beta_j}
        \bigg)
        \mathbf{t}^{\mu_{k+1}},
    \end{equation}
    where \(\vec{\alpha}=(\alpha_0,\dots,\alpha_k)\), and 
    \(\vec{\beta}=(\beta_0,\dots,\beta_k)\).
\end{definition}

In graphical form,
\begin{equation}
    \valignbox{\includegraphics[scale=0.8]{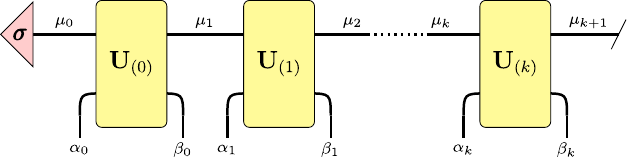}}.
\end{equation}

Here \(\pmb{\sigma}\), \(\mathbf{U}_{(j)}\), and \(\mathbf{t}\) are tensor representations of the initial environment state, the joint \(SE\) unitary at time \(j\), and the final environment trace, respectively. The environment indices \(\mu_j\) are the MPO's virtual bonds of dimension \(D = \dim{(\mu_j)} \le d_E^2\), corresponding to the dimension of the environment's Liouville space.

To recover the standard Choi form, we simply reorder and reshape the fused indices as
\begin{equation}\label{eq:pt-tensor-to-choi}
    \mathbf{T}^{\alpha_0,\dots,\alpha_k}{}_{\beta_0,\dots,\beta_k} \mapsto \mathbf{T}_{\mathrm{Choi}}^{a_0b_0\dots a_kb_k}{}_{a'_0b'_0\dots a'_kb'_k} \longleftrightarrow \Upsilon_{0:k}.
\end{equation}

In the tensor network picture, the action of \(\mathbf{T}\) on an initial state and sequence of instruments is realized as a network contraction between the `top comb' process tensor and the `bottom comb' sequence of instruments. The initial system state \(\dket{\rho_\mathrm{in}}\) and each instrument superoperator \(\mathcal{A}_j\) are mapped into their corresponding tensors \(\pmb{\rho}^\mathrm{in}\) and \(\mathbf{A}_{(j)}\) via the RFT operation. The final output state \(\pmb{\rho}^\mathrm{out}\) is then given by
\begin{equation}
    \pmb{\rho}^\mathrm{out}_{\alpha_{k+1}} = 
    \pmb{\sigma}_{\mu_0}\pmb{\rho}^\mathrm{in}_{\alpha_0}
    \bigg(
        \prod_{j=0}^k
        \mathbf{U}_{(j)}^{\mu_j\alpha_j}{}_{\mu_{j+1}\beta_j}
        \mathbf{A}_{(j)}^{\beta_j}{}_{\alpha_{j+1}}
    \bigg)
    \mathbf{t}^{\mu_{k+1}},
\end{equation}
where repeated indices imply summation (contraction). Graphically, the output state \(\rho_\mathrm{out}\) is described by
\begin{widetext}
    \begin{equation}
        \valignbox{\includegraphics[scale=0.9]{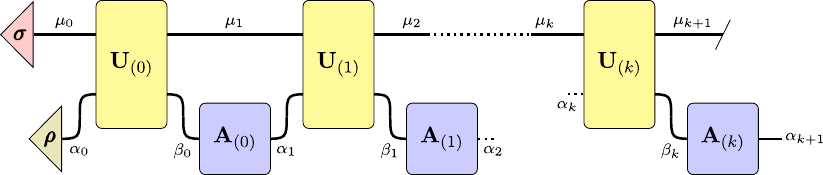}}.
    \end{equation}
\end{widetext}

This contraction naturally generalizes to cases where the instruments are themselves classically or `quantumly' correlated in time, in which case the product of local instrument tensors is replaced by a more general tensor network. Having established a one-dimensional temporal MPO representation of multi-time processes, we now lift this construction to composite systems with additional spatial structure.

\subsection{Higher-dimensional spatiotemporal tensor networks}
So far we have considered only the one-dimensional temporal structure of process tensors, where each tensor carries dense system indices of size \(d^2_S\). For composite systems---such as spin chains---these indices can be further factorized along spatial degrees of freedom, promoting the MPO to a \((1+1)D\) tensor network with time along one axis and space along the other. Concretely, we decompose the system legs of each temporal tensor into local legs 
associated with individual subsystems. For a unitary tensor \(\mathbf{U}_{(j)}\) acting on an \(n\)-partite system, we can perform a series of singular value decompositions (SVDs) along subsystem indices, yielding
\begin{equation}
\begin{aligned}
    \mathbf{U}_{(j)}^{\mu_j\alpha_j}{}_{\mu_{j+1}\beta_j} 
    &=
    \sum_{\nu_j^0,\dots,\nu_j^{n-2}}
    \mathbf{u}_{(0,j)}^{\mu_j^0\alpha_j^0}
    {}_{\beta_j^0\mu_{j+1}^0\nu_j^0}
    \\
    &\quad \times
    \mathbf{u}_{(1,j)}^{\nu_j^0\mu_j^1\alpha_j^1}
    {}_{\beta_j^1\mu_{j+1}^1\nu_j^1}
    \\
    &\qquad\vdots 
    \\
    &\quad\times
    \mathbf{u}_{(n-1,j)}^{\nu_j^{n-2}\mu_j^{n-1}\alpha_j^{n-1}}
    {}_{\beta_j^{n-1}\mu_{j+1}^{n-1}}.
\end{aligned}
\end{equation}
Here the superscript \(i=0,\dots,n-1\) labels subsystems (spatial coordinates), and \(\nu^i_j\) are newly introduced bond indices that link neighboring subsystems at each time \(t_j\). The environment-induced temporal bonds \(\mu^i_j\) remain on each subsystem line.

Performing this decomposition for all \(j\) and similarly for the initial environment tensor \(\pmb{\sigma}\), we obtain a \((1+1)D\) tensor network representation of the process tensor,
\begin{equation}
\begin{aligned}
    \mathbf{T}^{\{\alpha_j^i\}}{}_{\{\beta_j^i\}}
    &=
    \sum_{\{\mu_j^i\}}
    \sum_{\{\lambda^i\}}
    \sum_{\{\nu_j^i\}}
    \Biggl[
        \prod_{i=0}^{n-1}
        \pmb{\sigma}_{(i)}^{\lambda^{i-1}}{}_{\mu_0^i\lambda^i}
    \Biggr]
    \\
    &\quad{}\times
    \Biggl[
        \prod_{j=0}^{k}
        \prod_{i=0}^{n-1}
        \mathbf{u}_{(i,j)}^{\nu^{i-1}_j\mu_j^i\alpha_j^i}
        {}_{\beta_j^i\mu_{j+1}^i\nu_j^i}
    \Biggr]
    \\
    &\quad{}\times
    \Biggl[
        \prod_{i=0}^{n-1}
        \mathbf{t}_{(i)}^{\mu_{k+1}^i}
    \Biggr],
\end{aligned}
\end{equation}
with spatial boundary conditions such that left indices \(\lambda^{-1}, \nu^{-1}_j\) and right indices \(\lambda^{n-1}, \nu^{n-1}_j\) are absent. In this representation, horizontal bonds \(\mu_j^i\) carry temporal correlations (environmental memory), while vertical bonds \(\nu_j^i\) encode spatial correlations between subsystems at each time step. The spatial bond dimensions \(D_s = \dim(\nu_j^i)\) are set by the SVD truncation, while the temporal bond dimensions remain bounded by the environment size, \(D_t = \dim(\mu_j^i) \le d_E^2\). \Cref{fig:tensor-network-2d} displays a graphical representation of this \((1+1)D\) tensor network.

Unlike the one-dimensional temporal MPO representation, \((1+1)D\) tensor networks are generally not both exactly and efficiently contractible. In practice one must resort to approximate contraction schemes, whose cost is controlled by the entanglement across the chosen contraction direction. For 
spatiotemporal processes, however, the causal structure along the time dimension imposes additional constraints on the allowed correlations, which can often be exploited by tailored contraction algorithms~\cite{jorgensenExploitingCausalTensor2019,keelingProcessTensorApproaches2025}.

\begin{figure}[t]
    \centering
    \includegraphics[width=\columnwidth]{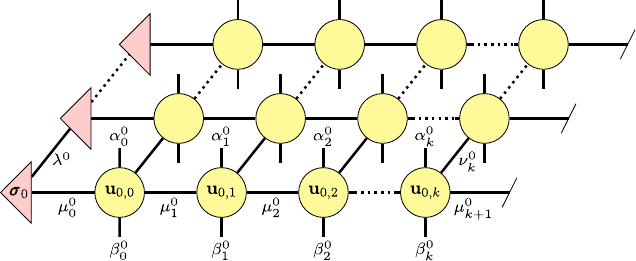}
    \caption{\textbf{Spatiotemporal tensor network form of a process tensor.} Horizontal bonds \(\mu_j^i\) encode temporal correlations, bonds \(\nu_j^i\) carry spatial correlations, and open indices \(\alpha_j^i\), \(\beta_j^i\) correspond to system inputs and outputs at each time \(j \in [0,k]\) and subsystem \(i \in [0,n)\).}
    \label{fig:tensor-network-2d}
\end{figure}

The above construction also extends straightforwardly to higher spatial dimensions by allowing each time slice to be represented by a higher-dimensional tensor network---for example, a two-dimensional projected entangled pair operator (PEPO) for a planar system, yielding a \((2+1)D\) spatiotemporal network. Moreover, while we derived the higher-dimensional network by factorizing the temporal MPO into spatial components, one may instead construct a tensor network reflecting the microscopic system--environment dynamics, especially when the interactions are local \cite{whiteUnifyingNonMarkovianCharacterization2025,kobayashiTensornetworkDecodersProcess2024}. This flexibility allows for efficient constructions tailored to the specific spatiotemporal structure of the underlying process.

\section{Spatiotemporal Pauli processes}\label{sec:spp}
We now introduce the \emph{spatiotemporal Pauli process} (SPP): a stochastic process over Pauli operator trajectories across space and time. While an SPP can be defined phenomenologically as a standalone correlated Pauli noise model, it arises operationally by applying the \emph{multi-time Pauli twirl} to a general quantum process tensor. This maps an arbitrary, possibly non-Markovian, process to a \emph{process-separable} Pauli comb. In tensor network language the twirl is implemented by fixed local contractions on the system legs at each time step, leaving the virtual environment bonds unchanged. The resulting classical network encodes trajectory weights, with temporal bond dimensions bounded by the corresponding process tensor bonds and, for finite environments, by the environment Liouville dimension \(d_E^2\). These results justify SPPs as a natural tensor-network-compatible description of correlated circuit-level Pauli noise, and provide the technical foundation for the explicit noise models and QEC studies developed in later sections.

\subsection{Multi-time Pauli twirling and SPPs}
We begin by extending the standard Pauli twirling operation for quantum channels to multi-time processes.

Let \(\mathbb{P}^{(n)}\) denote the \(n\)-qubit Pauli group modulo global phase, and let \(\mathbb{P}^{(n,k)} \coloneq (\mathbb{P}^{(n)})^{k+1}\) be the set of \((k+1)\)-length Pauli trajectories. Consider a \(k\)-slot process tensor with Choi operator \(\Upsilon_{0:k} \in \mathcal{B}(\bigotimes_{j=0}^k \mathcal{H}_{\mathrm{in}_j} \otimes \mathcal{H}_{\mathrm{out}_j})\). For each trajectory \(\mathcal{P}=(P_0,\ldots,P_k)\in\mathbb{P}^{(n,k)}\), define the associated Pauli product operator \(\Omega_{\mathcal{P}} \coloneq \bigotimes_{j=0}^k (P_j \otimes P_j)\).

\begin{definition}[Multi-time Pauli twirl]\label[definition]{def:multi-time-twirl} The multi-time Pauli twirl \(\mathcal{T}_P^{(k)}\) is the CPTP projector \begin{equation}\label{eq:twirl-upsilon} \mathcal{T}_P^{(k)}(\Upsilon_{0:k}) = \frac{1}{|\mathbb{P}^{(n,k)}|} \sum_{\mathcal{P}\in\mathbb{P}^{(n,k)}} \Omega_{\mathcal{P}}\, \Upsilon_{0:k}\, \Omega_{\mathcal{P}} . \end{equation} \end{definition}

Intuitively, \(\mathcal{T}_P^{(k)}\) independently twirls each time step's input-output pair on \(\mathcal{H}_{\inp_j} \otimes \mathcal{H}_{\out_j}\), preserving the causal structure of the process tensor. In the language of higher-order quantum maps~\cite{tarantoHigherOrderQuantumOperations2025}, the multi-time Pauli twirl constitutes a time-local \emph{superprocess}, since it maps physical processes to physical processes. Operationally, it corresponds to Pauli-frame randomization protocols~\cite{figueroa-romeroOperationalMarkovianizationRandomized2024}; standard randomized compiling provides one such implementation.

We now summarize the key properties of this mapping.
\begin{proposition}[Properties]
    The map \(\mathcal{T}_P^{(k)}\) satisfies the following properties:
    \begin{enumerate}
        \item \textbf{Validity}. For any CPTP process tensor \(\Upsilon_{0:k}\), its twirl \(\mathcal{T}_P^{(k)}(\Upsilon_{0:k})\) is also a CPTP process tensor.
        \item \textbf{Projector}. \(\mathcal{T}_P^{(k)}\) is idempotent, \(\mathcal{T}_P^{(k)} \circ \mathcal{T}_P^{(k)} = \mathcal{T}_P^{(k)}\).
        \item \textbf{Causality}. The affine (causal) constraints are preserved under twirling.
    \end{enumerate}
\end{proposition}
The first two properties generalize those of standard Pauli twirling for quantum channels, while the third follows from the unitality of Pauli operators (\(\mathcal{T}_P^{(k)}[\mathbb{I}] = \mathbb{I}\) on each local pair). We now demonstrate that this projection removes all quantum temporal correlations, yielding a process-separable object corresponding to a probability distribution over Pauli trajectories.

Let \(\mathcal{P} = (P_0, P_1, \ldots, P_k)\) denote a \emph{spatiotemporal Pauli trajectory}, where each \(P_j \in \mathbb{P}^{(n)}\) is an \(n\)-qubit Pauli operator at time step \(t_j\), and \(\Pi_{P_j} = \dketbra{P_j}{P_j}\) its corresponding Pauli Choi operator. Let \(\mathbb{P}^{(n,k)}\) denote the set of all such spatiotemporal Pauli trajectories over \(k+1\) time steps, and let \(\Pr{(\mathcal{P})}\) denote the trajectory weights induced by the twirl, defined explicitly in \cref{eq:spp-probabilities}.

\begin{theorem}[Twirling yields a process-separable Pauli process]\label{thm:pt-to-spp}
    For any process tensor \(\Upsilon_{0:k}\), its twirl
    \begin{equation}
        \Upsilon_{0:k}^{\mathcal{T}_P} \coloneq \mathcal{T}_P^{(k)}(\Upsilon_{0:k})
    \end{equation}
    admits a process-separable decomposition
    \begin{equation}
        \Upsilon_{0:k}^{\mathcal{T}_P} = \sum_{\mathcal{P}\in\mathbb{P}^{(n,k)}}\Pr{(\mathcal{P})}\Upsilon_\mathcal{P}, \quad \Upsilon_\mathcal{P} = \bigotimes_{j=0}^k\Pi_{P_j}.
    \end{equation}
    Probabilities \(\Pr{(\mathcal{P})}\) form a valid probability distribution given by
    \begin{equation}\label{eq:spp-probabilities}
        \begin{gathered}
            \Pr{(\mathcal{P})} = \frac{w(\mathcal{P})}{\mathcal{N}}, 
            \quad 
            \mathcal{N} \coloneq \sum_{\mathcal{P}\in\mathbb{P}^{(n,k)}} w(\mathcal{P}),
            \\
            w(\mathcal{P}) \coloneq \Tr\bigg[\bigg(\bigotimes_{j=0}^k\Pi_{P_j}\bigg)\Upsilon_{0:k}\bigg] \geq 0.
        \end{gathered}
    \end{equation}
\end{theorem}

The resulting twirled process tensor \(\Upsilon_{0:k}^{\mathcal{T}_P}\) is \emph{process-separable} in the sense of \cref{eq:process-separable}, i.e., it is a convex mixture of temporally factorized Pauli Choi operators. Equivalently, the twirl removes all genuinely quantum temporal correlations while allowing arbitrary classical temporal correlations through the joint distribution \(\Pr{(\mathcal{P})}\). A detailed proof is provided in \cref{app:proof-pt-to-spp}.

We now formalize this class of process-separable objects under the name \emph{spatiotemporal Pauli process}.
\begin{definition}[Spatiotemporal Pauli process]\label[definition]{def:def-spp}
A \emph{spatiotemporal Pauli process} (SPP) is a process tensor with Choi representation
\begin{equation}
    \Upsilon_\mathrm{SPP} = 
    \sum_{\mathcal{P}\in\mathbb{P}^{(n,k)}}
        \Pr_\mathrm{SPP}(\mathcal{P})
        \Upsilon_\mathcal{P},
        \quad
        \Upsilon_\mathcal{P} = \bigotimes_{j=0}^k\Pi_{P_j},
\end{equation}
where \(\Pr_\mathrm{SPP}(\mathcal{P})\) is a probability distribution with \(\Pr_\mathrm{SPP}(\mathcal{P}) \geq 0\) and \(\sum_{\mathcal{P}} \Pr_\mathrm{SPP}(\mathcal{P}) = 1\).
\end{definition}
In the special case of a fully temporally independent SPP---containing neither classical nor quantum temporal correlations---the trajectory probabilities factorize as
\begin{equation}
  \Pr_\mathrm{SPP}(P_0,\dots,P_k) = \prod_{j=0}^k \Pr_j(P_j).
\end{equation}
As a result, the Choi operator \(\Upsilon_\mathrm{SPP}^{\mathrm{Markov}}\) factorizes as in \cref{eq:markovian-process}, corresponding to a sequence of temporally independent Pauli channels.

The multi-time Pauli twirl should be interpreted operationally. It does not assert that the microscopic system-environment dynamics are physically transformed into a process-separable Pauli process. Rather, under the assumed randomization protocol (e.g., randomized compiling or Pauli-frame randomization), it yields an effective description that reproduces the Pauli-basis statistics of the original dynamics. This is particularly potent for applications such as in QEC. For example, in the context of a QEC circuit, \(\Upsilon_{\mathrm{SPP}}\) can be interpreted as the effective, generally correlated, circuit-level Pauli noise model obtained by applying the multi-time twirl to the original noisy circuit.

The reduction of non-Markovian dynamics to classical temporal correlations under Pauli twirling has also been recently identified in different formalisms. Liu et al.~\cite{liuNonMarkovianNoiseSuppression2024} derived an analogous result for single-slot superchannels by applying standard channel twirling to the associated ``Choi channel'' representation. A closely related statement was also obtained and demonstrated experimentally in the context of a noise mitigation protocol~\cite{liuRealizingUniversalNonMarkovian2025}. In contrast, our derivation establishes this reduction directly within the process tensor framework and rigorously generalizes it to arbitrary multi-time processes.

While SPPs can be defined directly via~\cref{def:def-spp} without reference to an underlying quantum process, our formulation using the multi-time twirl provides a direct bridge between the physically grounded system--environment dynamics and the emergent stochastic Pauli process. This connection endows SPPs with the same rich internal structure as process tensors, enabling a straightforward and efficient tensor network construction. In particular, we show that \(\mathcal{T}_P^{(k)}\) is implemented by fixed local contractions on the physical legs of the process tensor.

\subsection{Local tensor network implementations}
We now show that the multi-time Pauli twirl admits an exact local implementation on a process tensor network. Concretely, \(\mathcal{T}_P^{(k)}\) is realized by inserting a fixed Pauli tensor on each system input-output leg at every time step. This produces (i) a twirled process tensor MPO in the operator picture, and (ii) a classical tensor encoding the joint weights (and hence probabilities) of Pauli trajectories. Centrally, this transformation acts only on the open physical indices and leaves inner (environment) bonds unchanged.

Starting from the Pauli superoperator \(P^* \otimes P\), we define the Pauli tensor \(\mathbf{P}\) via the reorder--fuse--transpose (RFT) operation introduced in \cref{sec:pt-tn} as
\begin{multline}
  \overline{P}_{(x)}^{b'}{}_{a'} P_{(x)}^b{}_a
  \xmapsto{\;\mathrm{RFT}\;}
  \mathbf{P}^\alpha{}_{\beta x}
  \\
  \;\;\longleftrightarrow\;\;
  \valignbox{\includegraphics[scale=0.9]{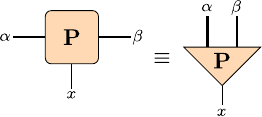}},
\end{multline}
where the index \(x\) of \(\dim{x} = 4^n\) enumerates the \(n\)-qubit Pauli operators.

\begin{proposition}[Local tensor network implementation of the multi-time Pauli twirl]\label[proposition]{prop:local-pt-twirl}
  Let \(\mathbf{T}\) be a tensor network (MPO) representation of a \(k\)-slot process tensor as in \cref{eq:pt-mpo}. Then the multi-time Pauli twirl \(\widetilde{\mathbf{T}} \coloneq \mathcal{T}_P^{(k)}(\mathbf{T})\) is obtained by contracting a pair of the Pauli tensor \(\mathbf{P}\) with each open system index pair (\(\alpha_j, \beta_j\)) at every time step as
  \begin{equation}\label{eq:pt-twirl-tensor}
    \widetilde{\mathbf{T}}^{\alpha_0\dots \alpha_k}{}_{\beta_0 \dots \beta_k} = 
    \mathbf{T}^{\alpha'_0 \dots \alpha'_k}{}_{\beta'_0 \dots \beta'_k}
    \Bigg(
    \prod_{j=0}^k
      \mathbf{P}^{\alpha'_j\alpha_j}{}_{x_j}
      \mathbf{P}^{x_j}{}_{\beta'_j\beta_j}
    \Bigg),
  \end{equation}
  where primed indices are summed over.
\end{proposition}

Graphically,
\begin{equation}
    \valignbox{\includegraphics[scale=0.8]{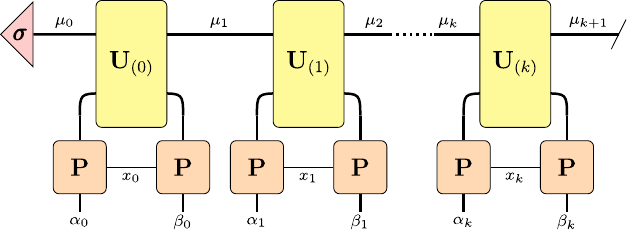}}.
\end{equation}
In particular, this operation implements the twirl locally on the physical legs and leaves the environment bonds \(\mu_j\) unchanged. This is a direct rewriting of \cref{def:multi-time-twirl} in the Pauli Liouville basis using \(\mathbf{P}\) and the RFT re-indexing in \cref{eq:rft}. The resulting tensor \(\widetilde{\mathbf{T}}\) is related to the standard Choi representation of the twirled process tensor via \cref{eq:pt-tensor-to-choi}.

The local contraction in \cref{prop:local-pt-twirl} also gives direct access to the (unnormalized) trajectory weights by contracting only one Pauli tensor per time step into the open legs of the untwirled tensor network. Let \(\mathbf{x}_{0:k} = (x_0, x_1, \ldots, x_k)\) index a spatiotemporal Pauli trajectory, where each \(x_j\) labels an \(n\)-qubit Pauli operator. The trajectory weights are defined as
\begin{equation}\label{eq:spp-weights}
  \mathbf{w}_{x_0\dots x_k} 
  \coloneq
  \mathbf{T}^{\alpha_0 \dots \alpha_k}{}_{\beta_0 \dots \beta_k} \prod_{j=0}^k \mathbf{P}^{\alpha_j}{}_{\beta_jx_j},
\end{equation}
and the corresponding normalized probabilities
\begin{equation}\label{eq:spp-prob-tensor}
  \Pr(\mathbf{x}_{0:k}) = \mathbf{p}_{x_0\dots x_k} \coloneq \frac{\mathbf{w}_{x_0\dots x_k}}{\mathcal{N}}, 
  \quad
  \mathcal{N} \coloneq \sum_{x_0\dots x_k} \mathbf{w}_{x_0\dots x_k}.
\end{equation}
This reproduces \cref{eq:spp-probabilities} under the identification \(x_j \leftrightarrow P_j\).

Although \(\mathbf{w}_{x_0\dots x_k}\) is an exponentially large tensor, it inherits an efficient MPS form from the process tensor MPO for which we can directly construct its local tensors.
\begin{definition}[Process tensor MPO to SPP MPS]\label[definition]{def:pt-mpo-to-spp-mps}
  Given a process tensor \(\mathbf T\) with Stinespring dilation \(\{\mathbf U_{(j)}, \pmb\sigma\}\), the spatiotemporal Pauli process obtained by multi-time twirling admits an MPS representation with tensors
  \begin{equation}\label{eq:spp-tensors}
  \begin{aligned}
    \mathbf{A}_{x_0\mu_{1}} 
    &= \pmb{\sigma}_{\mu_0}\mathbf{U}^{\mu_0\alpha_0}{}_{\mu_{1}\beta_0} \mathbf{P}^{\alpha_0}{}_{\beta_0 x_0} \\
    \mathbf{A}^{\mu_j}{}_{x_j\mu_{j+1}} 
    &= \mathbf{U}^{\mu_j\alpha_j}{}_{\mu_{j+1}\beta_j} \mathbf{P}^{\alpha_j}{}_{\beta_j x_j} \\
    \mathbf{A}^{\mu_k}{}_{x_k} 
    &= \mathbf{U}^{\mu_k\alpha_k}{}_{\mu_{k+1}\beta_k}\mathbf{t}^{\mu_{k+1}}
    \mathbf{P}^{\alpha_k}{}_{\beta_k x_k},
  \end{aligned}
  \end{equation}
  where the initial environment tensor \(\pmb{\sigma}\) and the final trace \(\mathbf{t}\) have been absorbed into the boundary tensors. The resulting MPS,
  \begin{equation}\label{eq:spp-mps}
    \mathbf{w}_{x_0\dots x_k} =
    \mathbf{A}_{x_0\mu_1}
    \Bigg(\prod_{j=1}^{k-1} 
    \mathbf{A}^{\mu_j}{}_{x_j\mu_{j+1} }\Bigg)
    \mathbf{A}^{\mu_k}{}_{x_k},
  \end{equation}
  encodes the (unnormalized) joint probability distribution over Pauli trajectories, denoting the exact probability tensor \(\mathbf{p}_{x_0\dots x_k} \coloneq \Pr(\mathbf{x}_{0:k}) = \mathbf{w}_{x_0\dots x_k}/\mathcal{N}\).
\end{definition}

\begin{figure}
    \centering
    \includegraphics[scale=1]{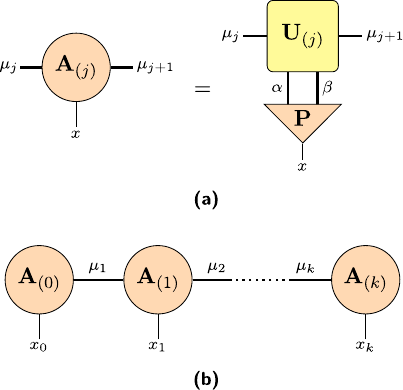}
    \caption{\textbf{One-dimensional SPP as an MPS over Pauli trajectories.} \textbf{(a)}~Local tensor obtained by contracting the system--environment unitary tensor \(\mathbf{U}_{(j)}\) with the Pauli tensor \(\mathbf{P}\), corresponding to a local Pauli twirl on the system. \textbf{(b)}~Resulting matrix product state (MPS) representation of the SPP trajectory weights, where environment bonds \(\mu_j\) mediate temporal correlations and physical indices \(x_j\) label Pauli outcomes at each time step. Sampling this MPS yields Pauli trajectories \(\mathbf{x}_{0:k}\) with probability \(\Pr(\mathbf{x}_{0:k})\).}
    \label{fig:spp-mps}
\end{figure}

Importantly, \cref{prop:local-pt-twirl} and \cref{def:pt-mpo-to-spp-mps} show that the map from a process tensor MPO to the SPP MPS acts only on the physical indices and leaves the environment bonds \(\mu_j\) unchanged. This immediately implies that the temporal bond dimension required to represent the induced SPP cannot increase.
\begin{lemma}[Bond dimension bound for SPP MPS]\label[lemma]{thm:spp-bond}
  Let \(\mathbf{T}\) be a process tensor admitting an MPO form whose inner bond dimensions are \(\dim \mu_j\) for each time step \(t_j\). The resulting SPP MPS encoding the joint probability distribution \(\Pr(\mathbf{x}_{0:k})\), defined in \cref{def:def-spp}, admits an MPS representation with inner bond dimensions \(D_j\) satisfying 
  \begin{equation}
    D_j \leq \dim \mu_j \quad \forall j.
  \end{equation}
\end{lemma}
\begin{proof}
    The tensors of the SPP MPS are obtained by applying a local, linear map to the physical indices of the process tensor MPO tensors. This operation acts strictly on the system Hilbert space and leaves the environment bonds \(\mu_j\) invariant. Since the minimum required bond dimension \(D_j\) is determined by the Schmidt rank across the temporal partition, and no operations are performed on the virtual bonds, the rank cannot increase. Therefore, the bond dimension is strictly bounded by the dimension of the original virtual space, \(D_j \le \dim \mu_j\).
\end{proof}
\begin{corollary}[Environment dimension bound]\label{thm:env-bond-dim-bound}
    For a process derived from a system interacting with a finite environment of dimension \(d_E\), the inner bond dimensions of the resulting SPP MPS are bounded by the Liouville space dimension of the environment,
    \begin{equation}
        D_j \le d_E^2.
    \end{equation}
\end{corollary}
For generic Haar-random system-environment unitaries, the bound in \cref{thm:spp-bond} is typically saturated (up to finite-size effects), since the induced temporal memory explores the full environment Liouville space.

This construction, illustrated in \cref{fig:spp-mps}, provides a direct method for sampling Pauli trajectories from the joint distribution using standard MPS sampling algorithms. Conversely, if the operator description is required, the twirled process tensor MPO \(\widetilde{\mathbf{T}}\) can be perfectly recovered from the SPP MPS by re-expanding the Pauli basis indices \(x_j\) as
\begin{equation}\label{eq:spp-mps-to-pt}
  \widetilde{\mathbf{T}}^{\alpha_0\dots \alpha_k}{}_{\beta_0 \dots \beta_k} 
  =
  \mathbf{p}_{x_0\dots x_k}\Bigg(
    \prod_{j=0}^k \mathbf{P}^{x_j\alpha_j}{}_{\beta_j}
  \Bigg).
\end{equation}
\Cref{eq:spp-weights,eq:spp-prob-tensor,eq:spp-mps-to-pt} hence provide two equivalent descriptions of the induced SPP: as a Pauli-diagonal MPO \(\widetilde{\mathbf{T}}\) in the operator picture, and as a classical MPS \(\mathbf{p}\) encoding \(\Pr(x_{0:k})\).

Finally, nothing in the above construction relies on a one-dimensional (MPS/MPO) geometry. The multi-time Pauli twirl is implemented by inserting the same fixed Pauli tensor \(\mathcal{P}\) on the open indices. Consequently, for any process tensor represented by a tensor network---for instance, a tree tensor network or a higher-dimensional PEPO---the twirl produces a tensor network encoding the classical joint trajectory weights. To illustrate the \((1+1)D\) case, the multi-time Pauli twirl maps the process tensor PEPO to a projected entangled pair state (PEPS) representation of the SPP. The constituent tensors of this PEPS network are given by
\begin{equation}\label{eq:spp-peps}
  \mathbf{a}_{(i, j)}^{\nu^{i-1}_j\mu_j^i}{}_{x_j^i\mu_{j+1}^i\nu_j^i} = 
  \mathbf{u}_{(i, j)}^{\nu^{i-1}_j\mu_j^i\alpha_j^i}{}_{\beta_j^i\mu_{j+1}^i\nu_j^i}
  \mathbf{P}^{\beta_j^i}{}_{\alpha_j^i x_j^i},
\end{equation}
As shown in \cref{fig:spp-peps}, spatial correlations are encoded along the bonds \(\nu^i_j\), while temporal correlations remain encoded along the bonds \(\mu^i_j\). The bond dimension bound of \cref{thm:spp-bond} applies individually to each temporal bond in the PEPS network.

\begin{figure}
    \centering
    \includegraphics[scale=1]{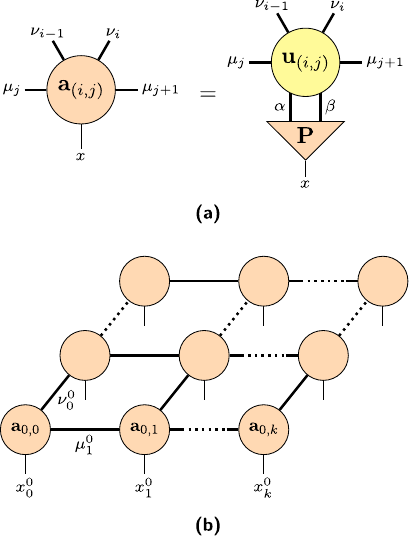}
    \caption{\textbf{\((1+1)D\) SPP as a PEPS over spatiotemporal Pauli faults.} \textbf{(a)}~Local tensor transformation for \((1+1)D\) tensor network representation. \textbf{(b)}~Projected entangled-pair state (PEPS) representation of the SPP trajectory weights. Spatial correlations are encoded along one axis by bonds \(\nu^i_j\), while temporal correlations are encoded along the other axis by bonds \(\mu^i_j\).}
    \label{fig:spp-peps}
\end{figure}

We next illustrate the SPP framework by introducing three simple system-environment Hamiltonians and explicitly computing the corresponding single-slot (two time slices) SPPs.

\subsection{Worked examples: single-qubit non-Markovian dynamics}
Consider a single system qubit \(S\) coupled to a single environment qubit \(E\), with joint Hilbert space \(\mathcal{H}_S \otimes \mathcal{H}_E\). We study three families of system-environment Hamiltonians \(H(\theta)\in\mathcal{B}(\mathcal{H}_S\otimes\mathcal{H}_E)\), generating unitary dynamics of the form
\begin{equation}
  U(\theta) = e^{-iH(\theta)},
\end{equation}
where \(\theta\) is a tunable parameter. Concretely, we consider:
\begin{enumerate}
    \item \textbf{Heisenberg interaction}
    \begin{equation}
        H_J(\theta_J) = -\frac{1}{2} \theta_J (X^S \otimes X^E + Y^S \otimes Y^E + Z^S \otimes Z^E)
    \end{equation}

    Generates exchange interactions; in particular, \(U_J(\pi/2)\) is equivalent to a SWAP operator up to a global phase. This is a paradigmatic model for non-Markovian dynamics and is known to generate strong temporal correlations.

    \item \textbf{Controlled-\(X\) rotation}
    \begin{equation}
        H_{RX}(\theta_{RX}) = -\frac{1}{2} \theta_{RX}  (X^S \otimes I^E - X^S \otimes Z^E)
    \end{equation}

    Generates a controlled rotation on the system along the \(X\)-axis, with the environment acting as the control. At \(\theta_{RX} = \pi/2\), the resulting unitary is equivalent to a CNOT gate up to a global phase.

    \item \textbf{Heisenberg with local field}
    \begin{equation}
        H_F(\theta_F) = H_J(\frac{\pi}{2}) - \theta_F(X^S + Y^S + Z^S) \otimes I^E
    \end{equation}

    A maximal exchange Heisenberg interaction modulated by a local field on the system, breaking the symmetry of the pure Heisenberg model.
\end{enumerate}

Given \(U(\theta)\) and choosing an initial environment state \(\rho_E = \ketbra{+}{+}\), we construct the corresponding single-slot (two time slices) process tensor \(\Upsilon_{0:1}(\theta)\) in the standard Choi representation using \cref{eq:pt-mpo,eq:pt-tensor-to-choi}. We then apply the multi-time Pauli twirl in \cref{eq:twirl-upsilon} to obtain the corresponding SPP \(\Upsilon_{0:1}^{\mathcal{T}_P}(\theta)\).

\begin{figure*}[t]
    \centering
    \includegraphics[width=0.95\textwidth]{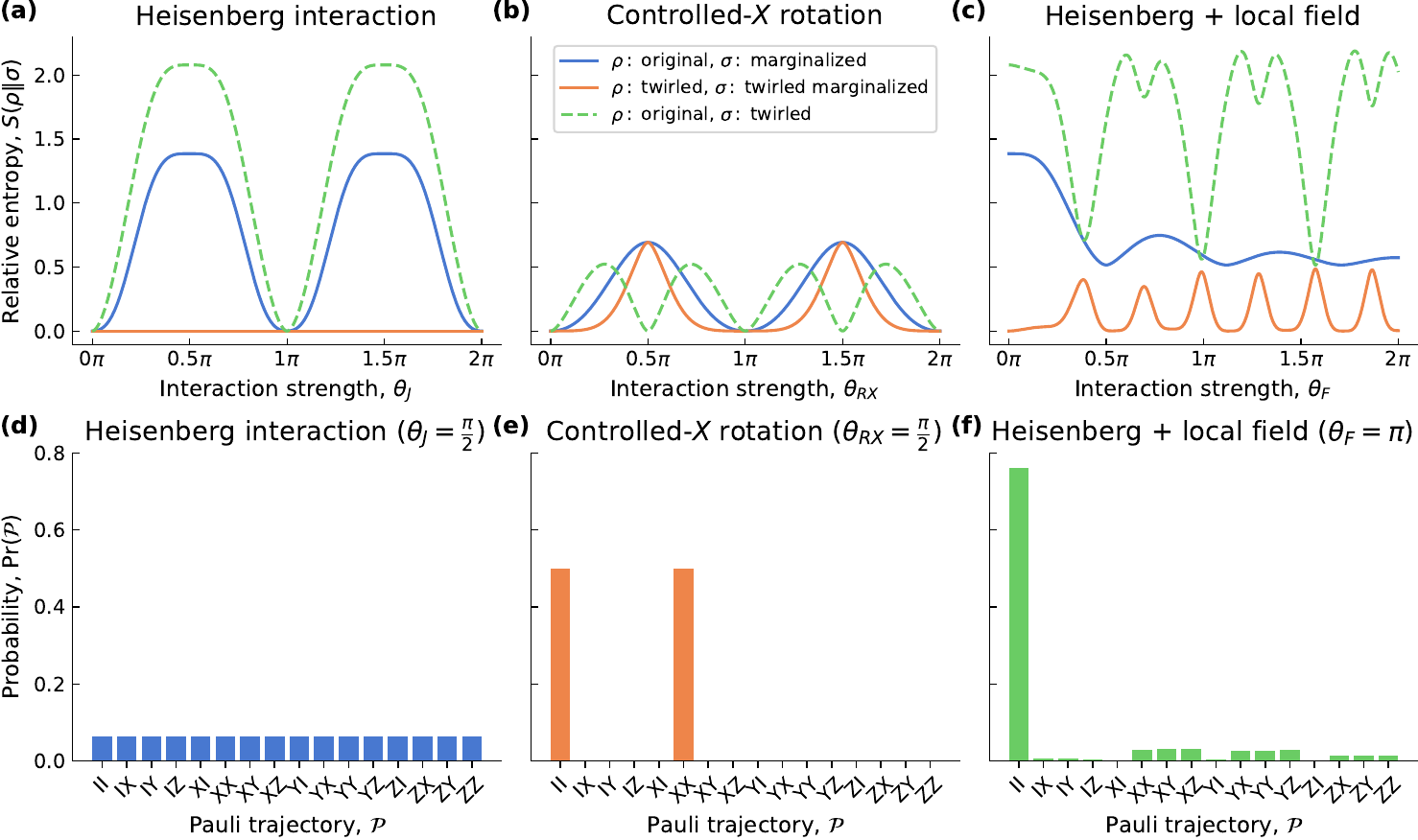}
    \caption{\textbf{Quantum relative entropy diagnostics of SPPs and Pauli trajectory probabilities}. For three single-qubit system-environment Hamiltonian families, we compare temporal correlation diagnostics before and after the multi-time Pauli twirl, together with representative SPP trajectory distributions. Top row: generalized quantum mutual information (GQMI) \(\mathcal{I}[\Upsilon_{0:1}] = S(\Upsilon_{0:1} \| \Upsilon_{0:1}^{\mathrm{Markov}})\), its twirled counterpart \(\mathcal{I}[\Upsilon_{0:1}^{\mathcal{T}_P}]\), and twirl relative entropy \(\mathcal{J}[\Upsilon_{0:1}]\), plotted as functions of interaction strength \(\theta \in [0,2\pi]\). Bottom row: Representative Pauli trajectory probabilities \(\Pr(P_{0:1})\) for selected parameter values. Columns correspond to \textbf{(a, d)} Heisenberg exchange, \textbf{(b, e)} controlled-\(X\) rotation, and \textbf{(c, f)} Heisenberg exchange with local field. Choi operators are normalized such that \(\Tr[\Upsilon_{0:1}] = 1\), and the environment is initialized in \(\ket{+}\).}
    \label{fig:worked-examples}
\end{figure*}

To study how these processes transform under twirling, we compute measures of temporal correlations alongside Pauli trajectory probabilities. A key diagnostic is the \emph{generalized quantum mutual information} (GQMI), defined as the quantum relative entropy between the original process tensor \(\Upsilon_{0:k}\) and the product of its single-time marginals
\begin{equation}
  \Upsilon_{0:k}^{\mathrm{Markov}} = \Tr_{\overline{0}}[\Upsilon_{0:k}] \otimes \Tr_{\overline{1}}[\Upsilon_{0:k}] \otimes \cdots \otimes \Tr_{\overline{k}}[\Upsilon_{0:k}],
\end{equation}
where \(\Tr_{\overline{j}}[\cdot]\) denotes the partial trace over all input-output spaces except those at time step \(t_j\). The GQMI \(\mathcal{I}\) is then
\begin{equation}
    \begin{aligned}
        \mathcal{I}[\Upsilon_{0:k}] 
        &= 
        S(\Upsilon_{0:k} \| \Upsilon_{0:k}^{\mathrm{Markov}})
        \\
        &=
        \Tr[\Upsilon_{0:k} (\log \Upsilon_{0:k} - \log \Upsilon_{0:k}^{\mathrm{Markov}})].
    \end{aligned}
\end{equation}

In addition, we introduce a second quantity \(\mathcal{J}[\Upsilon_{0:k}]\) based on the relative entropy between a process and its Pauli-twirled counterpart, quantifying the non-Pauli structure discarded by the multi-time twirl. Defining \(\Upsilon_{0:k}^{\mathcal{T}_P} = \mathcal{T}_P^{(k)}(\Upsilon_{0:k})\), we set
\begin{equation}
  \mathcal{J}[\Upsilon_{0:k}] = S(\Upsilon_{0:k} \| \Upsilon_{0:k}^{\mathcal{T}_P}).
\end{equation}
By a property of the quantum relative entropy, \(\mathcal{J}[\Upsilon_{0:k}] \geq 0\) for all \(\Upsilon_{0:k} \ge 0\), with equality if and only if \(\Upsilon_{0:k} = \Upsilon_{0:k}^{\mathcal{T}_P}\); equivalently, the process is invariant under the multi-time twirl.

We evaluate \(\mathcal{I}[\Upsilon_{0:1}]\), \(\mathcal{I}[\Upsilon_{0:1}^{\mathcal{T}_P}]\), and \(\mathcal{J}[\Upsilon_{0:1}]\) for each Hamiltonian family as functions of the interaction strength, where Choi operators \(\Upsilon_{0:1}\) are normalized such that \(\Tr[\Upsilon_{0:1}] = 1\). \Cref{fig:worked-examples} (top row) plots each quantity over \(\theta \in [0,2\pi]\). For the Heisenberg model in panel~(a), \(\mathcal{I}[\Upsilon_{0:1}]\) is maximal at \(\theta_J = (2n+1)\pi/2\), \(n \in \mathbb{Z}\), corresponding to a SWAP-like operation, and minimal for \(\theta_J = n\pi\), corresponding to the identity operation. Notably, the twirled process has \(\mathcal{I}[\Upsilon_{0:1}^{\mathcal{T}_P}] = 0\) for all \(\theta_J\), that is, due to the symmetry of the interaction, the induced SPP is always temporally uncorrelated.

To visualize the Pauli sequence probabilities directly, \cref{fig:worked-examples} (bottom row) plots \(\Pr(\mathcal{P}_{0:1})\) at representative parameter values for each model. For the Heisenberg interaction at \(\theta_J = (2n + 1)\pi/2\) in panel~(d), the resulting SPP is a uniform distribution over all 16 possible Pauli trajectories, corresponding to maximally depolarizing noise at each step; correspondingly, \(\mathcal{J}[\Upsilon_{0:1}]\) is maximal, indicating the original process is farthest from its twirled counterpart.

For the \(H_{RX}\) model in \cref{fig:worked-examples}(b), the GQMI is non-zero for both the original and twirled processes, indicating the presence of temporal correlations in the induced SPP. In particular, at \(\theta_{RX} = \pi/2\), we find \(\mathcal{I}[\Upsilon_{0:1}] = \mathcal{I}[\Upsilon_{0:1}^{\mathcal{T}_P}] = \ln 2\) and \(\mathcal{J}[\Upsilon_{0:1}] = 0\), showing that the original process is invariant under Pauli twirling and hence already an SPP. This is reflected in the trajectory distribution in \cref{fig:worked-examples}(e), where only the \(I_{t_0}I_{t_1}\) and \(X_{t_0}X_{t_1}\) sequences occur with non-zero probability, corresponding to a perfectly two-time correlated bit-flip channel. Away from these special points at \(\theta_{RX} \neq n\pi/2\), we observe that the twirling generally reduces the GQMI while \(\mathcal{J}[\Upsilon_{0:1}]\) increases, consistent with the removal of non-Pauli structure and non-Pauli-based temporal correlations.

Finally, for the local field model in \cref{fig:worked-examples}(c), we observe a non-trivial interplay between the exchange interaction and the symmetry-breaking local field. The GQMI \(\mathcal{I}[\Upsilon_{0:1}]\) is maximal at \(\theta_F = 0\) (pure Heisenberg dynamics) and then decays non-monotonically as the field strength increases. In contrast to the Heisenberg-only case, the twirled GQMI \(\mathcal{I}[\Upsilon_{0:1}^{\mathcal{T}_P}]\) is generically non-zero and displays oscillatory peaks and troughs as a function of \(\theta_F\). The twirl relative entropy \(\mathcal{J}[\Upsilon_{0:1}]\) exhibits two distinct types of troughs which alternate with peaks in \(\mathcal{I}[\Upsilon_{0:1}^{\mathcal{T}_P}]\). The trajectory probabilities in \cref{fig:worked-examples}(f) are dominated by the \(I_{t_0}I_{t_1}\) sequence, with small but correlated contributions from the Pauli sequences excluding an identity. Overall, this model highlights that SPPs can exhibit rich and structured temporal correlations, even when genuinely quantum temporal correlations are discarded.

Having established the SPP framework and illustrated it with worked examples, we now analyze temporal correlations in one-dimensional SPP MPS representations, developing a transfer operator framework to characterize memory and correlation decay.

\section{Characterizing temporal correlations in one-dimensional SPPs}\label{sec:correlation-structure}
Modeling SPPs within the tensor network formalism enables the application of established analytical and numerical tools to characterize the structure of spatiotemporal correlations. In this section, we primarily analyze the temporal correlation structure of SPPs to investigate how environmental memory is propagated and dissipated. We begin by defining multi-point correlation functions, evaluated via tensor contractions, which serve as the foundation for a transfer operator-based framework.

By examining the spectral properties of these operators---specifically the spectral gap and correlation time---we provide a compact characterization of the asymptotic decay of temporal correlations. Finally, we establish an isomorphism between SPP MPS representations and hidden Markov models (HMMs) under specific constraints. This connection provides a latent-state interpretation of the Pauli dynamics and facilitates the use of existing HMM-based methods for sampling and inference in large-scale QEC studies under correlated noise.

\subsection{Transfer operator framework}\label{sec:correlation-functions}
We begin by defining correlation functions for SPPs to quantify the statistical dependencies between Pauli operators across space and time. Since SPPs directly encode the joint probability distribution, these correlations can be straightforwardly evaluated via tensor contractions, which later motivates transfer operator formalism.

Let \(f : \mathbb{P}^{(n)} \to \mathbb{R}\) be a scalar function on the set of \(n\)-qubit Pauli operators \(\mathbb{P}^{(n)}\); for example, \(f_X(x) = 1\) if \(x=X\) and 0 otherwise. These functions play a role analogous to observables in quantum theory. Given a random Pauli operator \(x_t \in \mathbb{P}^{(n)}\) at time \(t\), the quantity \(f(x_t)\) defines a real-valued random variable. For two such variables \(f(x_t)\) and \(g(x_{t+\tau})\) evaluated at times \(t\) and \(t+\tau\), where \(\tau \geq 1\), the two-point correlation function (covariance) is defined as 
\begin{equation}
    \begin{split}
        C_{f,g}(\tau) 
        &= \mathbb{E}[(f(x_t) - \mathbb{E}[f(x_t)])(g(x_{t+\tau}) - \mathbb{E}[g(x_{t+\tau})])] \\
        &= \mathbb{E}[f(x_t) g(x_{t+\tau})] - \mathbb{E}[f(x_t)] \mathbb{E}[g(x_{t+\tau})],
    \end{split}
\end{equation}
where the expectation \(\mathbb{E}[\cdot]\) is taken with respect to a full joint probability distribution \(\Pr{(x_0,\dots,x_k)}\). At this stage, no assumptions are made regarding properties of the process such as time-homogeneity or stationarity.

In tensor network notation, each function \(f(x_t)\) and \(g(x_{t+\tau})\) corresponds to a real vector \(\mathbf{f}^{x_t}\) and \(\mathbf{g}^{x_{t+\tau}}\) indexed by Pauli labels. The expectation values can then be computed via the full tensor contraction of these vectors with the joint probability tensor \(\mathbf{p}_{x_0\dots x_k}\) (or its MPS representation) and summing over all remaining open indices. Concretely, for single and two-point expectations, we have
\begin{equation}
    \begin{aligned}
        \mathbb{E}[f(x_t)]
        &= \sum_{x_0,\dots,x_k} \mathbf{p}_{x_0\dots x_k} 
        \mathbf{f}^{x_t}, \\
        \mathbb{E}[f(x_t) g(x_{t+\tau})] 
        &= \sum_{x_0,\dots,x_k} \mathbf{p}_{x_0\dots x_k} 
        \mathbf{f}^{x_t} \mathbf{g}^{x_{t+\tau}},
    \end{aligned}
\end{equation}
respectively.

This construction extends directly to multi-point expectations by inserting the appropriate function vectors. To analyze these correlations more systematically, we now introduce a transfer operator-based framework.

Consider the MPS representation of an SPP with local tensors \(\mathbf{A}^{\mu_t}{}_{x_t\mu_{t+1}}\) where \(x_t \in \mathbb{P}^{(n)}\). In standard notation, this is equivalently the set of \(D \times D\) matrices \(A_{x_t}\) (where \(D = \dim{\mu}\) is the bond dimension) indexed by the \(x_t\), with elements given by
\begin{equation}
    (A_{x_t})_{\mu_t \mu_{t+1}} = \mathbf{A}^{\mu_t}{}_{x_t \mu_{t+1}}.
\end{equation}
Given a scalar function \(f : \mathbb{P}^{(n)} \to \mathbb{R}\), we define the following operators.
\begin{definition}[Transfer and emission operators]
    The transfer operator \(T\) and the emission operator \(E_f\) are defined as 
\begin{equation}
    T = \sum_{x_t\in\mathbb{P}^{(n)}}A_{x_t},
    \qquad
    E_f = \sum_{x_t\in\mathbb{P}^{(n)}} f(x_t) A_{x_t}.
\end{equation}
\end{definition}

The transfer operator \(T\) governs how correlations propagate along the virtual (environment) bonds. The observable-inserted transfer operator \(E_f\), which we refer to as the emission operator, encodes how an observable \(f\) couples to these bonds. A tensor graphical representation of these operators is shown in \cref{fig:transfer-emission-covariance}.

We note that the definition of the transfer operator here differs from the conventional quantum MPS transfer operator, which involves both the tensor and its complex conjugate~\cite{schollwockDensitymatrixRenormalizationGroup2005}. The distinction arises because the SPP MPS directly encodes the joint probability distribution over Pauli labels, rather than quantum amplitudes.

We now assume the SPP is \emph{time-homogeneous} (all tensors \(A_{x_t}\) are identical for all \(t\)) and \emph{ergodic}, such that the process reaches a unique stationary state on the environment space. As a result, the transfer operator \(T\) has a unique leading eigenvalue \(\lambda_1 = 1\) (after normalization), with corresponding left and right eigenvectors \(\bra{l_1}\) and \(\ket{r_1}\) satisfying
\begin{equation}
    \bra{l_1} T = \bra{l_1},
    \qquad
    T \ket{r_1} = \ket{r_1},
\end{equation}
with normalization \(\braket{l_1|r_1} = 1\). Eigenvectors \(\bra{l_1}\) and \(\ket{r_1}\) thus correspond to the left and right fixed points of the transfer operator. For such stationary processes, the expectation values can be written compactly as
\begin{equation}\label{eq:expectation-values-operator}
    \begin{gathered}
        \mathbb{E}[f(x_t)] = \braket{l_1|E_f|r_1} \eqcolon \braket{f}, \\
        \mathbb{E}[f(x_t)g(x_{t+\tau})] = \braket{l_1|E_f T^{\tau - 1} E_g|r_1} \eqcolon \braket{fg}.
    \end{gathered}
\end{equation}

To recast the covariance in terms of operators, we introduce the centered emission operator.
\begin{definition}[Centered emission operator]
    The centered emission operator \(\widetilde{E}_f\) corresponding to the function \(f\) is defined as
    \begin{equation}
        \widetilde{E}_f = E_f - \braket{f}\!T,
    \end{equation}
    satisfying \(\braket{l_1|\widetilde{E}_f|r_1} = 0\).
\end{definition}
This operator represents the deviation of the observable \(f\) around its stationary mean. Using this, the two-point covariance takes a canonical form,
\begin{equation}\label{eq:covariance_operator}
    C_{f,g}(\tau) = \braket{l_1|\widetilde{E}_f T^{\tau - 1} \widetilde{E}_g|r_1}.
\end{equation}

The covariance is thus described entirely in terms of the transfer and emission operators. It follows that the decay of temporal correlations is governed by the exponent of the transfer operator \(T\), and hence by its eigenvalue spectrum. Finally, this formulation extends naturally to multi-point correlation functions. For an ordered time sequence \(t_1 < t_2 < \dots < t_m\), the \(m\)-point correlator is
\begin{equation}
    C_{f_1,\dots,f_m} = \braket{l_1|\widetilde{E}_{f_1} T^{\Delta t_1} \widetilde{E}_{f_2} T^{\Delta t_2} \dots T^{\Delta t_{m-1}}\widetilde{E}_{f_m}|r_1},
\end{equation}
where \(\Delta t_j = t_{j+1} - t_j - 1\) is the number of intermediate time steps between \(t_j\) and \(t_{j+1}\). We now turn to the spectral properties of the transfer operator to characterize the decay of these correlations over time.

\begin{figure}
    \centering
    \includegraphics[width=\columnwidth]{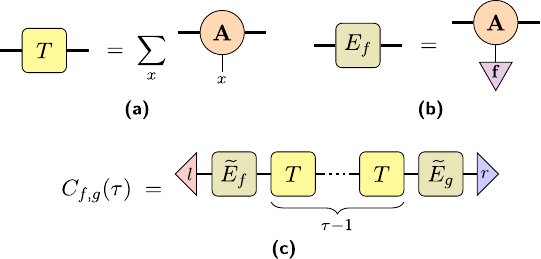}
    \caption{\textbf{Transfer operator construction for one-dimensional SPPs.} Assumes a time-homogeneous SPP MPS. \textbf{(a)} The transfer operator \(T = \sum_{x\in\mathcal{P}^{(n)}} A_x\), obtained by summing the local MPS tensors over the Pauli label \(x\). \textbf{(b)} The emission (observable-inserted) operator \(E_f = \sum_x f(x) A_x\), corresponding to insertion of a scalar function \(f:\mathcal{P}^{(n)} \to \mathbb{R}\). \textbf{(c)} Two-point covariance as a function of separation \(\tau\), expressed as \(C_{f,g}(\tau) = \bra{l_1}\widetilde{E}_f T^{\tau - 1} \widetilde{E}_g\ket{r_1}\), where \(\ket{r_1}\) and \(\bra{l_1}\) are the right and left fixed points of \(T\) and \(\widetilde{E}_f = E_f - \braket{l_1|E_f|r_1}T\) is the centered emission operator.}
    \label{fig:transfer-emission-covariance}
\end{figure}

\subsection{Spectral gap and correlation decay}\label{sec:spectral-analysis}
For clarity, we begin by assuming that the transfer operator \(T\) is diagonalizable. Its spectral decomposition is given by
\begin{equation}
    T = \sum_{i=1}^{D} \lambda_i \ketbra{r_i}{l_i},
\end{equation}
where the eigenvalues are ordered by their magnitudes \(|\lambda_1| \geq |\lambda_2| \geq \dots \geq |\lambda_{D}|\), the left and right eigenvectors satisfy \(\braket{l_i|r_j} = \delta_{ij}\), and the leading eigenvalue is normalized as \(\lambda_1 = 1\). Inserting this into the covariance function from \cref{eq:covariance_operator} yields
\begin{equation}\label{eq:covariance-spectral-expansion}
    C_{f,g}(\tau) = \sum_{i=2}^{D} \lambda_i^{\tau - 1} \! \braket{l|\widetilde{E}_f|r_i} \! \braket{l_i|\widetilde{E}_g|r}.
\end{equation}
Only the terms \(i \geq 2\) contribute since the leading eigenvectors correspond to the stationary state, for which \(\braket{l_1|\widetilde{E}_f|r_1} = \braket{l_1|\widetilde{E}_g|r_1} = 0\). This expression explicitly shows that the decay of correlations is a sum of exponentially decaying modes.

The asymptotic, long-time behavior is dominated by the subleading eigenvalue with the second-largest magnitude, which we denote as \(\lambda_* = |\lambda_2|\). From this, we define two key quantities:
\begin{enumerate}
    \item \textbf{The spectral gap} (\(\Delta\)): Defined as 
    \begin{equation}
        \Delta = 1 - \lambda_*,
    \end{equation}
    the spectral gap quantifies the rate of convergence to the stationary state. A large gap (\(\lambda_* \ll 1\)) implies correlations decay rapidly, indicating a short memory time. Conversely, a small gap (\(\lambda_* \approx 1\)) signifies long-lived temporal correlations.
    \item \textbf{The correlation time} (\(\xi\)): The characteristic timescale of the decay is given by the correlation time,
    \begin{equation}
        \xi = -\frac{1}{\ln \lambda_*}.
    \end{equation}
    In particular, for sufficiently large \(\tau\), the asymptotic decay is governed by this timescale, with \(|C_{f,g}(\tau)| \sim e^{-\tau/\xi}\).
\end{enumerate}
More broadly, the eigenspectrum of \(T\) provides a compact characterization of how environmental memory is propagated and dissipated. These spectral properties remain a useful descriptor even when relaxing assumptions such as diagonalizability and ergodicity. For instance, complex eigenvalues indicate oscillatory correlation behavior, while degenerate leading eigenvalues imply non-ergodic dynamics with multiple stationary states.

Finally, while \(T=\sum_x A_x\) is generally non-Hermitian (and may be non-normal), transfer operators induced by Haar-random system-environment unitaries are diagonalizable with probability one. More generally, when \(T\) is non-normal, eigenvalues alone need not fully characterize finite-time dynamics, and transient effects may arise; in such cases pseudospectral methods and related non-normal spectral theory may be employed~\cite{trefethenSpectraPseudospectraBehavior2005}. Next, we connect this transfer operator description to hidden Markov model realizations.

\subsection{Equivalence to hidden Markov models}
An SPP MPS, together with its associated transfer operator, naturally admits a hidden state interpretation. In this view, the virtual bond indices serve as latent variables whose dynamics drive the observed Pauli sequence. We make this connection explicit by linking SPP MPS representations to finite-state classical \emph{hidden Markov models} (HMMs). Concretely, under appropriate positivity and normalization constraints, an SPP MPS of bond dimension \(D\) is isomorphic to an \emph{edge-emitting} HMM with \(D\) hidden states. This equivalence provides a bridge to the extensive HMM literature, enabling efficient hidden state-based methods for sampling and modeling SPPs---tools we exploit in \cref{sec:temporalstormmodel,sec:qcamodel}.

Let \((S_t)_{t\geq 0}\) be a Markov chain on a finite state space \(\mathcal{S}=\{1,\dots,D\}\), and let \((X_t)_{t\geq 0}\) be a sequence of observations from the Pauli alphabet \(\mathbb{P}^{(n)}\). An edge-emitting HMM is fully specified by an initial distribution \(\pi_i\) over the hidden states and a set of transition-emission kernels \((K_x)_{ij}\). The initial distribution satisfies
\begin{equation}
    \pi_i \coloneq \Pr(S_0=i), 
    \quad 
    \pi_i \ge 0,\quad \sum_{i\in\mathcal S}\pi_i = 1.
\end{equation}
The transition-emission kernels \((K_x)_{ij}\) describe the joint probability of transitioning from hidden state \(i\) to \(j \in \mathcal{S}\) while emitting observation \(x \in \mathbb{P}^{(n)}\), defined as
\begin{equation}
    (K_x)_{ij} \coloneq \Pr{(S_{t+1}=j, X_t = x | S_t = i)},
\end{equation}
whose elements are non-negative and appropriately normalized as
\begin{equation}\label{eq:hmm-kernel-conditions}
    (K_x)_{ij} \geq 0 \quad \forall i,j,x,
    \quad
    \sum_{x,j} (K_x)_{ij} = 1 \quad \forall i.
\end{equation}
Summing over emitted symbols yields the hidden state \emph{transition matrix},
\begin{equation}
    (T_{\mathrm{HMM}})_{ij} = \sum_{x} (K_x)_{ij}.
\end{equation}
By \cref{eq:hmm-kernel-conditions}, \(T_{\mathrm{HMM}}\) is \emph{row-stochastic}, satisfying \(\sum_j (T_{\mathrm{HMM}})_{ij} = 1\) for all \(i\), and hence \(T_{\mathrm{HMM}}\ket{\mathbf{1}} = \ket{\mathbf{1}}\). If the chain is ergodic, it admits a unique stationary distribution \(\pi\) satisfying \(\pi^T T_{\mathrm{HMM}} = \pi^T\) and \(\pi^T\mathbf{1} = 1\). Accordingly, for a row-stochastic \(T_{\mathrm{HMM}}\), we may take \(\bra{l_1} = \pi^T\) and \(\ket{r_1} = \ket{\mathbf{1}}\).

The probability of observing a Pauli sequence \(x_{0:k}\) is given by a marginalization over all possible hidden state sequences,
\begin{equation}
    \Pr(x_{0:k}) = \sum_{s_0,\dots,s_{k+1}} \pi_{s_0} \prod_{t=0}^k (K_{x_t})_{s_t s_{t+1}}.
\end{equation}
This expression is in direct analogy with the SPP MPS probabilities in \cref{eq:spp-mps}. We now propose sufficient conditions under which an SPP MPS is equivalent to such an HMM.
\begin{proposition}[Sufficient conditions for SPP MPS-HMM equivalence]
    Let a time-homogeneous SPP be specified by an MPS \(\{A_x\}_{x\in\mathbb P^{(n)}}\) of bond dimension \(D\). If there exists a representation (gauge) in which \(A_x\) satisfies
    \begin{equation}
        (A_x)_{ij}\ge 0 \ \ \forall\, i,j,x, \quad
        \sum_{x\in\mathbb{P}^{(n)}}\sum_{j=1}^{D} (A_x)_{ij}=1 \ \ \forall\, i,
    \end{equation}
    then the process admits an edge-emitting Hidden Markov model (HMM) realization with \(D\) hidden states under the identification
    \begin{equation}
        (K_x)_{ij}\equiv(A_x)_{ij}.
    \end{equation}
    In particular, the SPP transfer matrix \(T = \sum_x A_x\) coincides with the HMM transition matrix \(T_{\mathrm{HMM}} = \sum_x K_x\).
\end{proposition}

These conditions are sufficient but not necessary; a generic SPP MPS need not admit a nonnegative, row-stochastic representation of the same bond dimension. However, any non-Markovian process can be rendered Markovian by enlarging the (in general, potentially unbounded) state space \cite{milzQuantumStochasticProcesses2021}. Constructing such an embedding optimally is model dependent and generally nontrivial~\cite{vidyasagarCompleteRealizationProblem2011}. Since we are primarily interested in utilizing HMMs as practical generative models for correlated noise, we do not pursue minimal HMM realizations here.

A practical advantage of HMM representations is access to standard, efficient procedures for inference, learning, and sampling~\cite{leonarde.baumMaximizationTechniqueOccurring1970}. Importantly, HMMs enable simple sequential sampling of Pauli trajectories via transition-emission update rules, making them well suited to large-scale Monte Carlo studies of QEC under correlated noise. This capability underpins the numerical analyses in subsequent sections, which we turn to next.

\section{Surface code performance under temporal correlations}\label{sec:temporalstormmodel}
In this section we present the first concrete application of the SPP framework, focusing on a simple, tunable model of temporally correlated Pauli noise and its consequences for surface code performance. We first position SPPs as practical QEC noise models and clarify how they may interface with standard circuit-level analyses. We then introduce the temporal storm model---a two-state SPP/HMM whose correlation time \(\xi\) (equivalently spectral gap \(\Delta\)) can be swept while keeping single-round marginal error rates fixed. Finally, we describe the simulation methodology and report memory and stability benchmarks, illustrating that temporal correlation structure can substantially modify logical scaling even when the average noise strength is controlled.

\subsection{SPPs as effective circuit-level noise models}
Up to this point, we have developed a general framework for spatiotemporal Pauli processes (SPPs). In principle, an SPP description of circuit noise can be obtained from a tomographically reconstructed process tensor~\cite{whiteDemonstrationNonMarkovianProcess2020,whiteUnifyingNonMarkovianCharacterization2025}. In practice, it is more useful to characterize an SPP directly at a chosen level of granularity---either at the circuit level, or phenomenologically at the level of QEC cycles. For the remainder of the paper, we focus on SPPs as generative noise models for QEC; broader uses, including decoding and noise mitigation, are discussed in \cref{sec:discussion}.

Correlations in realistic device noise can substantially alter QEC performance relative to memoryless baselines. 
Our goal here is not to catalog and deconstruct existing models, but to emphasize that a practical QEC-facing description of noise should (i) capture correlations at a level relevant to syndrome processing and decoding, while (ii) remain compatible with general underlying (Markovian or non-Markovian) dynamics. SPPs provide a convenient language for this purpose. Under randomized compiling, the multi-time Pauli twirl maps noisy stabilizer circuit dynamics to an effective SPP---a joint distribution over Pauli faults across space and time. This twirl is an operational construction: it does not assert that the microscopic dynamics are literally transformed into an SPP, but instead yields an effective model that captures the noise statistics relevant to stabilizer QEC (e.g., syndrome distributions and logical error rates) under the assumed randomization protocol.

Although the SPP framework resolves correlation at the level of individual operations, large-scale simulations and decoder interfaces often motivate a coarser representation. Accordingly, in our numerical studies we adopt a hybrid description: we treat within-round circuit noise using a standard independent gate-level model, while modeling inter-round memory effects by an SPP defined at the QEC cycle level. This choice reflects three pragmatic considerations. First, Pauli noise on a Clifford syndrome extraction circuit can be coarse-grained to an effective correlated Pauli error model on data and ancilla qubits at the end or start of each cycle~\cite{gicevFullyConvolutional3D2025}. Second, the correlations that most strongly impact long-time performance are frequently those that manifest across multiple rounds. Third, decoders and many characterization protocols naturally consume coarse-grained syndrome data, making a cycle-level SPP a convenient interface between microscopic noise and practical QEC tools.

In this section and the next, we exploit this bridge from microscopic dynamics, through SPPs, to efficiently simulable HMMs in two concrete settings. In \cref{sec:temporalstormmodel}, we introduce a simple one-dimensional temporal SPP with tunable spectral gap (equivalently, tunable correlation time) while keeping marginal error rates fixed, enabling a controlled study of how temporal correlations can impact surface code memory and stability experiments. In \cref{sec:qcamodel}, we turn to a genuinely spatiotemporal \((2+1)D\) SPP derived from a quantum cellular automaton model of environmental dynamics, which generates rich correlation structure and statistical-mechanical behavior while remaining parametrically simple.

\subsection{Temporal storm noise model}\label{sec:storm-model-details}
We now instantiate the general SPP formalism in a one-dimensional, analytically tunable setting, which we refer to as the \emph{temporal storm model}. Each system qubit is coupled to its own two-level environment with classical basis states \(\{\ket{0},\ket{1}\}\), representing \emph{calm} and \emph{storm} conditions, respectively.

\begin{figure*}[t]
    \centering
    \includegraphics[scale=1.0]{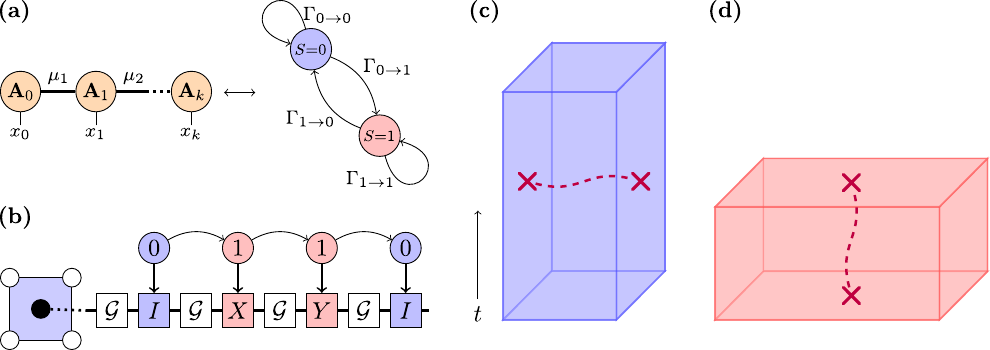}
    \caption{\textbf{Temporal storm model for surface code benchmarks.} A two-state latent environment \(s_t \in \{0,1\}\) (calm, storm) evolves as a Markov chain and modulates the Pauli label \(x_t \in \{I, X, Y, Z\}\) emitted at each time step, yielding a temporally correlated Pauli noise model. \textbf{(a)} Equivalent SPP MPS and HMM representations, with emission probabilities \(q_{s_t}^{(x)}\). \textbf{(b)} Cycle-level implementation in which each data or ancilla qubit has an independent storm environment; \(\mathcal{G}\) sets the circuit granularity. \textbf{(c, d)} Surface code memory and stability schematics with corresponding undetectable logical operators.}
    \label{fig:temporal-storm-overview}
\end{figure*}

The environment undergoes a time-homogeneous classical Markov chain with transition probabilities \(\Gamma_{s\to s'}\) between calm and storm (\cref{fig:temporal-storm-overview}). The system-environment interaction is modeled as a quantum channel conditioned on the environment state. When the environment \(E\) is in state \(\ket{s'}\), the system \(S\) evolves according to
\begin{equation}
    \mathcal{E}_{s'}[\rho_S]
    = \sum_x q^{(x)}_{s'} \,\sigma^{(x)} \rho_S \sigma^{(x)},
\end{equation}
where \(\sigma^{(x)}\) denotes the Pauli operators indexed by \(x\), and \(q^{(x)}_{s'}\) are probabilities satisfying \(\sum_x q^{(x)}_{s'} = 1\). The joint system-environment evolution over a single time step is thus described by the CPTP map
\begin{equation}
    \begin{aligned}
        \Phi[\rho_{SE}]
        &= \sum_{s,s',x} K^{(x)}_{s's} \,\rho_{SE}\, K^{(x)\dagger}_{s's},
        \\
        K^{(x)}_{s's}
        &= \sqrt{\Gamma_{s\to s'}\,q^{(x)}_{s'}}\;
        \ketbra{s'}{s}_E \otimes \sigma^{(x)}_S.
    \end{aligned}
\end{equation}

Because the environment dynamics are classical and the conditional system channels \(\mathcal{E}_{s'}\) are Pauli, the corresponding process tensor is invariant under the multi-time Pauli twirl in \cref{eq:twirl-upsilon}. In other words, this construction directly specifies an SPP, and may be viewed as the twirled image of a family of compatible underlying process tensors.

Labelling the environment transition probabilities as \(a = \Gamma_{0\to 1}\) and \(b = \Gamma_{1\to 0}\), the corresponding transfer operator \(T\) and SPP matrices \(A^{(x)}\) are
\begin{equation}
    T
    =
    \begin{pmatrix}
        1 - a & a \\
        b & 1 - b
    \end{pmatrix},
    \qquad
    A^{(x)} = T\,\mathrm{diag}\big(q_0^{(x)}, q_1^{(x)}\big).
\end{equation}
This is precisely a two-state HMM where \(T\) describes the evolution of the latent state, and the Pauli emission label \(x_t\) is drawn from \(q_{s_t}^{(x)}\). The spectrum of \(T\) is given by \(\lambda_1 = 1, \lambda_2 = 1 - a - b\). The subleading eigenvalue \(\lambda_2\) sets the spectral gap \(\Delta\) and correlation time \(\xi\),
\begin{equation}\label{eq:storm-spectral-gap}
    \Delta = a + b, \quad \xi = -\frac{1}{\ln(1-a-b)},
\end{equation}
where we assume \(0 \le a\), \(b \le 1\) with \(a + b < 1\), so that \(\lambda_2 \in (0, 1)\) and \(\xi > 0\).

The left and right eigenvectors associated with \(\lambda_1\), corresponding to the left and right stationary states, are
\begin{equation}\label{eq:storm-stationary-states}
    \bra{l} =
    \begin{pmatrix}
        \pi_0 & \pi_1
    \end{pmatrix},
    \qquad
    \ket{r} =
    \begin{pmatrix}
        1 \\
        1
    \end{pmatrix},
\end{equation}
where \(\pi_0 = b/(a+b)\) and \(\pi_1 = a/(a+b)\) are the stationary calm and storm fractions, respectively, and we have dropped the subscript 1 in \(l_1, r_1\). These immediately lead to the marginal Pauli error rates,
\begin{equation}\label{eq:storm-marginals}
    \bar{p}^{(x)}
    = \pi_0 q_0^{(x)} + \pi_1 q_1^{(x)} = \frac{bq_0^{(x)} + aq_1^{(x)}}{a+b}.
\end{equation}

By solving \cref{eq:storm-spectral-gap,eq:storm-marginals} for the transition rates \(a\) and \(b\), we construct a family of temporally correlated noise models with tunable spectral gap \(\Delta\) (or correlation time \(\xi\)) while maintaining fixed marginal error rates \(\bar{p}^{(x)}\). This decoupling of environmental memory from average noise strength underpins the parameter sweeps in the surface code simulations that follow.

\subsection{Simulation methodology}
\begin{figure*}[t]
    \centering
    \includegraphics[]{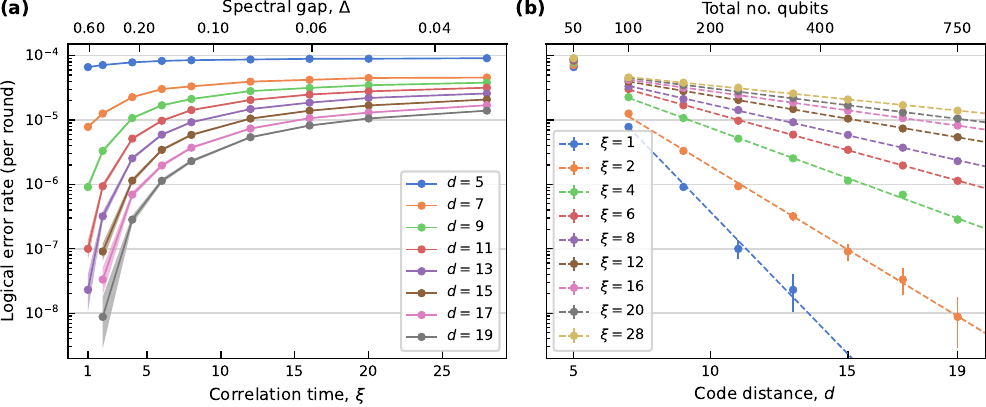}
    \caption{\textbf{Surface code memory versus correlation time.} Fixed marginal error rate \(p = 0.1\%\), variable code distance \(d\), and \(N_r = 3d\) number of rounds. \textbf{(a)} Effective per-round logical error rate versus correlation time \(\xi\) (bottom axis), with corresponding spectral gap shown on the top axis. The per-round error rate \(p_\mathrm{round}\) is inferred from the per shot failure rate \(p_\mathrm{shot}\) via \(p_\mathrm{round} = (1 - (1 - 2p_\mathrm{shot})^{1/N_r})/2\). \textbf{(b)} The same data plotted versus code distance \(d\) at fixed \(\xi\) (legend). At fixed marginal error rates, increasing temporal memory raises the logical error rate and progressively weakens its exponential suppression with distance. Error bars, where visible, indicate binomial standard errors.}
    \label{fig:temporal-storm-memory}
\end{figure*}

We benchmark surface code performance under the temporal storm noise model using two standard protocols: a memory experiment, which tests preservation of a logical observable through time, and a stability experiment, which tests the reliability of moving a logical observable through space~\cite{gidneyStabilityExperimentsOverlooked2022} (an essential primitive, e.g., for lattice surgery-based logical operations). Memory and stability experiments may be viewed as complementary space-time duals, and together they benchmark distinct but equally vital primitives of surface code logical computation. This symmetry is particularly relevant for correlated noise. In memory experiments, correlated data qubit faults can form spacelike chains implementing a logical operator, whereas in stability experiments, temporally correlated syndrome measurement (or ancilla-induced) faults can form timelike chains spanning many rounds, likewise inducing a logical operator.

For the memory experiment, we simulate rotated surface code circuits of distance \(d\) for \(N_r = 3d\) rounds of syndrome extraction, and measure logical failure probability \(p_\mathrm{shot}\) after decoding using minimum-weight perfect matching. We estimate an effective per-round error rate from the per shot failure probability as \(p_\mathrm{round} = (1 - (1 - 2p_\mathrm{shot})^{1/N_r})/2\). For the stability experiment, we simulate a diameter-6 surface code patch with a variable number of rounds \(N_r\) (setting the timelike distance), and report the logical failure probability per shot. In the parameter sweeps we use \(d = 5, 7, \dots, 19\) for memory, and \(N_r = 5, 10, \dots, 35\) for stability.

Noise is implemented in a hybrid manner following the modeling choice of \cref{sec:storm-model-details}. We apply a global independent, gate-level baseline noise model consisting of: a single-qubit depolarizing channel after every single-qubit Clifford gate and on each data qubit at the start of each round, a two-qubit depolarizing channel after each two-qubit entangling gate, a Pauli \(X\) flip (\(Z\) flip) channel before \(Z\)-basis (\(X\)-basis) measurements, and after \(Z\)-basis (\(X\)-basis) state preparations. All baseline noise processes are parameterized by a single characteristic error rate \(p = 0.1\%\).

Inter-round temporal correlations are incorporated by composing, at the start of each round, an additional Pauli fault on every data and ancilla qubit sampled from the temporal storm SPP described above. The circuit granularity \(\mathcal{G}\) in \cref{fig:temporal-storm-overview} is therefore taken to be an entire syndrome extraction subcircuit per round, though the construction is valid for finer granularities. Throughout the sweep we fix the marginal error rate of the chosen nontrivial Pauli fault to \(\bar{p}^{(x)} = p = 0.1\%\), and vary correlation time over \(\xi \in [1, 28]\). The initial state of the environment is sampled from the stationary distribution in \cref{eq:storm-stationary-states}.

Simulations are performed using Stim~\cite{gidneyStimFastStabilizer2021} for stabilizer circuit sampling, with temporally correlated errors generated externally from the temporal storm HMM. Decoding is performed using minimum-weight perfect matching via PyMatching~\cite{higgottSparseBlossomCorrecting2025} on a detector error model constructed from the independent circuit-level noise with marginal error rates matched to the correlated process. For each parameter point we estimate failure probabilities from Monte Carlo sampling over \(10^7\) shots; error bars, where visible, denote binomial standard errors.

\subsection{Surface code memory and stability benchmarks}

\begin{figure*}[t]
    \centering
    \includegraphics[width=0.95\textwidth]{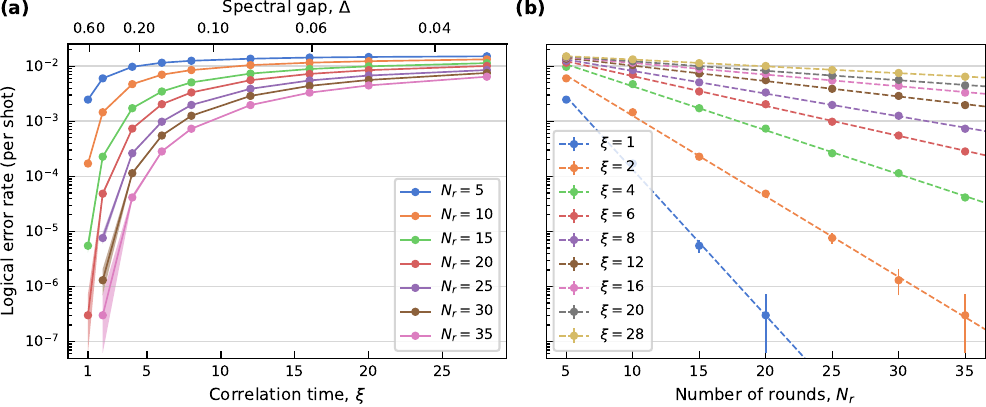}
    \caption{\textbf{Surface code stability versus correlation time.} Stability experiments benchmark the reliability of moving a logical observable through space, an essential primitive for logical operations implemented via lattice surgery. We fix the marginal error rate \(p=0.1\%\) for a diameter-6 surface code patch with variable number of rounds \(N_r\). \textbf{(a)} Logical failure rate per shot versus correlation time \(\xi\) (bottom axis), with the corresponding spectral gap shown on the top axis. \textbf{(b)} The same data plotted versus number of rounds (timelike distance) \(N_r\) at fixed \(\xi\). For short-memory noise, increasing \(N_r\) produces strong exponential suppression of logical failures; as \(\xi\) grows, the curves flatten, showing that temporal correlations erode the timelike protection relevant to lattice surgery-style logical measurements. Dashed lines represent exponential fits, and error bars indicate binomial standard error.}
    \label{fig:temporal-storm-stability}
\end{figure*}

\Cref{fig:temporal-storm-memory,fig:temporal-storm-stability} summarize surface code performance under the temporal storm model as correlation time \(\xi\) is increased at fixed marginal error rates \(\bar{p}^{(x)}\). Across both benchmarks, increasing \(\xi\) systematically degrades logical performance, highlighting that correlation structure in physical noise alone can substantially alter logical behavior.

We first consider the memory experiment. \Cref{fig:temporal-storm-memory}(a) displays the per-round logical error rate \(p_\mathrm{round}\) as a function of \(\xi\) for several code distances. For each \(d\), \(p_\mathrm{round}\) increases monotonically with \(\xi\), with the most pronounced shift occurring as the model crosses from nearly memoryless dynamics (\(\xi = 1\)) into moderate-to-long memory (\(1 < \xi \le 5\)). For larger \(\xi\), the curves tend to level off, consistent with a ``bursty'' noise picture in which memory concentrates faults into rarer but more damaging multi-round clusters. The \(d = 5\) data trends systematically above the higher distance curves and are therefore omitted from the fit in panel~(b). At the baseline error rate of \(p=0.1\%\), the distance-five code provides comparatively weak error suppression even in the low memory regime, so increasing \(\xi\) degrades performance without exhibiting the same clean scaling observed at larger distances.

\Cref{fig:temporal-storm-memory}(b) presents the same data as a performance-versus-distance scaling plot at fixed \(\xi\). The logical error rate remains approximately exponentially suppressed with increasing \(d\), but the decay exponent decreases in magnitude markedly as \(\xi\) grows. In this sense, temporal correlations blunt, though do not eliminate, the benefit of distance scaling, consistent with an effective spacetime distance reduction. From a resource perspective, the qubit overhead required to achieve a given logical error rate is highly sensitive to growth in correlation time, for example, from \(\xi=1\) to \(\xi=2\).

We next turn to the stability experiment, which provides a complementary probe of how temporal structure affects the reliability of surface code computational primitives. \Cref{fig:temporal-storm-stability} plots the logical failure probability per shot as a function of \(N_r\) for several \(\xi\). The same qualitative trend emerges---performance degrades markedly as \(\xi\) increases---but the stability protocol is a particularly direct diagnostic of the impact of temporal structure since logical failure is directly associated with timelike error chains. This is reflected in \cref{fig:temporal-storm-stability}(b), where the suppression of logical failure with increasing \(N_r\) progressively weakens as \(\xi\) grows, indicating that temporal correlations erode the benefit of increasing timelike distances. The close qualitative similarity to the memory results in \cref{fig:temporal-storm-memory} aligns with the space-time dual interpretation of stability experiments. These results also underscore the broader importance of studying correlated noise, since correlations can non-trivially interact with primitives of logical computation even under Pauli-frame randomization.

Throughout, decoding is performed using minimum-weight perfect matching with a detector error model constructed from the independent circuit-level noise, with marginal rates matched to the correlated process. The resulting performance degradation with increasing \(\xi\) therefore reflects, in part, decoder mismatch under standard assumptions, and motivates correlation-aware decoding or mitigation strategies leveraging the SPP framework. We discuss the matter further alongside other applications in \cref{sec:discussion}.

\section{Spatiotemporal noise via a quantum cellular automaton}\label{sec:qcamodel}
While \cref{sec:temporalstormmodel} demonstrated the impact of purely temporal correlation, noise in contemporary quantum hardware frequently exhibits genuinely spatiotemporal structure. Here we introduce a spatiotemporal SPP model derived from a two-dimensional quantum cellular automaton (QCA) environment (bath) coupled locally to the system, as schematically summarized in \cref{fig:qca-overview}. Operational Pauli twirling on the system maps the underlying coherent bath dynamics to an effective SPP that admits an exact description as a probabilistic cellular automaton (PCA) with a nonlinear transition kernel. We show that tuning a single parameter controlling the strength of coherent bath interactions drives the bath into a pseudo-critical regime with macroscopic error avalanches and critical slowing down, producing correlated fault patterns reminiscent of correlated error events reported in solid-state platforms. The model thus provides a microscopically motivated, yet tractably simulable, link between non-equilibrium statistical mechanics and QEC performance under spatiotemporal correlations.

We structure the analysis as follows. In \cref{sec:qca-details,sec:pca-details}, we present the microscopic QCA model, and its mapping, under operational system twirling, to an effective SPP with an exact PCA description. In \cref{sec:qca-stat-mech}, we characterize the resulting PCA as a statistical-mechanical system, isolating how the spatial parameter controlling coherent bath interactions drives a critical-like transition in the bath dynamics. Finally, in \cref{sec:qca-qec-results}, we benchmark rotated surface code memory under this noise model and show a sharp breakdown of distance scaling in the pseudo-critical regime.

\begin{figure*}
    \centering
    \includegraphics[width=0.95\textwidth]{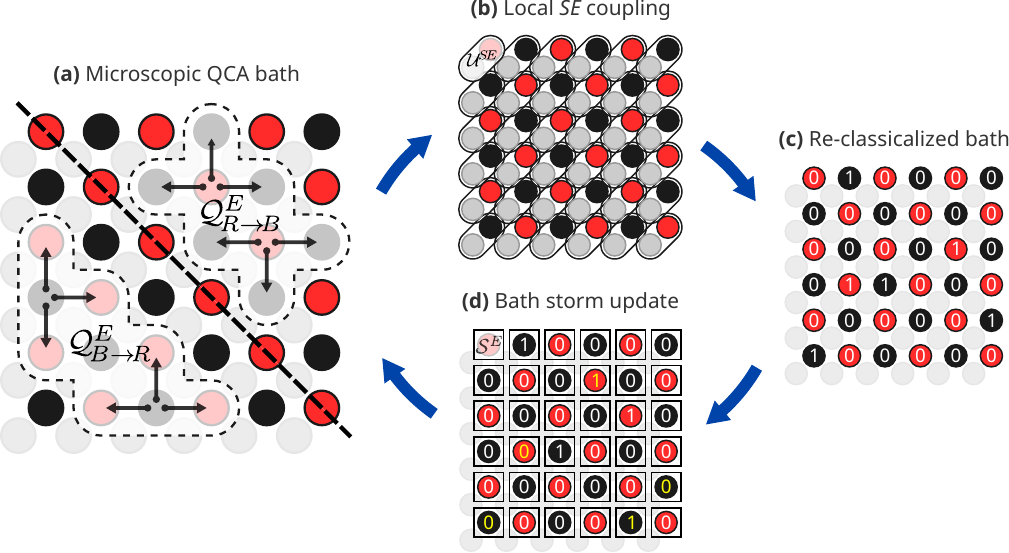}
    \caption{\textbf{QCA-derived spatiotemporal Pauli process.} Schematic of the microscopic quantum cellular automaton (QCA) bath and its effective probabilistic cellular automaton (PCA) description under system Pauli twirling. \textbf{(a)} The environment consists of bath qubits on a checkerboard lattice, with coherent QCA half-steps \(\mathcal{Q}^E_{R\to B}(\theta)\) and \(\mathcal{Q}^E_{B\to R}(\theta)\) propagating excitations between the red and black sublattices. \textbf{(b)} Each bath qubit couples locally to a corresponding system qubit through the controlled interaction \(\mathcal{U}^{SE}\). \textbf{(c)} After the operational system twirl, the interaction induces an effective \(Z\)-basis dephasing on the bath, re-classicalizing the environment into configurations \(s_t(i)\in\{0,1\}\). \textbf{(d)} The bath then undergoes an independent storm update \(\mathcal{S}^E(a, b)\), with local transitions \(0 \to 1\) and \(1 \to 0\) occurring with probabilities \(a\) and \(b\), respectively. Iterating these steps yields a classical PCA hidden Markov model whose active bath sites emit local Pauli faults, generating the spatiotemporally correlated noise used in the surface code simulations.}
    \label{fig:qca-overview}
\end{figure*}

\subsection{Microscopic QCA dynamics}\label{sec:qca-details}
We define the microscopic dynamics of the QCA noise model on a two-dimensional square lattice \(V\). At each site \(i\in V\), a system qubit \(S_i\) is locally coupled to an environment (bath) qubit \(E_i\). We impose a checkerboard bipartition of the lattice \(V = V_\mathrm{red} \cup V_\mathrm{black}\), such that nearest neighbors \(N(i)\) of a red site are strictly black, and vice versa. The environment is initialized as a classical mixture of \(Z\)-basis states, completely uncorrelated with the system.

The joint single-cycle CPTP map, illustrated in \cref{fig:qca-overview}, is defined as the composition of a local environment storm channel (as in \cref{sec:temporalstormmodel}), a two-layer coherent QCA update, and a system-environment interaction. Formally, we define the joint map as
\begin{equation}
    \mathcal{M}_{a, b, \theta} \coloneq \mathcal{U}^{SE} \circ \mathcal{Q}_{B\to R}^E(\theta) \circ \mathcal{Q}_{R\to B}^E(\theta) \circ \mathcal{S}^E (a, b).
\end{equation}
We detail the constituent maps below.
\begin{enumerate}
    \item \textbf{Environment storm channel}, \(\mathcal{S}^E(a, b)\)
    
    Each environment qubit undergoes an independent storm process flipping each qubit from \(\ket{0}\to\ket{1}\) with probability \(a\), and from \(\ket{1}\to\ket{0}\) with probability \(b\). The map for each site \(i\) is
    \begin{equation}
        \mathcal{S}^{(i)}(a, b)[\rho] = \sum_k L_k \rho L_k^\dagger,
    \end{equation}
    with Kraus operators,
    \begin{equation}
        \begin{gathered}
            L_0 = \sqrt{a}\ketbra{1}{0}, 
            \quad L_1 = \sqrt{1-a}\ketbra{0}{0}, 
            \\
            L_2 = \sqrt{b}\ketbra{0}{1}, 
            \quad L_3 = \sqrt{1-b}\ketbra{1}{1}.
        \end{gathered}
    \end{equation}
    The global channel is the product over sites, \(\mathcal{S}^E(a, b) = \bigotimes_{i\in V} \mathcal{S}^{(i)}(a, b)\).

    \item \textbf{QCA half-step (red \(\to\) black)}, \(\mathcal{Q}_{R\to B}^E(\theta)\)

    The QCA update is implemented as two layers of mutually commuting controlled-\(X\) rotations with the red and black sublattices as alternating controls and targets. For a directed environment edge \(i\to j\), define the local controlled-\(X\) rotation
    \begin{equation}
        Q^E_{i\to j}(\theta) = \ketbra{0}{0}^{E_i} \otimes I^{E_j} + \ketbra{1}{1}^{E_i} \otimes e^{-i\theta X^{E_j}}.
    \end{equation}
    The global unitary for the first half-step with red controls and black targets is then
    \begin{equation}
        Q^E_{R\to B}(\theta)
        = 
        \prod_{i\in V_\mathrm{red}} \prod_{j\in N(i)}
        Q^E_{i\to j}(\theta),
    \end{equation}
    and the corresponding channel denoted as
    \begin{equation}
        \mathcal{Q}_{R\to B}^E(\theta)[\rho] = Q^E_{R\to B}(\theta) \rho (Q^E_{R\to B}(\theta))^\dagger.
    \end{equation}

    \item \textbf{QCA half-step (black \(\to\) red)}, \(\mathcal{Q}_{B\to R}^E(\theta)\)
    
    For the second half-step with black controls and red targets, we define the global unitary
    \begin{equation}
        Q^E_{B\to R}(\theta) = \prod_{i\in V_\mathrm{black}} \prod_{j\in N(i)}
        Q^E_{i\to j}(\theta),
    \end{equation}
    with corresponding channel
    \begin{equation}
        \mathcal{Q}_{B\to R}^E(\theta)[\rho] = Q^E_{B\to R}(\theta) \rho (Q^E_{B\to R}(\theta))^\dagger.
    \end{equation}
    Together, these two maps implement a full QCA update on the environment.

    \item \textbf{Environment-system unitary}, \(\mathcal{U}^{SE}\)
    
    Each environment qubit \(E_i\) couples locally to the system qubit \(S_i\) via a controlled unitary interaction. We define
    \begin{equation}\label{eq:qca-se-unitary}
        \begin{gathered}
            U_{SE}^{(i)} = \sum_{k=0}^1 \ketbra{k}{k}^{E_i} \otimes V_k^{S_i}, 
            \quad U_{SE} = \bigotimes_{i \in V} U_{SE}^{(i)},
            \\
            \mathcal{U}^{SE}[\rho] = U_{SE} \rho U_{SE}^\dagger.
        \end{gathered}
    \end{equation}
    The key structural assumption is that the conditional unitaries \(\{V_k\}\) are orthogonal under the Hilbert-Schmidt inner product, 
    \begin{equation}\label{eq:se-unitary-condition}
        V_k V_k^\dagger = I \forall k, \quad \Tr(V_k V_l^\dagger) = 0 \ \forall k \ne l.
    \end{equation}
    Furthermore, in this work, we specialize to the minimal family of unitaries
    \begin{equation}\label{eq:qca-system-unitary}
        V_0 = I, \quad V_1 = \vec{n} \cdot \vec{\sigma} = n_X X + n_Y Y + n_Z Z,
    \end{equation}
    with \(\|\vec{n}\|^2 = 1\) and \(\vec{n} \in \mathbb{R}^3\).
\end{enumerate}
The entire dynamics are thus governed by three parameters: two temporal parameters \((a, b)\) specifying the environment storm channel \(\mathcal{S}^E (a, b)\), and a third spatial parameter \(\theta\) controlling the coherent QCA on the environment.

To obtain the effective SPP, we apply the multi-time Pauli twirl \(\mathcal{T}_P^S\) on \(\mathcal{M}_{a, b, \theta}\) exclusively on the system subspace. Because the maps \(\mathcal{S}^E (a, b)\) and \(\mathcal{Q}_{B\to R}^E(\theta),\) \(\mathcal{Q}_{R\to B}^E(\theta)\) act trivially on the system, the twirl can be `pushed through' to act solely on the interaction \(\mathcal{U}^{SE}\). Defining the twirled interaction as \(\widetilde{\mathcal{U}}^{SE} \coloneq \mathcal{T}_P^S[\mathcal{U}^{SE}]\), we write the effective twirled joint map
\begin{equation}
    \begin{aligned}
        \widetilde{\mathcal{M}}_{a, b, \theta} 
        \coloneq &
        \mathcal{T}_P^S[\mathcal{M}_{a, b, \theta}]
        \\
        = &
        \widetilde{\mathcal{U}}^{SE} \circ \mathcal{Q}_{B\to R}^E(\theta) \circ \mathcal{Q}_{R\to B}^E(\theta) \circ \mathcal{S}^E (a, b).
    \end{aligned}
\end{equation}

The central mechanism of this model is that, for a system-environment interaction of the controlled-unitary form in \cref{eq:qca-se-unitary}, satisfying the orthogonality condition in \cref{eq:se-unitary-condition}, the system twirl induces an effective \(Z\)-basis dephasing on the environment in each cycle. Explicitly, for each site \(i\) the reduced environment update obeys
\begin{equation}\label{eq:reduced-environment-update}
    \Tr_{S_i}(\widetilde{\mathcal{U}}^{SE}[\rho_{S_iE_i}]) = \sum_{k\in\{0,1\}} \ketbra{k}{k}_{E_i} \Tr_{S_i}[\rho_{S_iE_i}] \ketbra{k}{k}_{E_i},
\end{equation}
which removes the off-diagonal components of the environment state in the computational basis. Crucially, we do not impose a dephasing channel on the bath by hand. Rather, this re-classicalization of the environment emerges as a direct consequence of interacting with the operationally twirled system.

This construction yields three key features of the effective dynamics:
\begin{enumerate}
    \item \textbf{Nonlinear transition kernel}: The QCA step reduces to a PCA with a nonlinear transition kernel, with the probability of a target site flipping given by \(p = \sin^2(k\theta)\) as a function of the number of \(k\) excited neighbors.
    \item \textbf{Environment re-classicalization}: The twirled interaction applies an effective \(Z\)-basis dephasing to the environment every cycle, re-classicalizing the bath and collapsing any coherent superpositions.
    \item \textbf{System Pauli noise}: Conditioned on the environment state, the induced system noise is a Pauli channel with relative weights set by the coefficients \(n_X^2, n_Y^2, n_Z^2\) of the conditional unitary \(V_1\).
\end{enumerate}
Altogether, coherent local bath dynamics combined with operational system twirling produce a microscopically derived spatiotemporal noise model that is local and causal by construction, and that admits an exact description as a PCA hidden Markov model (HMM).
Full derivations of the mapping and the properties above are given in \cref{app:qca-pca-mapping}. We now summarize the resulting PCA HMM dynamics.

\subsection{Effective probabilistic cellular automaton}\label{sec:pca-details}
The effective spatiotemporal SPP admits an exact description as a classical HMM with a nonlinear transition kernel, where the nonlinearity is inherited from coherent interactions in the underlying QCA. The latent state is the environment bath configuration, and the emissions are the system Pauli faults.

Let \(s_{t}(i) \in \{0,1\}\) denote the state of environment site \(i \in V\) at the beginning of cycle \(t\). The Pauli emission pattern during cycle \(t\) is denoted by the string \(x_t \in \{x_t(i)\}_{i \in V}\), where each \(x_t(i) \in \{I, X, Y, Z\}\) is sampled conditional on the updated bath configuration \(s_{t+1}(i)\) at the end of the cycle.

Crucially, because the lattice is bipartite, each PCA half-step updates only one sublattice, and the update probability at a site depends only on the (fixed) configuration of the opposite sublattice. Consequently, conditioned on the opposite sublattice, updates are independent across sites within the active sublattice. This eliminates update collisions and renders the model highly parallelizable.

\begin{enumerate}
    \item[0.] \textbf{Environment initialization}: Choose an initial distribution \(\pi_0\) over configurations \(s_0\in\{0,1\}^{|V|}\) and sample \(s_{0}\sim\pi_0\).

    \item \textbf{Storm update}: Given \(s_{t}\), generate the intermediate configuration \(s_{t'}\) by updating each site \(i\in V\) independently as
    \begin{itemize}
        \item if \(s_{t}(i) = 0\), set \(s_{t'}(i) \gets 1\) with probability \(a\);
        \item if \(s_{t}(i) = 1\), set \(s_{t'}(i) \gets 0\) with probability \(b\);
        \item else set \(s_{t'}(i) \gets s_{t}(i)\).
    \end{itemize}
    \item \textbf{PCA half-step (black)}: Generate the next intermediate state \(s_{t''}\) by updating each black site \(i\in V_\mathrm{black}\), conditioned on the red configuration as
    \begin{itemize}
        \item for each \(i \in V_\mathrm{black}\), compute the number of excited red neighbors \(k(i) = \sum_{j \in N(i)} s_{t'}(j)\);
        \item set \(s_{t''}(i) \gets s_{t'}(i) \oplus 1\) with probability \(\sin^2(k(i)\theta)\), else set \(s_{t''}(i) \gets s_{t'}(i)\);
        \item for each \(i \in V_\mathrm{red}\), set \(s_{t''}(i) \gets s_{t'}(i)\).
    \end{itemize}
    \item \textbf{PCA half-step (red)}: Complete the PCA step to generate \(s_{t+1}\) as
    \begin{itemize}
        \item for each \(i \in V_\mathrm{red}\), compute the number of excited black neighbors \(k(i) = \sum_{j \in N(i)} s_{t''}(j)\);
        \item set \(s_{t+1}(i) \gets s_{t''}(i) \oplus 1\) with probability \(\sin^2(k(i)\theta)\), else set \(s_{t+1}(i) \gets s_{t''}(i)\);
        \item for each \(i \in V_\mathrm{black}\), set \(s_{t+1}(i) \gets s_{t''}(i)\).
    \end{itemize}
    \item \textbf{Pauli emission}: For each site \(i\in V\), sample the system Pauli emissions \(x_t\) as
    \begin{itemize}
        \item if \(s_{t+1}(i) = 0\), sample \(x_t(i)=I\);
        \item if \(s_{t+1}(i) = 1\), sample \(x_t(i) \in \{X, Y, Z\}\) with probabilities \(n_X^2, n_Y^2, n_Z^2\), respectively.
    \end{itemize}
    The resulting Pauli string \(x_t = \{x_t(i)\}_{i\in V}\) is the noise induced on the system during cycle \(t\).
    \item \textbf{Repeat}: Increment \(t\), and repeat steps 1-4 for as many cycles as required.
\end{enumerate}

The PCA describes a driven classical spin system with a nonlinear update rule. The storm process continuously injects and removes excitations locally, while the QCA-derived PCA spreads excitations across the lattice in a neighbor-dependent manner parameterized by \(\theta\). Importantly, the nonlinear transition kernel \(p = \sin^2(k\theta)\) is a uniquely quantum footprint of the coherent bath effects. The competition between these stochastic perturbations and nonlinear spatial propagations gives rise to non-trivial steady states and spatiotemporal correlation structures, which we evaluate next.

\subsection{Criticality and spatiotemporal correlations}\label{sec:qca-stat-mech}
Before studying the noise properties of this model for QEC, we first analyze the PCA dynamics from \cref{sec:pca-details} as a statistical-mechanical system. In particular, we isolate the effect of the spatial parameter \(\theta\), which controls the strength of the nonlinear PCA update via kernel \(p = \sin^2(k\theta)\), on the emergent correlation structure and steady-state properties of the system.

Throughout, we fix the storm parameters to a low spontaneous injection rate \(a=10^{-4}\) and a fast local relaxation rate \(b=0.5\), and sweep the controlled-rotation angle \(\theta \in [0, \pi]\). Pauli emission weights are taken to be uniform (unbiased) with \(n_X = n_Y = n_Z = 1/\sqrt{3}\). We consider lattices corresponding to rotated surface code layouts for distances \(d=9, 11, 13, 15\), with the number of sites (qubits) \(N_d = |V| = 2d^2 - 1\). To quantify the excitation statistics in the environment, we define the global density at time step \(t\) as
\begin{equation}
    \eta_t \coloneq \frac{1}{N_d}\sum_{i\in V} s_t(i),
\end{equation}
and compute the statistics of the time series \(\{\eta_t\}\). We simulate the bath dynamics up to \(t=10^6\) cycles, recording statistics only after a burn-in period of \(2 \times 10^5\) cycles. This procedure is repeated for 10 independent trajectories.

\begin{figure*}[t]
    \centering
    \includegraphics[width=0.95\textwidth]{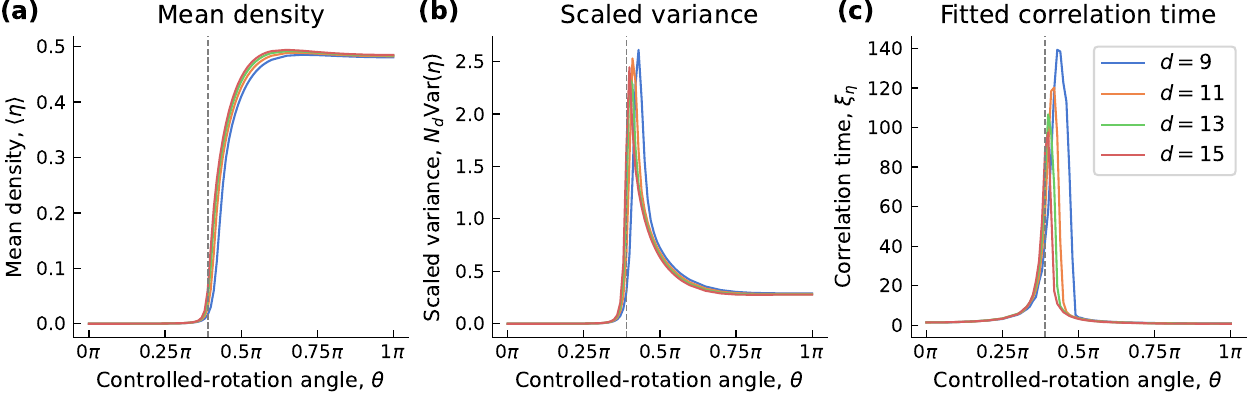}
    \caption{\textbf{Pseudo-criticality and critical slowing down in the QCA-induced PCA.}
    We fix the storm parameters to \(a=10^{-4}\) and \(b=0.5\), and sweep the spatial coupling parameter \(\theta\), which controls the strength of coherent bath interactions. Data are shown for lattices corresponding to rotated surface code patches of distances \(d=9, 11, 13, 15\) (number of sites \(N_d = 2 d^2 - 1\)). As a function of \(\theta\), we plot \textbf{(a)} the mean activity (spin) density \(\langle \eta_t\rangle_t\), where \(\eta_t \coloneq N_d^{-1}\sum_{i\in V} s_t(i)\); \textbf{(b)} the scaled variance \(N_d\,\mathrm{Var}(\eta_t)\); and \textbf{(c)} the correlation time \(\xi_\eta\) extracted by fitting the normalized autocorrelation \(C(\tau)\propto e^{-\tau/\xi_\eta}\). The vertical dashed line marks the estimated pseudo-critical threshold \(\theta_{\mathrm{th}}\approx 0.39\pi\). Taken together, the sharp density crossover, fluctuation peak, and correlation time spike identify a narrow pseudo-critical window in which bath fluctuations become both macroscopic and long-lived; their modest size-dependent sharpening and shift are consistent with finite-size critical behavior.}
    \label{fig:qca-statistical}
\end{figure*}

\Cref{fig:qca-statistical} summarizes three complementary diagnostics as a function of \(\theta\): the mean density \(\langle\eta\rangle\), the scaled variance \(N_d \mathrm{Var}(\eta)\), and a fitted correlation time \(\xi_\eta\) extracted from the autocorrelations of \(\eta_t\).

\Cref{fig:qca-statistical}(a) shows the mean density \(\langle\eta\rangle = \langle\eta_t\rangle_t\). For small \(\theta\), the PCA step is weak, and the dynamics are dominated by the storm process. Due to the strong relaxation \(b=0.5\), the bath remains in an effectively inactive or calm phase with a low background density \(\langle\eta\rangle \approx 2 \times 10^{-4}\). As \(\theta\) increases, the density undergoes a sharp, continuous crossover beginning around \(\theta \approx 0.35\pi\). The neighbor-dependent spreading outpaces the local decay, driving the bath into an active steady state that saturates near \(\langle\eta\rangle \approx 0.5\).

To probe the density fluctuations, we compute the scaled variance \(N_d \mathrm{Var}(\eta)\), with \(\mathrm{Var}(\eta) = \langle\eta_t^2\rangle - \langle\eta_t\rangle^2\), reported in \cref{fig:qca-statistical}(b). This quantity is analogous to susceptibility in equilibrium statistical mechanics. We observe a pronounced, narrow peak at \(\theta \approx 0.4\pi\), aligned with the sharp rise in mean density. This indicates a massive enhancement of macroscopic fluctuations, where in this near-critical regime, small local clusters of errors can spontaneously grow into system-wide avalanches before dying out. Over the evaluated system sizes, the peak exhibits a slight sharpening and shifting characteristic of finite-size scaling near a critical point.

Finally, to quantify memory in the bath dynamics, we compute the normalized autocorrelation function of the density time series,
\begin{equation}
    C(\tau) = \frac{\braket{(\eta_t - \braket{\eta})(\eta_{t + \tau} - \braket{\eta})}_t}{\mathrm{Var}(\eta)},
\end{equation}
and fit it to an exponential decay \(C(\tau) \propto e^{-\tau/\xi_\eta}\) to extract the effective correlation time. As shown in \cref{fig:qca-statistical}(c), \(\xi_\eta\) spikes dramatically near the transition threshold, reaching \(\xi_\eta \approx 140\) cycles for \(d=9\), consistent with a critical slowing down. Deep in either the calm or fully active phases, the bath relaxes rapidly (\(\xi_\eta \lesssim 3\)), corresponding to nearly Markovian dynamics. However, in the crossover regime, the competition between coherent spreading and stochastic decay leads to the emergence of macroscopic temporal memory.

In summary, tuning a single spatial parameter \(\theta\) drives the environment through a critical-like transition from a calm, weakly correlated background into a dense, active phase. The intermediate critical window maximizes both density fluctuations and temporal memory. We demonstrate below that this regime is highly detrimental for QEC, where enhanced susceptibility drives spatially correlated error avalanches, while critical slowing down causes these cascades to persist across multiple QEC cycles. We next evaluate surface code performance under these dynamics reusing the same model parameters.

\subsection{Surface code breakdown near pseudo-criticality}\label{sec:qca-qec-results}
We now evaluate surface code memory performance under the QCA noise model. As in 
\cref{sec:qca-stat-mech}, we fix the storm parameters to \(a=10^{-4}\) and \(b=0.5\), and sweep the controlled-rotation angle \(\theta\).

\begin{figure*}[t]
    \centering
    \includegraphics[width=0.95\textwidth]{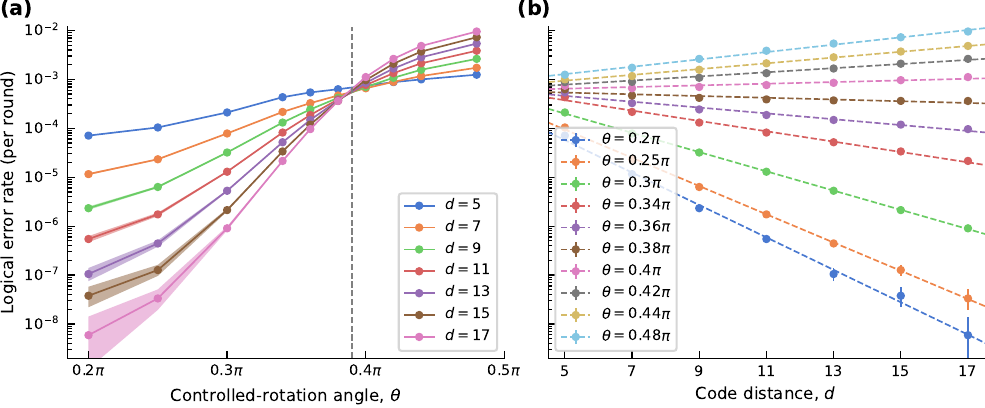}
    \caption{\textbf{Surface code breakdown near pseudo-critical spatiotemporal noise.} We fix \(a=10^{-4}\) and \(b=0.5\), and sweep the spatial coupling parameter \(\theta\) controlling the strength of coherent bath interactions at a baseline circuit-level noise rate of \(0.1\%\). Distances \(d\) are simulated for \(N_r=3d\) rounds. \textbf{(a)} Effective per-round logical error rate versus \(\theta\) for various \(d\); the vertical dashed line marks the pseudo-critical threshold \(\theta_{\mathrm{th}}\approx 0.39\pi\). \textbf{(b)} The same data versus distance \(d\) at fixed \(\theta\); dashed lines are exponential fits. Below the pseudo-critical window, logical errors remain exponentially suppressed with distance; near \(\theta_{\mathrm{th}}\) this suppression collapses, and above it the scaling reverses so that larger codes perform worse. The onset of this breakdown coincides with the enhanced fluctuations and critical slowing down observed in \cref{fig:qca-statistical}. Error bars, where visible, indicate binomial standard errors.}
    \label{fig:qca-qec-results}
\end{figure*}

The simulation protocol mirrors the surface code memory benchmark of \cref{sec:temporalstormmodel}. We simulate rotated surface code circuits of distance \(d\) for \(N_r = 3d\) rounds of syndrome extraction, applying the same global circuit-level baseline noise model at a characteristic rate \(0.1\%\). Logical failure rates are estimated via Monte Carlo sampling using Stim and decoding with minimum-weight perfect matching (MWPM) via PyMatching.

Inter-round correlated faults are incorporated by composing, at the start of each round, an additional Pauli fault sampled from the QCA-derived SPP, driven by the bath dynamics parameters \((a,b,\theta)\). A key distinction from the temporal storm sweep is that the spatially interacting PCA bath does not admit a simple closed-form expression for marginal Pauli error rates. Accordingly, we numerically estimate the required marginal error rate for each \(\theta\) to construct a marginalized detector error model for the MWPM decoder. Additionally, the bath is initialized in the all-zeros configuration for each shot.

\Cref{fig:qca-qec-results} details the logical error rate per round. In the calm regime (\(\theta \lesssim 0.34\pi\)), logical performance improves rapidly with increasing code distance. This behavior is consistent with dilute bath excitations and weak spatiotemporal structure, where the surface code remains effective at error suppression as previously explored~\cite{fowlerQuantifyingEffectsLocal2014}. As \(\theta\) enters the pseudo-critical window (\(\theta \approx 0.36\pi\)–\(0.40\pi\)), logical performance degrades sharply, and increasing code distance no longer yields reliable exponential suppression. Above a threshold \(\theta_{\mathrm{th}} \approx 0.39\pi\), the distance scaling trend entirely reverses, with larger codes performing worse than smaller ones. The scaling plot in \cref{fig:qca-qec-results}(b) makes this breakdown explicit.

These trends align directly with the statistical-mechanical diagnostics from \cref{sec:qca-stat-mech}. The same \(\theta\) window that strongly degrades QEC performance coincides with the regime that maximizes density fluctuations and exhibits critical slowing down. In this regime, enhanced susceptibility corresponds to rare but large spatiotemporal `avalanches'---spacetime-cone-like error clusters that can grow into system-wide cascades before dying out. Qualitatively similar error burst mechanisms have been reported in recent experiments on superconducting qubits, for example due to cosmic rays~\cite{mcewenResolvingCatastrophicError2022,liCosmicrayinducedCorrelatedErrors2025,harringtonSynchronousDetectionCosmic2025} or leakage events~\cite{mcewenRemovingLeakageinducedCorrelated2021,heExperimentalQuantumError2025}.

Taken together, these results establish a concrete mechanism by which coherence in local microscopic bath dynamics can manifest, after operational twirling, as tunable non-equilibrium criticality that degrades and ultimately defeats QEC. Earlier analytic work anticipated failures of this general kind, showing that correlated environments---particularly common bosonic baths---can degrade or destroy error correction~\cite{klesseQuantumErrorCorrection2005,novaisDecoherenceCorrelatedNoise2006,novaisSurfaceCodeThreshold2013,hutterBreakdownSurfacecodeError2014}, and deriving threshold theorems that bound the noise correlations fault tolerance can provably survive~\cite{terhalFaulttolerantQuantumComputation2005,aharonovFaultTolerantQuantumComputation2006,ngFaulttolerantQuantumComputation2009,preskillSufficientConditionNoise2013}. The scenario realized here is complementary and explicitly dynamical: criticality emerges in the bath's own non-equilibrium dynamics and is inherited by the twirled Pauli process, rather than being imposed through an assumed noise model or bath structure. In this sense, SPPs offer a route to extending these analyses beyond the reach of analytic bounds: microscopically derived correlated noise can be tractably sampled and confronted directly with circuit-level QEC performance at experimentally relevant scales.

We emphasize that MWPM with a marginalized detector model does not exploit correlations. Correlation-aware decoders, or approaches utilizing reweighting and model learning, may partially recover performance, but likely incur significant complexity overheads. In this sense, the QCA model offers an efficiently simulable stress test. It is derived from microscopic dynamics, exhibits nontrivial spatiotemporal correlations, and induces a breakdown of standard QEC protocols. We now conclude by discussing future directions and extensions unlocked by the SPP framework.

\section{Discussion}\label{sec:discussion}
In this work, we introduced \emph{spatiotemporal Pauli processes} (SPPs)---the natural multi-time generalization of Pauli channels---as a constructive bridge between microscopic non-Markovian device dynamics and the stochastic Pauli fault models that constitute the effective language of scalable stabilizer QEC. Under the multi-time Pauli twirl, as induced by Pauli-frame randomization, any process tensor is projected to a process-separable Pauli comb, equivalently a joint distribution over Pauli fault trajectories in spacetime. Crucially, this projection is implemented by fixed local contractions on process tensor networks, producing classical tensor networks whose temporal bond dimensions are bounded by the environment's Liouville dimension, and which support transfer operator diagnostics, hidden Markov realizations, and efficient generative sampling. Deployed in circuit-level surface code memory and stability benchmarks, the framework first isolates the effect of temporal memory alone through a controlled storm model at fixed marginal error rates. More notably, it identifies a microscopic route from near-critical bath dynamics to QEC failure: a coherent QCA bath whose twirled dynamics induce a nonlinear probabilistic cellular automaton, with a pseudo-critical regime marked by critical slowing down, macroscopic error avalanches, and a reversal of surface code distance scaling.

The remainder of this section outlines the applications and open directions this framework enables, before returning to the scope and caveats of our results. An immediate application is to model the correlated error mechanisms already documented in hardware. Leakage-induced correlations~\cite{mcewenRemovingLeakageinducedCorrelated2021}, cosmic ray and radiation-induced bursts~\cite{mcewenResolvingCatastrophicError2022,harringtonSynchronousDetectionCosmic2025,bratrudMeasurementCorrelatedCharge2025}, and slow parameter drift are naturally expressed as SPPs with appropriate latent dynamics, upgrading hand-inserted burst injections to microscopically parameterized generative models. In the same spirit, SPPs supply a unified language in which the broader landscape of correlated noise descriptions---from early open-systems and threshold analyses of QEC in correlated environments~\cite{alickiInternalConsistencyFaulttolerant2006,klesseQuantumErrorCorrection2005,novaisSurfaceCodeThreshold2013,hutterBreakdownSurfacecodeError2014,terhalFaulttolerantQuantumComputation2005,ngFaulttolerantQuantumComputation2009,preskillSufficientConditionNoise2013} to process tensor treatments of QEC~\cite{tanggaraStrategicCodeUnified2024,kobayashiTensornetworkDecodersProcess2024} and phenomenological Pauli extensions~\cite{claderImpactCorrelationsHeavy2021,fkamDetrimentalNonMarkovianErrors2025}---can be recast, compared, and extended on equal footing.

A first practical challenge is to characterize these effective processes on real devices. At the QEC-cycle level, SPPs could be learned from syndrome streams, extending methods for reconstructing phenomenological Pauli noise from syndrome data~\cite{wagnerOptimalNoiseEstimation2021,wagnerPauliChannelsCan2022}. At the circuit level, scalable Pauli noise learning protocols~\cite{erhardCharacterizingLargescaleQuantum2019,flammiaEfficientEstimationPauli2020,harperLearningCorrelatedNoise2023,hockingsScalableNoiseCharacterization2025} could be extended to explicitly multi-time processes. For smaller systems, process tensor tomography~\cite{whiteNonMarkovianQuantumProcess2022,whiteUnifyingNonMarkovianCharacterization2025} provides a complementary route from non-Markovian characterization to QEC-compatible Pauli models. Once learned or derived, SPPs support interpretable diagnostics including transfer spectra, correlation times, and burst statistics. One could also define a \emph{logical SPP} for temporally correlated logical Pauli errors, connecting microscopic noise correlations to emergent logical non-Markovianity~\cite{ziyadEmergentNonMarkovianityLogical2026}.

Several open questions remain concerning the algorithmic details of compressing, contracting, learning, and sampling from SPP tensor networks, particularly in higher-dimensional settings. One avenue is to adapt existing protocols, including boundary-MPS methods~\cite{jordanClassicalSimulationInfinitesize2008,vanderstraetenVariationalMethodsContracting2022}, tensor renormalization~\cite{levinTensorRenormalizationGroup2007,evenblyAlgorithmsEntanglementRenormalization2009}, and compressed contraction~\cite{grayHyperoptimizedApproximateContraction2024}, or to integrate the framework with importance sampling~\cite{bravyiSimulationRareEvents2013,myersSimulatingGeneralNoise2025}. SPP tensor networks can also be connected to factor graph representations of correlated Pauli noise~\cite{flammiaEfficientEstimationPauli2020}, providing a route to extend graphical-model approaches to multi-time processes. Related questions arise for SPP HMM realizations. Natural extensions include determining when compact exact or approximate HMM representations can be constructed, and linking the scaling and dynamics of the latent state space to the physical correlations they generate. This connects SPPs to the hidden Markov process and latent-state literature~\cite{leonarde.baumMaximizationTechniqueOccurring1970,vidyasagarCompleteRealizationProblem2011}.

The SPP framework enables these capabilities to feed directly into QEC protocol design. A learned or derived SPP can serve as a prior over spacetime Pauli fault trajectories conditioned on the observed syndrome, enabling correlation-aware decoding beyond marginalized detector models. Depending on the structure of the process, such priors could be incorporated through reweighted MWPM, HMM filtering, belief propagation, or neural decoders~\cite{nickersonAnalysingCorrelatedNoise2019,sivakOptimizationDecoderPriors2024,nayakIterativeDecodingStabilizer2026,bauschLearningHighaccuracyError2024,gicevFullyConvolutional3D2025}. Beyond decoding, SPPs can act as generative test beds for tailoring codes, syndrome extraction circuits, and mitigation strategies to hardware-specific spatiotemporal correlations~\cite{tanggaraStrategicCodeUnified2024,mcewenRelaxingHardwareRequirements2023,debroyLUCISurfaceCode2025}.

Finally, we delineate the scope of the framework and of our numerical results. As with ordinary Pauli channels, SPPs should be interpreted as effective stochastic noise models defined at a chosen operational level. This level may be set by Pauli-frame randomization, randomized compiling, stabilizer measurements, or simply by a Pauli approximation to the relevant dynamics. Hardware-specific limitations of these procedures, such as imperfect randomization or device-dependent control constraints, are not addressed here. Two further caveats delimit our numerical results. Decoding throughout uses minimum-weight perfect matching with marginalized detector error models (current standard practice). The reported degradation may therefore conflate two effects: decoder mismatch under independence assumptions, which the correlation-aware decoding discussed above could partially recover, and a genuine loss of correctability intrinsic to the correlated noise itself. We expect the intrinsic component to dominate in the pseudo-critical regime, where avalanches on the scale of the code distance defeat any decoding strategy at fixed marginal rates; quantifying the recoverable fraction more generally is an open problem for which SPPs provide a natural setting. Likewise, we characterize the QCA-induced transition phenomenologically, leaving its statistical-mechanical classification to future work.

Ultimately, the promise of fault tolerance rests on the premise that once physical error rates are pushed below threshold, scaling does the rest. Our results sharpen the caveat: below-threshold average error rates alone do not guarantee scalable error suppression, because correlations can blunt or even reverse the benefit of code distance. As devices enter the regime where rare correlated events dominate logical failure, frameworks that connect device physics to logical performance will be essential for diagnosing these failure modes, decoding around them, and ultimately engineering them away. SPPs provide such a bridge: a practical, operationally grounded route from microscopic correlated dynamics to circuit-level QEC.

\section*{Data and code availability}
The data and source code used to generate the results presented in this study are openly available at \url{https://github.com/jkfids/corrqec2}. The stability experiment circuit construction is adapted from the source code accompanying Ref.~\cite{gidneyStabilityExperimentsOverlooked2022}, archived at Ref.~\cite{gidneyDataStabilityExperiments2022} and licensed under CC BY 4.0.

\begin{acknowledgments}
This research was supported by the Commonwealth through an Australian Government Research Training Program Scholarship [DOI: \url{https://doi.org/10.82133/C42F-K220}]. This research is supported by the Ministry of Education, Singapore, under its Academic Research Fund (AcRF) Tier 1 grant, and funded through the SUTD Kickstarter Initiative (SKI 2021\_07\_02). This research is supported by the National Research Foundation, Singapore's Centre for Quantum Technologies grant CQT\_SUTD\_2025\_01. John F. Kam is supported by a CSIRO Future Science Platform Scholarship and the Agency for Science, Technology and Research (A*STAR) Research Attachment Programme (ARAP) at the Quantum Innovation Centre (Q.InC).
\end{acknowledgments}

\bibliographystyle{apsrev4-2}
\bibliography{references}

\appendix
\crefalias{section}{appendix}
\crefalias{subsection}{appendix}

\onecolumngrid

\clearpage
\section{Dictionary of notations}\label{app:notations}
\newcommand{\paddedvalignbox}{\adjustbox{valign=c,padding=0 4.6pt}}

\begingroup
\setlength{\abovecaptionskip}{0pt}
\captionof{table}{\textbf{Translation table of common quantum objects and operations.} Equivalent notation of quantum objects and operations in Dirac notation, abstract index notation, and graphical (tensor network) calculus.}
\label{tab:notations}
\centering
    \begin{tabular*}{\textwidth}{@{\extracolsep{\fill}}@{\hspace{20pt}}ccc@{\hspace{20pt}}}
        \toprule
        Dirac notation & Abstract index notation & Graphical calculus \\

        \midrule

        \makecell{State vector (ket)\\ \(\ket{\psi}\)} & \(\psi^a\) & \paddedvalignbox{\includegraphics[]{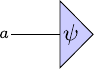}} \\

        \makecell{Dual vector (bra)\\\(\bra{\psi}\)} & \(\psi_a\) & \paddedvalignbox{\includegraphics[]{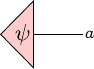}} \\

        \makecell{Inner product\\\(\braket{\phi | \psi}\)} & \(\phi_a\psi^a\) & \paddedvalignbox{\includegraphics[]{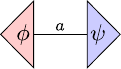}} \\

        \makecell{Outer product\\\(\ketbra{\psi}{\phi}\)} & \(\psi^a\phi_{a'}\) & \paddedvalignbox{\includegraphics[]{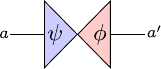}} \\

        \makecell{Operator\\\(A\)} & \(A^a{}_b\) & \paddedvalignbox{\includegraphics[]{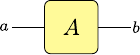}} \\

        \makecell{Multipartite operator\\\(A\)} & \(A^{a_0\dots a_{n-1}}{}_{b_0\dots b_{n-1}}\) & \paddedvalignbox{\includegraphics[]{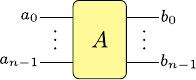}} \\

        \makecell{Composition\\\(A \circ B\)} & \(A^a{}_iB^i{}_b\) & \paddedvalignbox{\includegraphics[]{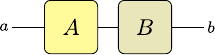}} \\

        \makecell{Tensor product\\\(A\otimes B\)} & \(A^a{}_bB^c{}_d\) & \paddedvalignbox{\includegraphics[]{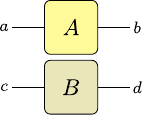}} \\

        \makecell{Identity \\ \(\mathbb{I}\)} & \(\delta^a{}_b\) & \paddedvalignbox{\includegraphics[]{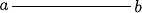}} \\

        \makecell{SWAP} & \(\delta^{a}{}_{d} \delta^{c}{}_{b}\) & \paddedvalignbox{\includegraphics[]{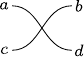}} \\

        \makecell{Unnormalized Bell state \\\(\ket{\Phi^+} = \dket{\mathbb{I}}\)} & \(\delta^{ab}\) & \paddedvalignbox{\includegraphics[]{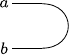}} \\

        \makecell{Vectorized operator\\\(\dket{\rho}\)} & \(\rho^{a'a} = \rho^a{}_i\delta^{a'i}\) & \paddedvalignbox{\includegraphics[]{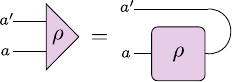}} \\

        \makecell{Trace\\\(\Tr{(A)}\)} & \(A^i{}_i\) & \paddedvalignbox{\includegraphics[]{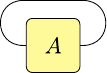}} \\


        \makecell{Transpose\\\(A^T\)} & \(A^a{}_b \to A^{b}{}_{a}\) & \paddedvalignbox{\includegraphics[]{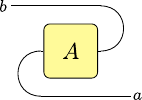}} \\

        \bottomrule
    \end{tabular*}
\endgroup

\clearpage

\section{Proof of process-separability of SPPs}\label{app:proof-pt-to-spp}
Here, we provide the proof of \cref{thm:pt-to-spp}.
\begin{proof}
    The dynamics of a \(k\)-slot process tensor are described by its Choi operator \(\Upsilon_{0:k}\in\mathcal{B}(\mathcal{H}_{\Upsilon_{0:k}})\), where \(\mathcal{H}_{\Upsilon_{0:k}}:=\bigotimes_{j=0}^k(\mathcal{H}_{\inp_j}\otimes\mathcal{H}_{\out_j})\). We expand \(\Upsilon_{0:k}\) in the Pauli basis on each input-output pair as
    \begin{equation}\label{eq:pauli-expansion}
        \Upsilon_{0:k} = \sum_{\vec{\mu},\vec{\nu}} \chi_{\vec{\mu},\vec{\nu}} \bigotimes_{j=0}^k (P_{\mu_j}\otimes P_{\nu_j}),
    \end{equation}
    denoting \(\vec{\mu} = (\mu_0,\mu_1,\cdots,\mu_k)\) and \(\vec{\nu} = (\nu_0,\nu_1,\cdots,\nu_k)\). Here the Pauli operators \(P_{\mu_j}\) and \(P_{\nu_j}\) act on \(\mathcal{H}_{\inp_j}\) and \(\mathcal{H}_{\out_j}\), respectively. Up to a convention-dependent normalization factor, \(\chi_{\vec{\mu},\vec{\nu}}\) are real coefficients given by
    \begin{equation}
        \chi_{\vec{\mu},\vec{\nu}} = \Tr\left[P_{\vec{\mu},\vec{\nu}}\Upsilon_{0:k}\right], 
        \qquad
        P_{\vec{\mu},\vec{\nu}} \coloneq \bigotimes_{j=0}^k (P_{\mu_j}\otimes P_{\nu_j}).
    \end{equation}
    The multi-time Pauli twirl \(\mathcal{T}_P^{(k)}\) is defined in \cref{def:multi-time-twirl}, which we restate here,
    \begin{equation}
        \Upsilon_{0:k}^{\mathcal{T}_P} := \mathcal{T}_P^{(k)}(\Upsilon_{0:k}) = \frac{1}{|\mathbb{P}^{(n)}|^{k+1}} \sum_{Q_0,\dots,Q_k\in\mathbb{P}^{(n)}} \bigg(\bigotimes_{j=0}^k Q_j\otimes Q_j\bigg) \Upsilon_{0:k} \bigg(\bigotimes_{j=0}^k Q_j\otimes Q_j\bigg).
    \end{equation}
    Substituting the Pauli expansion from \cref{eq:pauli-expansion} yields
    \begin{equation}\label{eq:pauli-sub}
        \Upsilon_{0:k}^{\mathcal{T}_P} = \frac{1}{|\mathbb{P}^{(n)}|^{k+1}} \sum_{Q_0,\dots,Q_k\in\mathbb{P}^{(n)}} \sum_{\vec{\mu},\vec{\nu}} \chi_{\vec{\mu},\vec{\nu}} \bigotimes_{j=0}^k (Q_j P_{\mu_j} Q_j \otimes Q_j P_{\nu_j} Q_j).
    \end{equation}
    Using the Pauli conjugation averaging identity,
    \begin{equation}
        \frac{1}{|\mathbb{P}^{(n)}|}\sum_{Q\in\mathbb{P}^{(n)}} (Q P_{\mu} Q) \otimes (Q P_{\nu} Q) = \delta_{\mu,\nu} (P_{\mu} \otimes P_{\nu}),
    \end{equation}
    all cross terms in \(\chi_{\vec{\mu},\vec{\nu}}\) with \(\vec{\mu} \neq \vec{\nu}\) in \cref{eq:pauli-sub} vanish, simplifying the expression to
    \begin{equation}
        \Upsilon_{0:k}^{\mathcal{T}_P} = \sum_{\vec{\xi}} \tilde{\chi}_{\vec{\xi}} \bigotimes_{j=0}^k (P_{\xi_j} \otimes P_{\xi_j}),
    \end{equation}
    where we have relabelled the remaining coefficients as \(\tilde{\chi}_{\vec{\xi}} = \chi_{\vec{\xi},\vec{\xi}}\), \(\vec{\xi} = (\xi_0, \xi_1, \dots, \xi_k)\).

    To prove the absence of quantum temporal correlations, we must show that \(\Upsilon_{0:k}^{\mathcal{T}_P}\) is separable across the bipartitions between distinct time steps. The key insight is that the Bell basis states, \(\ket{\beta_i} \in \{\ket{\Phi^+}, \ket{\Phi^-}, \ket{\Psi^+}, \ket{\Psi^-}\}\) (for a single-qubit example), are simultaneous eigenvectors of operators of the form \(P_\xi \otimes P_\xi\). Consequently, \(\Upsilon_{0:k}^{\mathcal{T}_P}\) is diagonal in the Bell product basis \(\ket{\mathcal{B}_{\vec{\alpha}}}=\bigotimes_{j=0}^k \ket{\beta_{\alpha_j}}\). Explicitly, we perform the change of basis
    \begin{equation}
        \Upsilon_{0:k}^{\mathcal{T}_P} 
        =
        \sum_{{\vec{\alpha}}{\vec{\beta}}} \lambda_{\vec{\alpha},\vec{\beta}} \ketbra{\mathcal{B}_{\vec{\alpha}}}{\mathcal{B}_{\vec{\beta}}}, 
        \qquad
        \lambda_{\vec{\alpha},\vec{\beta}} 
        \coloneq 
        \bra{\mathcal{B}_{\vec{\alpha}}} 
        \sum_{\vec{\xi}} \tilde{\chi}_{\vec{\xi}} \bigotimes_{j=0}^k (P_{\xi_j} \otimes P_{\xi_j})
        \ket{\mathcal{B}_{\vec{\beta}}},
    \end{equation}
    and since \(\ket{\mathcal{B}_{\vec{\alpha}}}\) are eigenvectors of \(\bigotimes_{j=0}^k (P_{\xi_j} \otimes P_{\xi_j})\), we have \(\lambda_{\vec{\alpha},\vec{\beta}} = 0\) for all \(\vec{\alpha} \neq \vec{\beta}\). Hence, we may write
    \begin{equation}\label{eq:bell-basis-decomposition}
        \Upsilon_{0:k}^{\mathcal{T}_P} 
        =
        \sum_{\vec{\alpha}} \lambda_{\vec{\alpha}} \ketbra{\mathcal{B}_{\vec{\alpha}}}{\mathcal{B}_{\vec{\alpha}}},
        \qquad
        \lambda_{\vec{\alpha}} 
        \coloneq
        \bra{\mathcal{B}_{\vec{\alpha}}} 
        \sum_{\vec{\xi}} \tilde{\chi}_{\vec{\xi}} \bigotimes_{j=0}^k (P_{\xi_j} \otimes P_{\xi_j})
        \ket{\mathcal{B}_{\vec{\alpha}}}.
    \end{equation}

    Because \(\Upsilon_{0:k}^{\mathcal{T}_P}\) is positive semidefinite and the vectors \(\ket{\mathcal{B}_{\vec{\alpha}}}\) form an eigenbasis, the corresponding eigenvalues \(\lambda_{\vec{\alpha}}\) in \cref{eq:bell-basis-decomposition} are non-negative. Moreover, since \(\Upsilon_{0:k}^{\mathcal{T}_P}\) is diagonal in the Bell product basis, its trace is the sum of these eigenvalues,
    \begin{equation}
        \sum_{\vec{\alpha}} \lambda_{\vec{\alpha}} = \Tr[\Upsilon_{0:k}^{\mathcal{T}_P}] = \Tr[\Upsilon_{0:k}],
    \end{equation}
    where the equality follows from the trace preservation of the twirl. Hence, \(\Upsilon_{0:k}^{\mathcal{T}_P}\) is manifestly process-separable (cf. \cref{eq:process-separable}), and the normalized weights \(p_{\vec{\alpha}} = \lambda_{\vec{\alpha}} / \Tr[\Upsilon_{0:k}]\) define a probability distribution as in \cref{eq:spp-probabilities}.

    Finally, the coefficients \(\lambda_{\vec{\alpha}}\) can be related to the Pauli decomposition coefficients \(\tilde{\chi}_{\vec{\xi}}\) via a Walsh-Hadamard transform,
    \begin{equation}
        \lambda_{\vec{\alpha}} = \sum_{\vec{\xi}} \tilde{\chi}_{\vec{\xi}} \prod_{j=0}^k M_{\alpha_j\xi_j},
         \qquad
         M_{\alpha\xi} \coloneq \bra{\beta_\alpha}(P_\xi\otimes P_\xi)\ket{\beta_\alpha}.
    \end{equation}
\end{proof}

\section{Derivation of QCA to PCA HMM mapping}\label{app:qca-pca-mapping}
In this appendix, we derive the effective probabilistic cellular automaton (PCA) HMM from the microscopic quantum cellular automaton (QCA) noise model under the system Pauli twirl. We first summarize the setting and assumptions in \cref{app:assumptions}, then derive the effective environment dephasing in \cref{app:env-dephasing}, the PCA kernel in \cref{app:pca-kernel}, and the system Pauli emission channel in \cref{app:system-pauli-emission}. Finally, we summarize the full mapping in \cref{app:summary-mapping}.

\subsection{Setting and assumptions}\label{app:assumptions}
We fix the microscopic QCA noise model from \cref{sec:qca-details} and collect here the
standing assumptions and notation used to derive the effective PCA HMM.
The environment consists of qubits on a bipartite graph \(V = R \cup B\) (red/black sublattices). Throughout, \(\{\ket{s}\}\) denotes the computational basis of the environment Hilbert space, with
\(s\in\{0,1\}^{|V|}\), and we write \(\Pi_s := \ketbra{s}{s}\).
For the bipartition we use
\begin{equation}
s=(r,b)\in \{0,1\}^{|R|}\times\{0,1\}^{|B|},\qquad \ket{s}=\ket{r}_R\otimes\ket{b}_B,
\end{equation}
with corresponding projectors \(\Pi_r^R:=\ketbra{r}{r}_R\) and \(\Pi_b^B:=\ketbra{b}{b}_B\).
For a site \(j\in B\) we denote its red neighbors by \(N_R(j)\subseteq R\), and similarly
for \(i\in R\) we denote its black neighbors by \(N_B(i)\subseteq B\). We make the following assumptions, labelled (I)--(V).

\begin{enumerate}
    \item[(I)] \textbf{Classical environment initialization}. Before any dynamics, the environment is initialized as a mixture of computational basis states
    \begin{equation}\label{eq:initial_rho_env}
        \rho_E = \sum_s \Pr(s) \Pi_s.
    \end{equation}

    \item[(II)] \textbf{Single cycle structure}. One cycle of the microscopic dynamics is the composition of (i) an environment-only storm channel \(\mathcal{S}^E\), (ii) two QCA half-steps \(\mathcal{Q}^E = \mathcal{Q}_{B\to R}^E \circ \mathcal{Q}_{R\to B}^E\) on the environment, and (iii) a local system-environment unitary interaction \(\mathcal{U}^{SE}\), as detailed in \cref{sec:qcamodel}.
    
    \item[(III)] \textbf{QCA unitaries}. Denote the QCA half-step unitaries as \(Q_{R\to B}^E(\theta)\) and \(Q_{B\to R}^E(\theta)\) for \(\mathcal{Q}_{R\to B}^E\) and \(\mathcal{Q}_{B\to R}^E\) respectively. Fix \(\theta \in \mathbb{R}\), and for each directed edge \(i \to j\) with \(i \in R\), \(j \in B\), define the controlled \(X\)-rotation on the target \(E_j\),
    \begin{equation}
        U_{i\to j}(\theta):=\ketbra{0}{0}_{E_i}\otimes I_{E_j}+\ketbra{1}{1}_{E_i}\otimes e^{-i\theta X_{E_j}}.
    \end{equation}
    The unitary \(Q^{E}_{R \to B}(\theta)\) is the product of these (commuting) gates over all such edges,
    \begin{equation}
        Q^{E}_{R\to B}(\theta):=\prod_{j\in B}\ \prod_{i\in N_R(j)} U_{i\to j}(\theta).
    \end{equation}
    Similarly, \(Q^{E}_{B\to R}(\theta)\) is defined by swapping the roles of \(R\) and \(B\).

    \item[(IV)] \textbf{System-environment interaction}.
    At each site, the system interacts locally with the corresponding environment qubit via a controlled unitary
    \begin{equation}\label{eq:se-unitary}
        U_{SE}=\sum_{z\in\{0,1\}}\Pi_z^E\otimes V_z^S,
    \end{equation}
    where \(\Pi^{E}_{z}:=\ketbra{z}{z}_E\). In general, the conditional unitaries \(\{V_z\}\) must be orthogonal under the Hilbert-Schmidt inner product,
    \begin{equation}\label{eq:hs-orthogonal}
        \Tr(V_zV_{z'}^\dagger) = d_S \delta_{z,z'},
    \end{equation}
    where \(d_S\) is the local system dimension. In this work, we specialize to
    \begin{equation}\label{eq:unitary-decomp}
        V_0 = I, \quad V_1 = \vec{n} \cdot \vec{\sigma} = n_X X + n_Y Y + n_Z Z, \quad \|\vec{n}\|^2 = 1,
    \end{equation}
    for some \(\vec{n} \in \mathbb{R}^3\).

    \item[(V)] \textbf{System Pauli twirl}.
    Each cycle includes Pauli-frame randomization on the system, yielding the effective (local) system-environment interaction
    \begin{equation}
        \widetilde{\mathcal{U}}^{SE}[\rho] = \frac{1}{4}\sum_{P\in\{I,X,Y,Z\}} (I \otimes P) U_{SE} (I \otimes P) \rho (I \otimes P) U_{SE}^\dagger (I \otimes P).
    \end{equation}
    For later use, it is convenient to record the generalized operator-sum representation for \(\widetilde{\mathcal{U}}^{SE}\). Inserting \cref{eq:se-unitary} into the above,
    \begin{equation}\label{eq:twirled-kraus-ops}
    \widetilde{\mathcal{U}}^{SE}[\rho] = \sum_P \sum_{z,z'} K_{P, z} \rho K_{P, z'}^\dagger, \quad K_{P, z} = \frac{1}{2} \Pi_z \otimes (P V_z P).
\end{equation}
\end{enumerate}
We now proceed to derive the effective PCA kernel and system Pauli emission channel under these assumptions.

\subsection{Environment re-classicalization induced by system twirling}\label{app:env-dephasing}
We first show that, under the system Pauli twirl, the reduced action of the system-environment interaction on the environment is exactly a \(Z\)-basis dephasing channel. This is the mechanism by which the environment is re-classicalized at each cycle.

For a single site (suppressing the site index), let \(\rho_S\) be an arbitrary (fixed) system state, and define the reduced environment channel induced by the twirled interaction \(\widetilde{\mathcal{U}}^{SE}\) in (V) as
\begin{equation}\label{eq:system-partial-trace}
    \mathcal{E}[\rho_E] = \Tr_S(\widetilde{\mathcal{U}}^{SE}[\rho_E \otimes \rho_S]),
\end{equation}
Using the Kraus representation in \cref{eq:twirled-kraus-ops}, we may write
\begin{equation}
    \begin{split}
        \mathcal{E}[\rho_E] 
        &= \sum_P \sum_{z,z'} \Tr_S(K_{P, z} (\rho_E \otimes \rho_S) K_{P, z'}^\dagger) \\
        &= \sum_P \sum_{z,z'} \frac{1}{4} \Tr_S((\Pi_z \otimes PV_zP) (\rho_E \otimes \rho_S) (\Pi_{z'} \otimes PV_{z'}P)^\dagger) \\
        &= \sum_{z,z'} \Pi_z \rho_E \Pi_{z'} \cdot \frac{1}{4} \sum_P \Tr(P V_z P \rho_S P V_{z'}^\dagger P).
    \end{split}
\end{equation}
By the cyclicity of the trace and \(P^2 = I\), we express
\begin{equation}\label{eq:post-cyclic}
    \mathcal{E}[\rho_E] = \sum_{z,z'} \Pi_z \rho_E \Pi_{z'} \cdot \Tr(V_z \frac{1}{4}\sum_P(P \rho_S P) V_{z'}^\dagger).
\end{equation}
The Pauli average inside the trace is the completely depolarizing channel on \(\rho_S\),
\begin{equation}
    \frac{1}{4}\sum_{P \in \{I,X,Y,Z\}} P \rho_S P = \frac{\mathbb{I}}{2}.
\end{equation}
So \cref{eq:post-cyclic} becomes,
\begin{equation}
    \mathcal{E}[\rho_E] = \sum_{z,z'} \Pi_z \rho_E \Pi_{z'} \cdot \frac{1}{2} \Tr(V_z V_{z'}^\dagger),
\end{equation}
Invoking the Hilbert-Schmidt orthogonality in \cref{eq:hs-orthogonal}, we obtain
\begin{equation}
    \mathcal{E}[\rho_E] = \sum_z \Pi_z \rho_E \Pi_z.
\end{equation}
That is, the reduced environment update is exactly dephasing in the computational basis, and independent of the choice of \(\rho_S\). For the full environment lattice, the interaction \(U_{SE} = \bigotimes_{i\in V} U_{SE}^{(i)}\) factorizes across sites and the twirl is applied locally, so the reduced environment update also factorizes as
\begin{equation}
    \Delta_E[\rho] 
    \coloneq \bigotimes_{i\in V} \sum_{z\in\{0,1\}} \Pi_z^{E_i} \rho \Pi_z^{E_i}
    = \sum_{s \in \{0,1\}^{|V|}} \Pi_s \rho \Pi_s,
\end{equation}
For use in the next section, we also define the sublattice dephasing channels
\begin{equation}
    \Delta_B[\rho] \coloneq \sum_b(\Pi_b^B \otimes I_R) \rho (\Pi_b^B \otimes I_R), \quad
    \Delta_R[\rho] \coloneq \sum_r(I_B \otimes \Pi_r^R) \rho (I_B \otimes \Pi_r^R),
\end{equation}
such that \(\Delta_E = \Delta_B \circ \Delta_R = \Delta_R \circ \Delta_B\).

\subsection{Derivation of the PCA kernel}\label{app:pca-kernel}
We now derive the effective classical update rule generated by the coherent QCA step under the per-cycle environment dephasing derived in \cref{app:env-dephasing}. Concretely, we evaluate the map \(\rho_E \mapsto (\Delta_E \circ \mathcal{Q}^E)[\rho_E]\) for an initial classical state \(\rho_E\) as in (I). We show that this is equivalent to a two-step PCA update with flip probabilities \(\sin^2(k\theta)\), where \(k\) is the number of excited neighbors in the opposite sublattice. Note that although the storm channel \(\mathcal{S}^E\) precedes \(\mathcal{Q}^E\), it does not affect the form of the PCA kernel, since it is a classical channel that preserves the diagonal structure of \(\rho_E\).

First, we fix an environment basis configuration \(s = (r,b)\) and consider the action of \(Q_{R\to B}^E(\theta)\) from (III). For a single directed edge (gate) \(i \to j\) with \(i \in R\), \(j \in B\),
\begin{equation}\label{eq:single-i-single-j}
    U_{i\to j}(\theta)[\ket{r_i} \otimes \ket{\psi}_j] = \ket{r_i} \otimes (e^{-i\theta X_{E_j}})^{r_i} \ket{\psi}_j.
\end{equation}
The control bit \(r_i \in \{0,1\}\) controls whether to apply an \(X\)-rotation to the target \(j\) or to do nothing. For a fixed target \(j \in B\), the product over its red neighbors yields
\begin{equation}\label{eq:single-j}
    \prod_{i \in N_R(j)} U_{i\to j}(\theta)[\ket{r} \otimes \ket{\psi}_j] = \ket{r} \otimes (e^{-ik_j(r)\theta X_{E_j}}) \ket{\psi}_j,
\end{equation}
where \(k_j(r) \coloneq \sum_{i \in N_R(j)} r_i\) is the number of excited red neighbors. Because gates on distinct targets \(j\) commute, the full first half-step acts as
\begin{equation}\label{eq:first-half-step-factor}
    Q_{R\to B}(\theta)[\ket{r} \otimes \ket{b}] = \ket{r} \otimes \bigotimes_{j \in B} (e^{-ik_j(r)\theta X_{E_j}}) \ket{b_j}.
\end{equation}
By the exact same reasoning, \(Q_{B\to R}(\theta)\) can be written as
\begin{equation}\label{eq:second-half-step-factor}
    Q_{B\to R}(\theta)[\ket{r} \otimes \ket{b}] = \bigotimes_{i \in R} (e^{-ik_i(b)\theta X_{E_i}}) \ket{r_i} \otimes \ket{b},
\end{equation}
The key observation is that \(Q_{B\to R}(\theta)\) (the \emph{second} half-step) in \cref{eq:second-half-step-factor} uses \(B\) only as a computational-basis control. It leaves each basis state \(\ket{b}\) of \(B\) unchanged and applies a \(b\)-dependent unitary to \(R\), without coupling distinct basis states \(\ket{b}\) and \(\ket{b'}\). It can therefore be written as the block-diagonal operator
\begin{equation}
    Q_{B\to R}(\theta) = \sum_b \Pi_b^B \otimes W^R_b,
\end{equation}
for some unitaries \(W_b^R\) on \(R\). In this form, we can express the action of \(\Delta_B\) on the second half-step as
\begin{equation}
    \Delta_B[Q_{B \to R} \rho Q_{B \to R}^\dagger] = Q_{B \to R} \Delta_B[\rho] Q_{B \to R}^\dagger,
\end{equation}
i.e., the commutation relation \(\Delta_B \circ \mathcal{Q}_{B \to R} = \mathcal{Q}_{B \to R} \circ \Delta_B\). 

Now using \(\Delta_E = \Delta_R \circ \Delta_B\), we may write
\begin{equation}\label{eq:maps-commutation}
    \Delta_E \circ \mathcal{Q} = \Delta_R \circ \Delta_B \circ \mathcal{Q}_{B \to R} \circ \mathcal{Q}_{R \to B} = \Delta_R \circ \mathcal{Q}_{B \to R} \circ \Delta_B \circ \mathcal{Q}_{R \to B}.
\end{equation}
This expression allows us to insert \(\Delta_B\) immediately after the first half-step \(\mathcal{Q}_{R \to B}\). Crucially, this enables us to show that the intermediate coherences in \(B\) generated by the first half-step do not affect the final distribution, and that the two QCA half-step unitaries can be treated as separate PCA updates.

We now evaluate the action of \(\Delta_B \circ \mathcal{Q}_{R \to B}\) on a basis state \(\ket{r}\!\ket{b}\). For a single site \(j \in B\) and any \(\varphi \in \mathbb{R}\), recall
\begin{equation}
    e^{-i\varphi X} \ket{0} = \cos(\varphi) \ket{0} - i \sin(\varphi) \ket{1}, \quad
    e^{-i\varphi X} \ket{1} = \cos(\varphi) \ket{1} - i \sin(\varphi) \ket{0}.
\end{equation}
Hence, for \(z \in \{0,1\}\),
\begin{equation}
    \Delta_z[e^{-i\varphi X} \ketbra{z}{z} e^{i\varphi X}] = \cos^2(\varphi) \ketbra{z}{z} + \sin^2(\varphi) \ketbra{1\oplus z}{1\oplus z},
\end{equation}
since off-diagonal terms are removed by the dephasing.

We now take \(\varphi = k_j(r) \theta\). Using the factorized form of \(\mathcal{Q}_{R \to B}\) in \cref{eq:first-half-step-factor}, we can therefore interpret \(\Delta_B \circ \mathcal{Q}_{R \to B}\) as a classical flip rule on each black site \(j\in B\), conditional on \(r\). Explicitly, the bit \(b_j \mapsto b'_j\) is flipped with probability
\begin{equation}
    \Pr(b'_j = b_j \oplus 1 | r) = \sin^2(k_j(r) \theta), \quad
    \Pr(b'_j = b_j | r) = \cos^2(k_j(r) \theta).
\end{equation}
Furthermore, because \(\mathcal{Q}_{R \to B}\) factorizes across \(j\) in \cref{eq:first-half-step-factor}, the per-site updates can be applied independently across \(j\).

By the exact same logic, the map \(\Delta_R \circ \mathcal{Q}_{B \to R}\) implements a classical flip rule on each red site \(i \in R\), conditional on the updated black configuration \(b'\). Here, the bit \(r_i \mapsto r'_i\) is flipped with probability
\begin{equation}
    \Pr(r'_i = r_i \oplus 1 | b') = \sin^2(k_i(b') \theta), \quad
    \Pr(r'_i = r_i | b') = \cos^2(k_i(b') \theta),
\end{equation}

Combining the two half-steps, the QCA update \(\Delta_E \circ \mathcal{Q}\) under the system Pauli twirl induces the following two-step PCA kernel on configurations \(s = (r,b) \mapsto s'= (r',b')\):
\begin{enumerate}
    \item Update \(b \mapsto b'\) by flipping each \(j \in B\) independently with probability \(\sin^2(k_j(r) \theta)\);
    \item Update \(r \mapsto r'\) by flipping each \(i \in R\) independently with probability \(\sin^2(k_i(b') \theta)\).
\end{enumerate}
This is precisely the two-step PCA kernel defined in \cref{sec:pca-details}.

\subsection{Conditional system Pauli emission}\label{app:system-pauli-emission}
We now derive the induced system Pauli channel conditioned on the (updated) environment configuration. Consider a single site and suppress the site index. Let \(\widetilde{\mathcal{U}}^{SE}\) be the twirled interaction from (V), and define the reduced system channel
\begin{equation}
    \mathcal{F}[\rho_S] \coloneq \Tr_E(\widetilde{\mathcal{U}}^{SE}[\rho_E \otimes \rho_S]).
\end{equation}
From \cref{app:env-dephasing,app:pca-kernel}, the environment state \(\rho_E\) entering the effective interaction is always a mixture of computational basis states. Inserting the Kraus representation from \cref{eq:twirled-kraus-ops},
\begin{equation}
    \begin{split}
        \mathcal{F}[\rho_S]
        &= \sum_P \sum_{z,z'} \Tr_E(K_{P,z} (\rho_E \otimes \rho_S) K_{P,z'}^\dagger) \\
        &= \sum_P \sum_{z,z'} \frac{1}{4} \Tr_E((\Pi_z \otimes PV_zP) (\rho_E \otimes \rho_S) (\Pi_{z'} \otimes PV_{z'}P)^\dagger) \\
        &= \sum_{z,z'} \Tr(\Pi_z \rho_E \Pi_{z'}) \cdot \frac{1}{4} \sum_P P V_z P \rho_S P V_{z'}^\dagger P.
    \end{split}
\end{equation}
By the cyclicity of the trace and \(\Pi_z\Pi_{z'} = \delta_{z,z'} \Pi_z\),
\begin{equation}
    \mathcal{F}[\rho_S] = \sum_z \Tr(\Pi_z \rho_E) \cdot \frac{1}{4} \sum_P P V_z P \rho_S P V_z^\dagger P,
\end{equation}
with \(\sum_z \Tr(\Pi_z \rho_E) = 1\). As a result, \(\mathcal{F}\) is a convex combination of \emph{conditional} system channels,
\begin{equation}
    \mathcal{F}[\rho_S] = \sum_z p_z \mathcal{F}_z[\rho_S], \quad p_z \coloneq \Tr(\Pi_z \rho_E), \quad \mathcal{F}_z[\rho_S] \coloneq \frac{1}{4} \sum_P P V_z P \rho_S P V_z^\dagger P.
\end{equation}
For \(z=0, V_0 = I\), we immediately have
\begin{equation}
    \mathcal{F}_0[\rho_S] = \rho_S.
\end{equation}
For \(z=1, V_1 = \vec{n} \cdot \vec{\sigma}\), we observe that \(\mathcal{F}_1\) is exactly the Pauli twirl of a unitary channel with single Kraus operator \(V_1\). The Pauli twirl of a channel with Kraus operators \(\{K_i\}\) can be written as
\begin{equation}
    \mathcal{T}_P(\mathcal{E})[\rho] = \frac{1}{d^2} \sum_P \sum_i |\Tr(P K_i)|^2 P \rho P,
\end{equation}
where \(d\) is the system dimension. Substituting \(K_1 = V_1\) and \(d=2\) yields
\begin{equation}
    \mathcal{F}_1[\rho_S] 
    = \sum_P \frac{1}{4}|\Tr(PV_1)|^2 P \rho_S P 
    = \sum_P \frac{1}{4}|\Tr(P (n_X X + n_Y Y + n_Z Z))|^2 P \rho_S P.
\end{equation}
Using orthogonality of Pauli operators under the Hilbert-Schmidt inner product, we obtain the local system Pauli channel
\begin{equation}
    \mathcal{F}_1[\rho_S] = n_X^2 X \rho_S X + n_Y^2 Y \rho_S Y + n_Z^2 Z \rho_S Z,
\end{equation}
with \(n_X^2 + n_Y^2 + n_Z^2 = 1\) by \cref{eq:unitary-decomp}.

For the full lattice, \(U_{SE} = \bigotimes_{i\in V} U_{SE}^{(i)}\) and the twirl is applied locally, so conditional on the updated environment configuration \(s_{t+1} \in \{0,1\}^{|V|}\), the emitted Pauli string \(x_t \in \{I, X, Y, Z\}^{|V|}\) factorizes across sites with the above single-site channel.

\subsection{Summary of full mapping}\label{app:summary-mapping}
We now summarize the full mapping from the microscopic QCA noise model to the effective PCA HMM induced by the system Pauli twirl. We iterate through the PCA HMM steps as defined in \cref{sec:pca-details}, and justify each in reference to the assumptions and derivations throughout \cref{app:assumptions,app:env-dephasing,app:pca-kernel,app:system-pauli-emission}.

\begin{enumerate}
    \item[0.] \textbf{Environment initialization}. Assumed in (I), the initialization is identical in both models.

    \item \textbf{Storm update.}
    Assumed in (II). The storm channel \(\mathcal{S}^E\) is only on the environment and preserves a diagonal \(\rho_E\), so it is simulated exactly as a classical update on \(s_t\) and does not affect the subsequent reductions.

    \item \textbf{PCA half-step (black)}.
    By (II)–-(III), the microscopic model applies \(\mathcal{Q}_{R\to B}^E\). Under the twirl, \cref{app:env-dephasing} dephases the environment in the computational basis, and \cref{app:pca-kernel} shows \(\Delta_B\circ \mathcal{Q}_{R\to B}^E\) is exactly the conditional independent flip update on \(B\) with probability \(\sin^2(k\theta)\).

    \item \textbf{PCA half-step (red)}.
    Likewise, \cref{app:pca-kernel} shows \(\Delta_R\circ \mathcal{Q}_{B\to R}^E\) is exactly the conditional independent flip update on \(R\) with probability \(\sin^2(k\theta)\), conditioned on the updated black configuration.
    
    \item \textbf{Pauli emission}.
    By (IV)–-(V) and \cref{app:system-pauli-emission}, conditional on \(s_{t+1}(i)\) the emitted Pauli is \(I\) if \(s_{t+1}(i)=0\), and \(P\in\{X,Y,Z\}\) with probabilities \(n_P^2\) if \(s_{t+1}(i)=1\), independently across sites.
    
    \item \textbf{Repeat}. 
    By (II) and \cref{app:env-dephasing}, tracing out the system after the twirled interaction applies \(\Delta_E\), so the environment entering the next cycle is again diagonal; hence the above steps iterate cycle-by-cycle.
\end{enumerate}
This completes the mapping from the microscopic QCA model in \cref{sec:qca-details} to the effective PCA HMM in \cref{sec:pca-details}.

\end{document}